\documentclass[11pt]{amsart}
\usepackage{graphicx} % Required for inserting images
\usepackage{amsmath,amsfonts,amsthm,amssymb,amscd}
\usepackage{lipsum}
\usepackage{xcolor}
\usepackage{amsfonts}
\usepackage{graphicx}
\usepackage{epstopdf}
\usepackage{algorithm}
\usepackage{algorithmic}
\usepackage{upgreek}
\usepackage{overpic}
\usepackage{enumerate}
\usepackage{enumitem}
\usepackage{bm}
\usepackage[foot]{amsaddr}
\usepackage{mathtools}
\usepackage{intcalc}
\usepackage{nth}
\usepackage[hidelinks]{hyperref}
\usepackage[letterpaper,top=3.5cm,bottom=3.5cm,left=3.5cm,right=3.5cm,marginparwidth=1.75cm]{geometry}

\newcommand{\bE}{\mathbb{E}}

\newcommand{\bR}{\mathbb{R}}

\newcommand{\sfA}{\mathsf{A}}
\newcommand{\sfN}{\mathsf{N}}

\newcommand{\sfPi}{\mathsf{\Pi}}
\newcommand{\sfU}{\mathsf{U}}
\newcommand{\sfPhi}{\mathsf{\Phi}}

\newcommand{\baoab}{\text{\normalfont\scshape baoab}}
\newcommand{\abao}{\text{\normalfont\scshape abao}}
\newcommand{\bao}{\text{\normalfont\scshape bao}}
\newcommand{\ab}{\text{\normalfont\scshape ab}}
\newcommand{\hoh}{\text{\normalfont\scshape hoh}}
\newcommand{\hmc}{\text{\normalfont\scshape hmc}}
\newcommand{\ho}{\text{\normalfont\scshape ho}}
\newcommand{\sfB}{\mathsf{B}}
\newcommand{\sfO}{\mathsf{O}}
\newcommand{\sfQ}{\mathsf{Q}}
\newcommand{\sfP}{\mathsf{P}}
\newcommand{\sfH}{\mathsf{H}}

\newtheorem{theorem}{Theorem}[section]
\newtheorem{lemma}[theorem]{Lemma}
\newtheorem{assumption}[theorem]{Assumption}
\newtheorem{proposition}[theorem]{Proposition}
\newtheorem{definition}[theorem]{Definition}
\newtheorem{example}[theorem]{Example}

\newtheorem{corollary}[theorem]{Corollary}
\theoremstyle{remark}
\newtheorem{remark}[theorem]{Remark}
\numberwithin{equation}{section}

\newcommand{\xqed}[1]{%
    \leavevmode\unskip\penalty9999 \hbox{}\nobreak\hfill
    \quad\hbox{\ensuremath{#1}}}
\newcommand{\Endofdef}{\xqed{\lozenge}}

\usepackage{tikz}
\usepackage{mathdots}
\usepackage{yhmath}
\usepackage{cancel}
\usepackage{color}
\usepackage{siunitx}
\usepackage{array}
\usepackage{multirow}
\usepackage{tabularx}
\usepackage{extarrows}
\usepackage{booktabs}
\usetikzlibrary{fadings}
\usetikzlibrary{patterns}
\usetikzlibrary{shadows.blur}
\usetikzlibrary{shapes}

\makeatletter
\renewcommand\paragraph{\@startsection{paragraph}{4}{\z@}%
  {3.25ex \@plus1ex \@minus.2ex}% vertical space before
  {-1em}%                        % horizontal space after (negative = run-in)
  {\normalfont\bfseries}}        % The style (Bold)
\makeatother

\definecolor{darkred}{rgb}{.7,0,0}

\definecolor{darkgreen}{rgb}{.15,.55,0}

\definecolor{darkblue}{rgb}{0,0,0.7}

\title[Delocalization of bias in unadjusted HMC and underdamped Langevin]{Delocalization of bias in unadjusted Hamiltonian Monte Carlo and underdamped Langevin }
\author{Yifan~Chen\textsuperscript{\dag}}
\author{Xiaoou~Cheng\textsuperscript{\ddag}}
\author{Jonathan Niles-Weed\textsuperscript{\ddag}}
\author{Jonathan Weare\textsuperscript{\ddag}}
\address{\textsuperscript{\dag}Department of Mathematics, University of California, Los Angeles, CA 90095, USA}
\address{\textsuperscript{\ddag}Courant Institute, New York University, NY 10012, USA}
\email{yifanchen@math.ucla.edu, chengxo@nyu.edu, jnw@cims.nyu.edu, weare@nyu.edu}

\begin{document}
\begin{abstract}
   Unadjusted samplers such as unadjusted Hamiltonian Monte Carlo  and underdamped Langevin  are well-known to be biased. Metropolis--Hastings adjustment has been conventionally incorporated into Hamiltonian Monte Carlo to eliminate the bias. However, this adjustment can significantly increase the iteration complexity due to the small step size required for reasonable Metropolis acceptance rates. In this work, we extend the
\emph{delocalization of bias} phenomenon, previously established for the overdamped
Langevin algorithm, to these two unadjusted algorithms. We show that to control the $W_2$ bias of any $K$-dimensional marginal of a high-dimensional distribution,  $O(\sqrt{K})$ integration steps suffice up to  $\log d$ terms, assuming either weak or sparse interactions among variables.  The discrete-time integrators here introduce technical difficulties beyond those of the overdamped setting, which we address through a broadly applicable matrix-polynomial framework that characterizes their propagators.  Our result for the underdamped Langevin algorithm is valid for all large friction parameters, implying that the Leimkuhler-Matthews integrator for the overdamped Langevin dynamics also exhibits delocalization of bias.
\end{abstract}
\maketitle

\section{Introduction}

Hamiltonian (or Hybrid) Monte Carlo (HMC) and underdamped Langevin (UL)\footnote{Note that the unadjusted Langevin algorithm (ULA) is instead a discretization of the \emph{overdamped} Langevin dynamics.} are widely used to sample from high dimensional probability distributions. For a target distribution $\pi \propto \exp(-V)$, where $V: \mathbb{R}^d \to \mathbb{R}$ denotes the potential energy, both sampling dynamics evolve in terms of a position variable $q$, on which the potential $V$ is defined, and an auxiliary momentum variable $p$. HMC relies on the Hamiltonian ordinary differential equation 
\begin{equation}
\label{eqn-continuous-time-Hamiltonian}
\begin{aligned}
&\frac{{\rm d} q_t}{{\rm d}t} = p_t \,,\\
& \frac{{\rm d}p_t}{{\rm d}t} = -\nabla V(q_t)\,.
\end{aligned}
\end{equation}
corresponding to the
Hamiltonian $\mathcal{H}(q,p) = V(q) + |p|_{\ell^2}^2/2$, where $|\cdot|_{\ell^2}$ is the $\ell^{2}$ norm of a vector. 
The closely related continuous-time UL is described by the stochastic differential equation
\begin{equation}
\label{eqn-continuous-time-underdamped}   
\begin{aligned}
&{\rm d} q_t = p_t {\rm d}t\,,\\
&{\rm d} p_t = -\nabla V(q_t) {\rm d}t-\gamma p_t{\rm d}t+\sqrt{2\gamma}{\rm d}B_t\,.
\end{aligned}
\end{equation}
where $\gamma > 0$ is the friction parameter and $B_t$ is a standard Brownian
motion. Both continuous-time dynamics admit $\pi \otimes \mathcal{N}(0, I)$ as their invariant measure on  phase space $(q,p)$, allowing samplers built upon them to target $\pi$ via the $q$-marginal.

Applying these continuous-time dynamics numerically requires time discretization. For HMC, one must also refresh the momentum regularly to ensure sampling ergodicity. The most commonly used discretization of \eqref{eqn-continuous-time-Hamiltonian} 
%is the leap-frog scheme, also called velocity Verlet, which alternates between position and momentum updates:
is the leap-frog scheme (also called velocity Verlet)~\cite{leimkuhler2015moleculardynamics}, under which the $(k+1)$-th step of HMC is
\begin{enumerate}
\item Sample $p_0  \sim \mathcal{N}(0, I)$ and set $q_0 = X_k$.
\item Inner loop (leap-frog integration): for $j = 0$ to $m-1$,
\begin{equation}
\label{eqn-leap-frog}
\begin{aligned}
q_{j+1} &= q_j + h p_j  - \frac{h^2}{2}\nabla V(q_j)\,,\\
p_{j+1} &= p_j - \frac{h}{2}\nabla V(q_j) - \frac{h}{2}\nabla V(q_{j+1})\,.
\end{aligned}
\end{equation}
Set $X_{k+1} = q_m$.
\end{enumerate}
Here, $h$ is the step size, and we also refer to one HMC iteration as one outer loop. For \eqref{eqn-continuous-time-underdamped}, a commonly used discretization is the BAOAB scheme~\cite{leimkuhler2015moleculardynamics}, which is written as 
\begin{equation}
\label{eqn-baoab}
\begin{aligned}
    q_{k+1} &= q_k + \frac{h}{2}(1+\eta)p_k - \frac{h^2}{4}(1+\eta)\nabla V(q_k) + \frac{h}{2}\sqrt{1-\eta^2}\xi_{k}\,,\\
    p_{k+1} &= \eta(p_k - \frac{h}{2}\nabla V(q_k)) + \sqrt{1-\eta^2}\xi_{k}-\frac{h}{2}\nabla V(q_{k+1})\,.
\end{aligned}
\end{equation}
The scheme \eqref{eqn-baoab} is constructed through an operator splitting. We denote the solution maps for the three components of \eqref{eqn-continuous-time-underdamped} over a step size $h$ as follows:
\begin{equation}
\label{eqn-def-bao}
\begin{aligned}
&{\mathsf{B}_h}(q, p) = (q, p-h\nabla V(q))\,,\\
&{\mathsf{A}_h}(q, p) = (q + hp, p)\,,\\
&{\mathsf{O}_h}(q, p; \xi) = (q, \eta p + \sqrt{1-\eta^2}\xi)\,.
\end{aligned}
\end{equation}
Here, $\xi \sim \mathcal{N}(0, I)$ is an independent Gaussian vector, and $\eta = \exp(-\gamma h)$. Splitting methods for underdamped Langevin dynamics are defined with different strings of A, B, and O. By standard convention, operators employed symmetrically in the string are applied with a half step $h/2$, while those that appear only once are applied with a full step $h$. Specifically, the BAOAB scheme \eqref{eqn-baoab} is derived from the composition $\mathsf{B}_{h/2} \circ \mathsf{A}_{h/2} \circ \mathsf{O}_h (\cdot;\xi) \circ \mathsf{A}_{h/2} \circ \mathsf{B}_{h/2}$, which we abbreviate as $\mathsf{B}\mathsf{A}\mathsf{O}\mathsf{A}\mathsf{B}$. We use sans-serif $\mathsf{A}$, $\mathsf{B}$, $\mathsf{O}$ and their compositions for the solution maps, while the same letters in roman type identify the corresponding scheme. For the frictionless case $\gamma = 0$, we have $\eta = 1$, so the $\mathsf{O}_h$ step becomes the identity, 
and the BAOAB scheme \eqref{eqn-baoab} reduces to the leap-frog scheme \eqref{eqn-leap-frog}.

\subsection{Motivation} The unadjusted HMC and UL algorithms are biased. Because of the finite step size
$h$, when the marginal $\rho_k$ of $q_k$ converges, its limit $\pi_h$ differs from
the target $\pi$. Both HMC and UL can be modified with a Metropolis--Hastings (MH) accept/reject step to eliminate bias. In fact, HMC is most frequently applied with MH adjustment~\cite{neal2011mcmc, carpenter2017stanprobabilistic, salvatier2016probabilisticprogramming}. However, adjusted HMC and UL are well known to degrade as $d$ increases and the step size $h$ must be decreased to maintain a non-zero acceptance probability.

In \cite{chen2024convergence}, we have identified a ``delocalization of bias'' phenomenon for the unadjusted \emph{overdamped} Langevin algorithm: To control the bias of low-dimensional marginals, a step size scaling with the low dimension suffices, with only a logarithmic dependence on the full dimension $d$.
The analysis in \cite{chen2024convergence} relies on a  novel $W_{2,\ell^\infty}$ metric, which captures the accuracy of low-dimensional marginals. For two probability measures $\mu$ and $\nu$ on $\mathbb{R}^d$, their $W_{2,\ell^\infty}$ distance is defined as
\begin{equation}
        W_{2,\ell^{\infty}}(\mu, \nu) = \left( \min_{\gamma \in \Pi(\mu,\nu)} \int |x-y|^2_{\ell^\infty} \gamma({\rm d}x,{\rm d}y)\right)^{1/2}\, ,
    \end{equation}
where 
$|\cdot|_{\ell^\infty}$ is the $\ell^{\infty}$ norm of a vector and 
$\Pi(\mu,\nu)$ represents the set of measures on the joint space $\mathbb{R}^d\times\mathbb{R}^d$ that have marginals $\mu$ and $\nu$. We note that, for any choice of $K$ coordinates, the corresponding marginal distributions $\mu^{(K)}$ and $\nu^{(K)}$ of $\mu$ and $\nu$, satisfy the upper bound $W_2(\mu^{(K)}, \nu^{(K)}) \leq \sqrt{K} W_{2,\ell^\infty}(\mu, \nu)$, since $K |x-y|_{\ell^\infty}^2\geq \sum_{t=1}^K|x^{(j_t)} - y^{(j_t)}|^2$ for any $1\leq j_t \leq d$. In particular, $W_{2,\ell^\infty}(\mu, \nu)$ upper bounds the distance between any one-dimensional marginals. 

Our analysis for unadjusted overdamped Langevin in \cite{chen2024convergence} proceeds via a coupling argument under the $W_{2,\ell^\infty}$ distance. In this paper, we extend that analysis to unadjusted HMC and UL, which are both much more commonly used for complicated, high-dimensional sampling applications than overdamped Langevin. 
This
extension introduces significant technical challenges, as the dynamics now evolve on the full phase space $(q,p)$ with an auxiliary momentum variable. The coupling framework in \cite{chen2024convergence} requires analyzing the discrete dynamics over multiple steps to bound the accumulation of discretization error. 
However, the propagator matrices on the $(q,p)$ space exhibit intricate position-momentum coupling, making their multi-step composition difficult to control.
To this end, we introduce a novel matrix polynomial representation of the propagators. Such a representation allows for precise characterization of the multi-step dynamics and enables sharp error bounds under the $\ell^\infty$ norm.

\subsection{Literature review}
Hamiltonian and underdamped Langevin dynamics have long been foundational to the simulation of physical systems. 
HMC was originally introduced in \cite{duane1987hybrid} for particle physics, and was later adapted and popularized for modern statistics and machine learning by \cite{neal2011mcmc}. Today, it serves as one of the standard samplers in probabilistic programming packages like Stan~\cite{carpenter2017stanprobabilistic} and PyMC~\cite{salvatier2016probabilisticprogramming}. In the realm of molecular dynamics (MD), unadjusted HMC is closely related to the Andersen
thermostat~\cite{andersen1980molecular}.
Similarly, UL is commonly employed as a robust thermostat~\cite{frenkel2023understanding,allen2017computer,leimkuhler2015moleculardynamics}. The BAOAB splitting scheme was introduced in~\cite{leimkuhlerrationalconstruction} and variants have since been implemented widely in MD software packages~\cite{openmm7,abraham2015gromacs,kieninger2022gromacs,case2025recent,rackers2018tinker,phillips2020scalable,thompson2022lammps}.

To maintain stability in integration of Hamilton's ODE~\eqref{eqn-continuous-time-Hamiltonian} or the underdamped Langevin SDE~\eqref{eqn-continuous-time-underdamped},  the step size $h$ must be chosen smaller than the shortest vibrational period among all processes in the system (e.g.\ about a femtosecond in MD). For adjusted HMC, by contrast, \cite{beskos2010optimaltuning} establishes the scaling $h \sim d^{-1/4}$ as the requirement for a non-vanishing acceptance probability as $d \to \infty$, under a product measure with bounded fourth-order derivatives of the potential.
The same scaling was later shown for an adjusted UL scheme in \cite{riou2022metropolis}, and also appears in mixing time guarantees for adjusted HMC under warm starts~\cite{chen2023whendoesa}. Recent work~\cite{zhang2026algorithmicwarm} shows that unadjusted HMC generates a warm start in $O(d^{1/4})$ steps up to logarithmic factors.

Unadjusted algorithms avoid the rejection bottleneck but incur a discretization bias. A rich body of literature has emerged focusing on the non-asymptotic analysis of unadjusted HMC and UL algorithms, establishing bias bounds as the number of iterations goes to infinity. 
These analyses rely on various distance metrics and target measure assumptions. The most fundamental setting, detailed in Assumption \ref{assumption-V-log-concave}, requires the potential to be a $C^2$ function that is strongly convex with Lipschitz gradient. However, because commonly employed integrators possess weak second-order accuracy~\cite{leimkuhler2013robustefficient, bou-rabee2025unadjustedhamiltonian}, fully demonstrating their accuracy  often requires bounding the higher-order derivatives of the potential~\cite{monmarche2021highdimensionalmcmc}. Under these frameworks, the theoretical guarantees show power laws $h \sim d^{-c}$ for some order $c >0$ depending on the specific assumptions. For unadjusted HMC, \cite{mangoubi2019mixinghamiltonian} requires $h \sim d^{-1/2}$ to bound the Wasserstein distance, while \cite{mangoubi2018dimensionallytight} achieves an $h \sim d^{-1/4}$ scaling for a high-probability bound under a third-order regularity condition. Furthermore, bias bounds arising from convergence guarantees have been derived in total variation distance ($h \sim d^{-3/4}$)~\cite{bou-rabee2023mixingtime}, Kullback-Leibler (KL) divergence ($h \sim d^{-3/4}$)~\cite{bou-rabee2026tailsensitivekl}, and R\'enyi divergence ($h \sim d^{-3/2}$)~\cite{bou-rabee2026tailsensitivekl}. For unadjusted UL, \cite{leimkuhler2024contraction} establishes the Wasserstein bias for the BAOAB scheme, requiring $h \sim d^{-1/2}$ under Assumption \ref{assumption-V-log-concave}, with improved dimensional dependence achieved under stricter smoothness assumptions. Wasserstein bias for different integrators of UL was studied in \cite{cheng2018underdampedlangevin, dalalyan2018samplinglogconcave, monmarche2021highdimensionalmcmc}. \cite{monmarche2021highdimensionalmcmc} also establishes a total variation bias bound through a Wasserstein-total variation regularization property. Under KL divergence, \cite{ma2019thereanalog} similarly requires $h \sim d^{-1/2}$. These analyses have also been generalized to non-convex settings using techniques such as reflection coupling~\cite{eberle2018couplingsquantitative, schuh2024convergencekinetic, chak2024reflectioncoupling} and entropic approaches~\cite{monmarche2023entropicapproach, camrud2024secondorder}.

Recent literature has increasingly unified the theoretical analysis of HMC and UL algorithms. For instance, under the umbrella framework of ``generalized HMC'' (gHMC), \cite{gouraud2024hmcunderdamped} determines the optimal scaling of the friction parameter and the HMC integration time to attain accelerated convergence regarding the condition number for Gaussian targets, while \cite{camrud2024secondorder} develops KL convergence guarantees for gHMC. Moving beyond traditional fixed-step integrators, \cite{shen2019randomizedmidpoint} introduces the Randomized Midpoint Method (RMM) for UL dynamics. They establish an improved dimension dependence $h \sim d^{-1/3}$ in the Wasserstein distance. Inspired by this mechanism, \cite{bou-rabee2025unadjustedhamiltonian} achieves the same $d^{-1/3}$ scaling for unadjusted HMC by employing a stratified Monte Carlo integrator. 

Crucially, the aforementioned guarantees quantify the error across all variables of the joint system. In contrast, our work builds upon \cite{chen2024convergence,lacker2025hierarchical} to focus exclusively on controlling the error of low-dimensional marginals, which is a more practical notion of error in high-dimensional applications where the majority of coordinates act as nuisance variables. We previously introduced the term ``delocalization of bias'' in \cite{chen2024convergence} to describe the phenomenon where the $W_2$ bias of low-dimensional marginals remains nearly independent of the full system dimension. Our analysis there focused on the unadjusted (overdamped) Langevin algorithm, employing a coupling framework to establish Wasserstein bias bounds for the marginals. \cite{lacker2025hierarchical} investigates this phenomenon under KL divergence using a hierarchical analysis of marginal relative entropies. In \cite{cui2026steinsmethod}, a related question is studied for graphical models,
where  dimension-independent bounds are established, through Stein's method, on the marginal error between a
target distribution and its localized approximation. 
 Prior to our work in \cite{chen2024convergence}, \cite{bou-rabee203meanfield} and \cite{durmus2024asymptotic} observed that large step sizes remain viable in high dimensions if the observable of interest has a Lipschitz constant that decays with dimension, thereby canceling out the dimensional scaling of the full system's Wasserstein bias (e.g., an averaged observable $f(x) = \frac{1}{d} \sum_{i=1}^d \phi(x^{(i)})$ with an $\ell^2$-Lipschitz $\phi$). In such regimes, a constant step size is sufficient to achieve a bounded error regardless of dimension. Besides these analysis efforts, there are ample numerical evidence demonstrating the viability of dimension-free step size choices~\cite{chada2025unbiasedkinetic, robnik2025blackboxunadjusted}.

\subsection{Main results}
In \cite{chen2024convergence}, we have shown for overdamped Langevin that the delocalization of bias holds for product measures, for Gaussian measures, and for distributions with sparse interactions between variables. We also provided a negative example showing that this effect does not hold universally. In this paper, we establish the delocalization effect for the unadjusted HMC and
UL algorithms under analogous assumptions, namely weak or sparse interactions, or
Gaussian targets. Our analysis is conducted under the following strongly log-concave and log-smooth assumption.
\begin{assumption}
 \label{assumption-V-log-concave}
     Let $\pi \propto \exp(-V)$ with $V \in C^2(\mathbb{R}^d)$. Assume $V$ is $\alpha$-strongly convex and $\beta$-smooth such that $\alpha I \preceq \nabla^2 V(x) \preceq \beta I$ for any $x \in \mathbb{R}^d$, where $0 < \alpha \leq \beta < \infty$. 
 \end{assumption}

 Throughout, unadjusted HMC uses the leap-frog integrator~\eqref{eqn-leap-frog} and
unadjusted UL uses the BAOAB scheme~\eqref{eqn-baoab}. As a warm-up, we present the delocalization of bias result for Gaussian distributions. Since the BAOAB scheme is unbiased for any Gaussian distribution~\cite{leimkuhler2013robustefficient}, only unadjusted HMC is discussed in this example. We employ the explicit formula for the biased Gaussian distribution and obtain a $W_{2,\ell^\infty}$ bias bound directly. The proof appears in Appendix~\ref{appendix-gaussian}.

\begin{example}[Gaussian distributions]
\label{thm-Gaussian-hmc-main-text}    
    Consider $\pi =\mathcal{N}(\mu, \Sigma)$ where $\mu \in \bR^d$ and $\alpha I\preceq \Sigma^{-1} \preceq \beta I$. Then, for $h < 1 / \sqrt{\beta}$, unadjusted HMC, when it converges, has the following bias: 
    \begin{equation}
        W_{2,\ell^\infty}(\pi_h, \pi) = O(h^2 \sqrt{\beta \log (2d)}) = O(h\sqrt{\log (2d)})\,.
    \end{equation}
\end{example}
 As in the case of overdamped Langevin dynamics \cite{chen2024convergence}, the proof of Example~\ref{thm-Gaussian-hmc-main-text} offers no insight into which target distributions beyond the Gaussian family evince the delocalization of bias effect for unadjusted HMC or UL. Our main results describe two different sets of structural assumptions on $V$ under which this effect appears.
 For general distributions satisfying Assumption~\ref{assumption-V-log-concave}, our first results address the ``weak interactions'' setting of~\cite{lacker2025hierarchical}, made precise below.
\begin{assumption}[weak interactions]
\label{assumption-weak-interactions}
The Hessian $\nabla^2V$ admits the decomposition $\nabla^2 V = \nabla^2 V^{(D)}+ \nabla^2 V^{(O)}$ 
where $\nabla^2 V^{(D)} = \mathrm{diag}(\nabla^2 V)$ is the diagonal part and $\nabla^2 V^{(O)}$ is the off-diagonal part. We assume there exists a sufficiently small constant $C^{(O)} > 0$, independent of $\alpha$, $\beta$, $h$, and $d$, such that the the off-diagonal part satisfies 
\[|\nabla^2 V^{(O)}|_{\ell^\infty} \leq C^{(O)} \alpha
\] uniformly over $\mathbb{R}^d$.
\end{assumption}
Assumption \ref{assumption-V-log-concave} ensures that every diagonal entry satisfies $[\nabla^2 V^{(D)}]_{ii} \geq \alpha$. The interactions defined in Assumption \ref{assumption-weak-interactions} are therefore weak, as the total intensity of interactions between any given variable and all other variables, measured by off-diagonal row sums, is small compared to the independent components on the diagonal.

We now state our main result under weak interactions. The proofs for the two
algorithms are given in Appendix~\ref{appendix-sampling-bias-bounds}.
\begin{theorem}[weak interactions]
    \label{thm-weak-formal}
    Let Assumptions \ref{assumption-V-log-concave} and \ref{assumption-weak-interactions} hold. For either unadjusted HMC  with $mh = 1/\sqrt{20\beta}$ and $C^{(O)} = 1/50$, or unadjusted UL  with $h \leq \frac{1-\eta}{2\sqrt{\beta}}$ and $C^{(O)} = 1/20$, the resulting invariant distribution $\pi_h$ satisfies
    \begin{equation}
    \label{eqn-weak-informal}
    W_{2,\ell^\infty}(\pi_h, \pi) = O\left(\frac{\beta}{\alpha}h\sqrt{\log(2d)}\right)\,.
    \end{equation}
 \end{theorem}
 Essentially, the distribution with weak interactions constitutes a perturbation of the product measure. The delocalization effect holds for product measures, since the dynamics are separable for each variable, leading to a dimension-free bias in each coordinate.  This effect persists as long as the interactions between variables are sufficiently weak. 

The second setting of interest is the case of ``sparse interactions'' studied in~\cite{chen2024convergence}, where each variable interacts directly with only a few others. Because the interactions are local, the influence of one variable spreads only gradually through the interaction structure over multiple steps, keeping the marginal bias small. This interaction structure is described by a graphical model, with each variable a node, as illustrated in Figure~\ref{fig:sparse-potential}. Let $G$ be an undirected graph with $d$ nodes, labeled by $1\leq i \leq d$, and write $i \sim j$ when nodes $i$
and $j$ are connected by an edge. The neighborhood $\sfN(i)$ is the set of nodes
connected to $i$. Here the neighborhood relationship is symmetric, and without loss of generality, we take $i \sim i$ for all $1\leq i \leq d$. We quantify the locality of interactions through a growth condition on the
neighborhood sizes. Define recursively $\sfN_k(i) = \{1\leq j \leq d: \exists\ l \in \sfN_{k-1}(i), \text{ such that } j \sim l \}$, for $k \geq 2$, with $\sfN_1(i) := \sfN(i)$, and write  $s_k \coloneqq \max_{1 \leq i \leq d} |\sfN_{2k}(i)|$ for the relevant
neighborhood cardinality where $s_k \geq 1$ by definition. We make the following assumption on the potential.
\begin{assumption}

\label{assumption-sparse-interactions}

The potential has the form \begin{equation}
\label{eqn-sparse-potential-function}
    V(x) = \sum_{i=1}^d V_i (\mathcal{X}^i)\,,
\end{equation} 
where $\mathcal{X}^i = \{x^{(j)}: j \in \sfN(i)\}$ and each $V_i$ depends only on the variables in $\mathcal{X}^i$. Moreover, the sparsity parameter $s_k$ satisfies
the polynomial growth bound $s_k \leq C(k+1)^n$ for constants $C, n > 0$.
\end{assumption}
The main result under sparse interactions is as follows. The proofs are given in
Appendix~\ref{appendix-sampling-bias-bounds}.
 \begin{figure}[!htbp]
    \centering
  \includegraphics[width=0.8\linewidth]{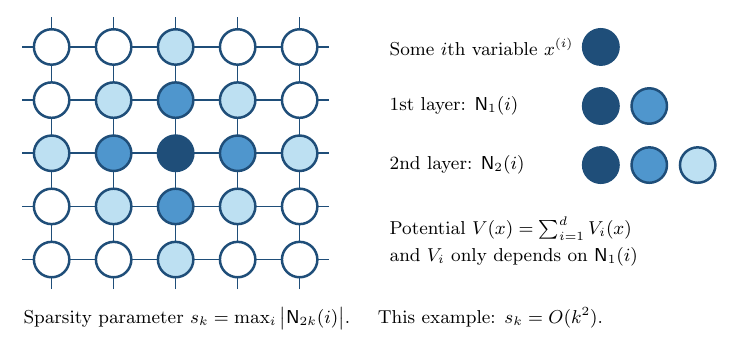}
      \caption{Illustration of a potential $V(x)$ with sparse interactions. Adapted from \cite{chen2024convergence}.}
    \label{fig:sparse-potential}
\end{figure}
 \begin{theorem}[sparse interactions]
     \label{thm-sparse-formal}
     Let Assumptions~\ref{assumption-V-log-concave} and~\ref{assumption-sparse-interactions} hold. For either unadjusted HMC with $mh = 1/\sqrt{20\beta}$,  or unadjusted UL with $h \leq \frac{1-\eta}{2\sqrt{\beta}}$, the resulting invariant distribution $\pi_h$ satisfies
     \begin{equation}
        \label{eqn-sparse-informal}
        W_{2,\ell^\infty}(\pi_h, \pi) = h\sqrt{\log(2d)}\left( O\left( \frac\beta\alpha\log(2d)\right)\right)^{\frac n 2 + 1}\,.
     \end{equation}
 \end{theorem}
 
 Compared to our results for the unadjusted Langevin algorithm in \cite{chen2024convergence}, we have an improvement from $\sqrt{h\log(2d)}$ to $h\sqrt{\log(2d)}$. The new scaling aligns with \cite{leimkuhler2024contraction}, where the Wasserstein-2 bias of the BAOAB scheme is $h\sqrt{d}$. It does not show a second order accuracy due to lack of higher order regularity conditions on $V$~\cite{monmarche2021highdimensionalmcmc, bou-rabee2025unadjustedhamiltonian}. We also note that the step-size conditions for both methods align with existing results. For HMC, the chosen length of the integration time scales with $1/\sqrt{\beta}$ for both the continuous and discrete  dynamics~\cite{chen2019optimalconvergence, bou-rabee203meanfield}, while we choose the specific constant to optimize the contraction rate in our analysis. For UL, we allow for the same step-size condition as in~\cite{leimkuhler2024contraction}. Specifically, the condition indicates that $h \leq \frac{1}{2\sqrt{\beta}}$. 
 Besides, the relation $\eta = \exp(-\gamma h)$ implies the condition $\gamma > 2\sqrt\beta$. In comparison, \cite{monmarche2023almostsure} establishes that $\gamma > \sqrt{\beta}-\sqrt{\alpha}$ is the sharp condition to guarantee Wasserstein contraction for general distributions, whereas for Gaussian targets, contraction holds for any $\gamma$. Therefore, in the highly ill-conditioned regime ($\alpha \ll \beta$), our friction requirement matches the scaling of this tight condition up to a factor of $2$. This scaling also aligns with continuous-time dynamics, where the optimal Wasserstein-2 convergence rate of $O(\alpha/\gamma)$ is attained at $\gamma =\sqrt{\beta+\alpha}$, a value that is asymptotically equivalent to $\sqrt{\beta}$ when $\alpha \ll \beta$.
 Notably, the step-size condition accommodates the entire high friction regime, and ensures that the bias bounds remain valid in the overdamped limit $\gamma \to \infty$, where the BAOAB scheme reduces to the Leimkuhler-Matthews (LM) integrator, a method for the overdamped dynamics that achieves weak second-order accuracy~\cite{leimkuhlerrationalconstruction}. Our results thus also imply the delocalization of bias effect  for the LM scheme.
 
 The key technical challenge of extending our analysis in \cite{chen2024convergence} is that the propagator matrices of HMC and UL are more involved. The coupling argument, a common device for obtaining Wasserstein bounds, relies on the propagator of the dynamics to induce contraction and to control the accumulation of discretization error. The lack of one-step contraction in the $W_{2,\ell^\infty}$ metric necessitated the use of a multi-step coupling argument in \cite{chen2024convergence}, where the sparsity of the potential was leveraged to control the accumulated errors over these steps. 
 In this paper, we employ a novel matrix polynomial framework that applies to both HMC and UL, to characterize the multi-step propagator matrices
 on the joint $(q,p)$ space. The framework takes a global-in-time view, giving the full multi-step propagator
at once as a polynomial of the Hessians along the trajectory, with coefficients we
can characterize and bound. On the joint $(q,p)$ space, position and momentum differences mix as the trajectory evolves, so the position-difference propagator is not a simple product of one-step maps but a structured polynomial pairing the two. The matrix-polynomial framework makes this structure explicit. The polynomial's coefficients control both the contraction and the error accumulation in the $\ell^\infty$ norm. Under sparse interactions, in particular, the
$\ell^\infty$ norm of the multi-step propagator is magnified over its $\ell^2$
norm by only a factor governed by the potential's sparsity.  In the easier weak interaction regime, the propagators are nearly diagonal, so contraction
follows from a more elementary, near-single-step argument. For HMC this
contraction is still read off from the multi-step polynomial, since one outer loop
comprises several integration steps, whereas for UL it follows directly from the
splitting structure of the scheme without invoking the polynomial.

 \subsection{Organization of this paper} In Section~\ref{sec-sketch-techniques}, we present our analysis framework based on the coupling argument and introduce the notation for the dynamics. Section~\ref{sec-matrix-polynomial} introduces a matrix polynomial framework to characterize the propagator matrices, and Section~\ref{sec-propagator-bounds} establishes bounds for these propagators. Discretization error estimates  are provided in Section~\ref{sec-discretization-error-main-text}. We conclude the paper in Section~\ref{sec-conclusion}. All proofs are deferred to the appendices.
 \subsection{Notation} For a random vector $X$, we define $|X|_{2,\ell^\infty} = (\bE[|X|_{\ell^\infty}^2])^{1/2}$. We write $A = O(B)$ or $A\lesssim B$ to denote that there exists a constant $C$
independent of $\alpha, \beta, h, d$ such that $A \leq CB$. 
We use $|\cdot|_{\ell^\infty}$ and $|\cdot|_{\ell^2}$ to represent the $\ell^{\infty}$ and $\ell^2$ norms for vectors and matrices. When applied to matrices, they stand for the corresponding operator norms. We use $\preceq$ for the Loewner order such that if $M\preceq N$ where $M, N$ are symmetric
matrices, then the matrix $N-M$ is positive semi-definite. We use $\lfloor x \rfloor$ and $\lceil x \rceil$ to denote the greatest integer less than or equal to $x$, and the least integer greater than or equal to $x$, respectively.

 \section{Sketch of Techniques}
 \label{sec-sketch-techniques}
 In this section, we outline our framework of the coupling argument and identify the propagator matrices as the key components for analysis. %, and provide bounds for propagators in HMC and underdamped Langevin.
\subsection{Wasserstein bias}
\label{sec-sub-Wasserstein-bias}
Our proof utilizes a coupling argument to establish both the Wasserstein convergence rate and the accumulated discretization error across iterations. These results then yield the final Wasserstein bias bounds. 
The following proposition outlines the Wasserstein convergence for HMC when the potential describes weak interactions.

\begin{proposition}[HMC, weak interactions]
\label{prop-sketch-weak-bound-hmc}
Suppose Assumptions \ref{assumption-V-log-concave} and \ref{assumption-weak-interactions} hold with $C^{(O)} = 1/50$. Let $\rho_k$ be the marginal distribution of $X_k$, the $k$-th step in the HMC chain. For  $mh = 1/\sqrt{20\beta}$, we have
\begin{equation}
\label{eqn-sketch-weak-bound-hmc}
W_{2,\ell^\infty}(\rho_{k+1}, \pi) \leq \left(1 - \frac{\alpha}{400\beta}\right)W_{2,\ell^\infty}(\rho_k, \pi) + O(h\sqrt{\log(2d)})\,.
\end{equation}
\end{proposition}
By taking $k \to \infty$, the iterative inequality \eqref{eqn-sketch-weak-bound-hmc} establishes the final bias bound in Theorem \ref{thm-weak-formal}. The inequality consists of two components. The first characterizes the Wasserstein contraction rate of the HMC chain, while the second bounds the discretization error relative to an ideal HMC process that preserves $\pi$ exactly. The detailed proof is in Appendix~\ref{appendix-sampling-bias-bounds}. For the case of sparse interactions, we instead rely on a multi-step recursive inequality.
\begin{proposition}[HMC, sparse interactions]
\label{prop-sketch-sparse-bound-hmc}
    Suppose Assumptions \ref{assumption-V-log-concave} and \ref{assumption-sparse-interactions}  hold,  and set
$r_i \coloneqq \lceil ih\sqrt\beta e + \frac{\log\sqrt d}{\log(5/3)}\rceil$, a
neighborhood threshold chosen to control the $\ell^\infty$ norm of the
propagators as in Section~\ref{sec-bound-sparse}.
When $mh = 1/\sqrt{20\beta}$, 
    we have
    \begin{equation}
    \label{eqn-sketch-sparse-bound-hmc}
        W_{2,\ell^\infty}(\rho_{k+N}, \pi) \leq \left( 1 - \frac{\alpha}{200\beta}\right)^N \sqrt{d}W_{2,\ell^\infty}(\rho_k, \pi) + O\left(\left(\sum_{i=0}^{Nm-1}\sqrt{s_{r_i}}\right)\sqrt{\beta}h^2\sqrt{\log(2d)}\right)\,.
    \end{equation}
\end{proposition}
See Appendix~\ref{appendix-sampling-bias-bounds} for the detailed proof. When deriving a bias bound, to preserve contraction in the first term, we take $N \sim (\beta/\alpha) \log d$, which necessitates a multi-step coupling. The discretization error remains controlled due to the sparsity of the Hessians, which manifests as the $\sqrt{s_{r_i}}$ terms in the bound. By taking $k \to \infty$, \eqref{eqn-sketch-sparse-bound-hmc} results in the final bias bound in Theorem \ref{thm-sparse-formal}.

We now lay out the coupling framework underlying Propositions \ref{prop-sketch-weak-bound-hmc} and \ref{prop-sketch-sparse-bound-hmc}. The framework is common to HMC and UL. Once a contraction is established, over one step or several, the total bias is a sum of single-step discretization errors, each propagated forward under the discrete-time scheme. The object governing this propagation, which is the \emph{propagator matrix} acting on the difference between two paths, is therefore central to the argument. Working with the discrete-time scheme is itself a deliberate choice, as the existing $\ell^2$ analyses in continuous time~\cite{chen2019optimalconvergence,mangoubi2021mixing} exploit the inner-product structure, but the $\ell^\infty$ norm is less compatible with continuous-time differentiation. For HMC, the propagator matrices contract in the $\ell^\infty$ norm in one outer loop under weak interactions, while the $\ell^\infty$ contraction emerges over multiple outer loops under sparse interactions. Moreover, the same class of matrices controls how the discretization error propagates across inner steps and subsequent outer loops. 

The UL analysis follows a similar template, but its iterative inequalities are more involved due to a weighted norm and the absence of
momentum refreshment. We therefore
present the framework for both algorithms but give the iterative inequalities only
for HMC.

\subsection{HMC} 
\label{sec-sketch-hmc}

We compare two chains driven by the same momentum refreshments
$\xi_k, \ldots, \allowbreak\xi_{k+N-1}$. The exact chain $Y$ has an inner loop following the
continuous Hamiltonian dynamics \eqref{eqn-continuous-time-Hamiltonian}, while the
approximate chain $X$ has an inner loop running the leap-frog scheme
\eqref{eqn-leap-frog}. 
We assume $Y_0 \sim \pi$ is at stationarity. For each chain, one outer loop maps the current position to the next by drawing a fresh momentum $\xi$, integrating the inner dynamics for time $mh$ (over
$m$ leap-frog steps in the approximate case), and reading off the resulting position. We write $\mathsf{\Phi}^{mh}_\xi$ for one \emph{exact} outer loop and
$\mathsf{\Phi}^m_{h,\xi}$ for its \emph{leap-frog} counterpart, the subscript $\xi$
recording the Gaussian refreshment used in that loop. After $N$ outer loops, the
approximate position is the $N$-fold composition along the shared refreshment
sequence,
\begin{equation}
\label{eqn-approximate-solution-X}
    X_{k+N} = \mathsf{\Phi}^m_{h,\xi_{k+N-1}} \mathsf{\Phi}^m_{h,\xi_{k+N-2}} \cdots \mathsf{\Phi}^m_{h,\xi_k}(X_k)\,,
\end{equation}
and the exact position $Y_{k+N}$ is given analogously by composing the
$\mathsf{\Phi}^{mh}_{\xi}$ maps. We omit the composition symbol $\circ$
throughout for brevity.

\noindent\textbf{One-step bound.} Coupling the two chains over a single outer loop, we have
        \begin{equation}
        \label{eqn-sketch-one-step}
\begin{aligned}
     &|X_{k+1}-Y_{k+1}|_{2,\ell^\infty}=|\mathsf{\Phi}^m_{h,\xi_k}(X_k) - \mathsf{\Phi}^{mh}_{\xi_k}(Y_k)|_{2,\ell^\infty} \\
    \leq  &\underbrace{|\mathsf{\Phi}^m_{h,\xi_k}(X_k) - \mathsf{\Phi}_{h, \xi_k}^{m}(Y_k)|_{2,\ell^\infty}}_{(a)} + \underbrace{|\mathsf{\Phi}_{h, \xi_k}^{m}(Y_k) - \mathsf{\Phi}_{\xi_k}^{mh}(Y_k)|_{2,\ell^\infty}}_{(b)}\, ,
\end{aligned}
\end{equation}
where $(a)$ measures how the leap-frog loop changes the distance between two paths and $(b)$ is the discretization error in one outer loop. If $\mathsf{\Phi}^m_{h,\xi}$ is a contraction in the $\ell^\infty$ norm, then it suffices to bound the part (b). Weakly interacting potentials guarantee this contraction.

\noindent\textbf{Multi-step bound.}
In general, $\mathsf{\Phi}^m_{h,\xi}$ contracts in the $\ell^2$ norm but not in the $\ell^\infty$ norm. However, an $\ell^\infty$ contraction can be induced by the $\ell^2$ contraction through coupling over multiple steps, which motivated the multi-step coupling argument for the sparse interactions in \cite{chen2024convergence}. The challenge persists here, where the sparsity of the potential controls the multi-step error accumulation. Iterating over $N$ outer loops yields the same split into a
contraction term $(a)$ and a discretization term $(b)$,
        \begin{equation}
        \label{eqn-sketch-multi-step}
\begin{aligned}
    &|X_{k+N} - Y_{k+N}|_{2,\ell^\infty} 
    = |\mathsf{\Phi}^m_{h,\xi_{k+N-1}}\cdots \mathsf{\Phi}_{h,\xi_k}^m(X_k) - \mathsf{\Phi}^{mh}_{\xi_{k+N-1}}\cdots \mathsf{\Phi}^{mh}_{\xi_k}(Y_k)|_{2,\ell^\infty}\\
    \leq &\underbrace{|\mathsf{\Phi}^m_{h,\xi_{k+N-1}}\cdots \mathsf{\Phi}_{h,\xi_k}^m(X_k) - \mathsf{\Phi}^m_{h,\xi_{k+N-1}}\cdots \mathsf{\Phi}_{h,\xi_k}^m(Y_k)|_{2,\ell^\infty}}_{(a)} \\
    & \quad + \underbrace{|\mathsf{\Phi}^m_{h,\xi_{k+N-1}}\cdots \mathsf{\Phi}_{h,\xi_k}^m(Y_k)- \mathsf{\Phi}^{mh}_{\xi_{k+N-1}}\cdots \mathsf{\Phi}^{mh}_{\xi_k}(Y_k)|_{2,\ell^\infty}}_{(b)}\,,
\end{aligned}
\end{equation}
where now the map in $(a)$ becomes an $\ell^\infty$ contraction once $N \sim (\beta/\alpha) \log d$.

\noindent \textbf{Discretization error.} It remains to quantify the discretization term $(b)$, for both a single outer
loop and for $N$ of them. We begin with $(b)$ of \eqref{eqn-sketch-one-step}. To
expose its structure we need the finer building blocks of an outer loop. Let $\mathsf{U}^{mh}_{\textsc{hmc}}$ and $\mathsf{U}^m_{\textsc{hmc},h}$ denote the
exact and leap-frog solution maps on the full phase space $\mathbb{R}^{2d}$ for time $mh$, and
$\mathsf{\Pi}_1$ the projection onto position, so that
$\mathsf{\Phi}^{mh}_\xi(q) = \mathsf{\Pi}_1\mathsf{U}^{mh}_{\textsc{hmc}}(q,\xi)$
and $\mathsf{\Phi}^m_{h,\xi}(q) = \mathsf{\Pi}_1\mathsf{U}^m_{\textsc{hmc},h}(q,\xi)$. Telescoping the difference $\mathsf{\Pi}_1\mathsf{U}^{m}_{\textsc{hmc},h}
- \mathsf{\Pi}_1\mathsf{U}^{mh}_{\textsc{hmc}}$ over the $m$ integration steps and
using that $\mathsf{U}^{ih}_{\textsc{hmc}}(Y_k,\xi_k) \sim \pi \otimes
\mathcal{N}(0,I)$ is at stationarity (since $Y_0 \sim \pi$ and the inner loop is
exact), term $(b)$ of \eqref{eqn-sketch-one-step} is bounded by 
\begin{equation}
\label{eqn-one-step-hmc-discretization-error-decomposition-propagation}
    (b) \leq \sum_{i=0}^{m-1}\big|\,\sfPi_1 \sfU_{\hmc,h}^{m-i-1}(\sfU_{\hmc,h}^{1} - \sfU_{\hmc}^h)(q^*, p^*)\,\big|_{2, \ell^\infty}\,,
    \qquad (q^*, p^*)\sim \pi \otimes \mathcal{N}(0,I)\,.
\end{equation}

Each term represents the error incurred in a single integration step, propagated over the subsequent $m-i-1$ leap-frog steps. 

The discretization error over $N$ outer loops, term $(b)$ of
\eqref{eqn-sketch-multi-step}, telescopes in the same way at the level of outer
loops. It is a sum over loops $l = 0, \ldots, N-1$, of the one-loop error
incurred at loop $l$, propagated through the remaining $N-l-1$ outer loops. Combined with a single-loop decomposition similar to \eqref{eqn-one-step-hmc-discretization-error-decomposition-propagation}, the total error is a sum of single integration-step errors, each
propagated forward through the remaining leap-frog steps within its own loop and
through all subsequent outer loops. Term $(b)$ in \eqref{eqn-sketch-multi-step} can be therefore bounded as
\begin{equation}
\label{eqn-multi-step-hmc-discretization-error-decomposition-propagation}
    (b) \leq \sum_{l=0}^{N-1}\sum_{i=0}^{m-1}|\sfPhi_{h,\xi_{k+N-1}}^m \cdots \sfPhi_{h,\xi_{k+l+1}}^m\sfPi_1 \sfU_{\hmc,h}^{m-i-1}(\sfU_{\hmc,h}^{1} - \sfU_{\hmc}^h)(q^*, p^*)|_{2,\ell^\infty}\,,
\end{equation}
where $(q^*, p^*) \sim \pi \otimes \mathcal{N}(0,I)$. Here $i$ indexes the inner integration steps and $l$ the outer loops. Characterizing the propagator matrices that govern how a local error evolves into
the position coordinate is therefore the core of the discretization-error
analysis.

\subsection{UL}
\label{sec-sketch-underdamped}
The framework for UL follows that for HMC, with two modifications. First, the
dynamics evolve without momentum refreshment, so we do not project onto the
position between steps. Second, since contraction in the $\ell^2$ norm
requires a weighted metric \cite{monmarche2021highdimensionalmcmc,leimkuhler2024contraction}, we adopt a weighted $\ell^\infty$ norm throughout.
\subsubsection{Notation and weighted norm}
 We use solution-map notation analogous to HMC. The one-step BAOAB map with step
size $h$ and O-step noise $\xi$ is $\sfU_{\baoab,h}^\xi$, and we write
$\sfU_{\abao,h}^\xi$, $\sfU_{\bao,h}^\xi$, and $\sfU_{\ab,h}$ for the ABAO, BAO,
and AB sub-schemes, defined by the operator compositions
\begin{equation}
\label{eqn-def-baoab-aba-bao-ab}
\begin{aligned}
    &\sfU_{\baoab,h}^\xi \coloneqq \sfB_{h/2} \sfA_{h/2} \sfO_{h}(\cdot; \xi) \sfA_{h/2}\sfB_{h/2}\,, \quad
    \sfU_{\abao,h}^\xi\coloneqq \sfO_h(\cdot;\xi)\sfA_{h/2}\sfB_{h}\sfA_{h/2}\,,\\
    & \sfU_{\bao,h}^\xi \coloneqq \sfO_{h}(\cdot;\xi)\sfA_{h/2}\sfB_{h/2}\,, \quad
    \sfU_{\ab,h} \coloneqq \sfB_{h/2}\sfA_{h/2}\,.
\end{aligned}
\end{equation}
We use half steps for each $\sfA$ and $\sfB$ in BAO and AB rather than the
conventional full steps, since we invoke them as the components
$\sfU_{\baoab,h}^\xi = \sfU_{\ab,h} \circ \sfU_{\bao,h}^\xi$ of the BAOAB scheme. In place of the continuous-time UL dynamics
\eqref{eqn-continuous-time-underdamped}, we measure the discretization error
against the HOH splitting $\sfU_{\hoh,h}^\xi \coloneqq \sfH_{h/2}\, \sfO(\cdot;\xi)\,
\sfH_{h/2}$, where $\sfH_h \coloneqq \sfU_{\hmc,h}$ is the Hamiltonian solution map. Following \cite{leimkuhler2024contraction}, this keeps the error estimate accurate in the high-friction regime, and HOH preserves the target
$\pi \otimes \mathcal{N}(0,I)$.

We consider the weighted $\ell^\infty$ norm
\begin{equation}
\label{eqn-def-weighted-linf-norm}
    |(q,p)|_{\ell^\infty_w} \coloneqq \left|W\begin{bmatrix} q\\ p \end{bmatrix}\right|_{\ell^\infty}\,,
    \qquad
    W = \begin{bmatrix}
        I & bI\\
        0 & \sqrt{a-b^2}\,I
    \end{bmatrix}\,,
\end{equation}
where $a = \tfrac{1}{\beta}$ and $b = \tfrac{h}{1-\eta}$. 
The weight matrix $W$
is taken from~\cite{leimkuhler2024contraction}, where the same weight makes BAOAB
contract in $\ell^2$. Equivalence with the Euclidean norm requires $b^2 \le a/4$
\cite{monmarche2021highdimensionalmcmc,leimkuhler2024contraction}, under which
\begin{equation}
\label{eqn-l2-norm-equivalence}
    \tfrac12\big(|q|_{\ell^2}^2 + a |p|_{\ell^2}^2\big) \leq |(q,p)|_{\ell^2_w}^2 \leq \tfrac32\big(|q|_{\ell^2}^2 + a |p|_{\ell^2}^2\big)\,.
\end{equation}
This condition is met under the step-size requirement
$h \leq \frac{1-\eta}{2\sqrt{\beta}}$ for UL. For
probabilistic bounds we use the expected norms $|(q,p)|_{2,\ell^\infty_w}
\coloneqq \sqrt{\bE[\,|(q,p)|_{\ell^\infty_w}^2\,]}$ and its unweighted
counterpart $|\cdot|_{2,\ell^\infty}$.

Starting from $(X_k, P_k)$, the BAOAB chain evolves over $N$ steps to
\begin{equation}
\label{eqn-ul-chains}
    (X_{k+N}, P_{k+N}) = \sfU_{\baoab,h}^{\xi_{k+N-1}}\cdots\sfU_{\baoab,h}^{\xi_k}(X_k, P_k)\,,
\end{equation}
and the HOH chain evolves analogously from $(Y_k, P'_k)$ by composing the
$\sfU_{\hoh,h}^\xi$ maps, where $\xi_k, \ldots, \xi_{k+N-1}$ are the
O-step Gaussian noises, shared by both chains to match the coupling
below.

\subsubsection{Analysis framework via coupling}
\label{sec-subsub-underdamped-coupling-framework}
Given a fixed $k$, we couple the two chains and initialize $(Y_0, P'_0)$ from
the HOH stationary distribution $\pi\otimes\mathcal{N}(0,I)$. 

\noindent\textbf{Multi-step bound.}  As in HMC, the
position gap splits into a contraction term $(a)$ and a discretization term
$(b)$,
\begin{equation}
\label{eqn-sketch-underdamped}
\begin{aligned}
   &|X_k - Y_k|_{2,\ell^\infty}\\
   \leq & \underbrace{\left|\sfPi_1\left(\sfU_{\baoab,h}^{\xi_{k-1}}\cdots\sfU_{\baoab,h}^{\xi_0} (X_0, P_0) -\sfU_{\baoab,h}^{\xi_{k-1}}\cdots \sfU_{\baoab,h}^{\xi_0} (Y_0, P'_0)\right)\right|_{2,\ell^\infty}}_{(a)}\\
   & + \underbrace{\left|\sfPi_1\left(\sfU_{\baoab,h}^{\xi_{k-1}}\cdots\sfU_{\baoab,h}^{\xi_0} (Y_0, P'_0) -\sfU_{\hoh,h}^{\xi_{k-1}}\cdots \sfU_{\hoh,h}^{\xi_0} (Y_0, P'_0)\right)\right|_{2,\ell^\infty}}_{(b)}\,.
\end{aligned}
\end{equation}
In both the weak and sparse interaction cases, the contraction is established
over multiple steps, though for different reasons. For weak interactions, we
exploit the one-step weighted contraction of the $\sfA\sfB\sfA\sfO$ operator,
writing $(\sfB\sfA\sfO\sfA\sfB)^k = \sfB\sfA\sfO\,(\sfA\sfB\sfA\sfO)^{k-1}\,\sfA\sfB$,
so that the contraction of the inner $\sfA\sfB\sfA\sfO$ operators drives the bound
while $\sfB\sfA\sfO$ and $\sfA\sfB$ contribute only constant prefactors. These
prefactors are what preclude a strict one-step bound and require accumulating
contraction over several steps. This decomposition follows the $\ell^2$ argument
for the BAOAB  scheme in \cite{leimkuhler2024contraction}. For sparse interactions,
$\sfA\sfB\sfA\sfO$ is no longer a weighted-$\ell^\infty$ contraction, so we
instead analyze the full $\sfB\sfA\sfO\sfA\sfB$ map and upgrade the $\ell^2$ contraction of
\cite{leimkuhler2024contraction} to an $\ell^\infty$ contraction over multiple
steps. Since we only need the asymptotic bias in the position variable, one
option there is to track the propagator matrices governing error propagation into
the position coordinates, whose marginal is non-Markovian across intermediate
steps. In both cases, we establish the multi-step contraction from the initial
state through to the final step $k$ rather than step by step.

In both regimes, \eqref{eqn-sketch-underdamped} is taken in the unweighted $\ell^\infty$ norm of the position coordinate to match the Wasserstein bias structure. The weighted contraction enters by inserting $W$ and $W^{-1}$ at the
intermediate steps and invoking the equivalence in
\eqref{eqn-l2-norm-equivalence}.  In Section \ref{sec-sub-Wasserstein-bias}, we omit an explicit iterative Wasserstein bias inequality for UL analogous to Propositions \ref{prop-sketch-weak-bound-hmc} and \ref{prop-sketch-sparse-bound-hmc}, because the insertion of the weight matrices and the end-to-end bias bound complicate the statement.

\noindent\textbf{Discretization error.} The discretization term $(b)$ is handled through the same propa\-gator-ma\-trix
mechanism as HMC, following the asymptotic-bias analysis of BAOAB in
\cite{leimkuhler2024contraction} with the weighted $\ell^2$ norm replaced by the
$\ell^\infty$ norm. To bound the error accurately across regimes of $\gamma$,
\cite{leimkuhler2024contraction} analyzes the process over blocks of
$\tilde l = \lceil \frac{1}{2\sqrt\beta h} \rceil$ steps. Organizing $(b)$ in
\eqref{eqn-sketch-underdamped} block by block gives
\begin{equation}
\label{eqn-multi-step-baoab-discretization-error-decomposition-propagation}
\begin{aligned}
    (b) & \leq  \sum_{i=0}^{\lfloor k/\tilde l\rfloor - 1} \left|\sfPi_1\sfU_{\baoab,h}^{k-(i+1)\tilde l}(\sfU_{\baoab,h}^{\tilde l} - \sfU_{\hoh,h}^{\tilde l}) (q^*, p^*)\right|_{2,\ell^\infty} \\
    & \quad + \left|\sfPi_1\left(\sfU_{\baoab,h}^{k-\lfloor k/\tilde l \rfloor \tilde l}  - \sfU_{\hoh,h}^{k - \lfloor k/\tilde l \rfloor \tilde l}\right)(q^*, p^*)\right|_{2,\ell^\infty}\,,
\end{aligned}
\end{equation}
where $(q^*, p^*)\sim \pi \otimes \mathcal{N}(0,I)$, and $\sfU_{\baoab,h}^l$,
$\sfU_{\hoh,h}^l$ denote the $l$-fold compositions of the BAOAB and HOH solution
maps, with the coupled Gaussian noises from \eqref{eqn-sketch-underdamped}
omitted for brevity. Each term is the error over at most $\tilde l$ steps,
propagated through the subsequent BAOAB steps. Characterizing these propagator
matrices for the evolution of the difference between two paths is therefore central to the accumulated discretization-error analysis.

\section{A Matrix Polynomial Framework for Propagator Analysis}
\label{sec-matrix-polynomial}
As shown in the previous section, the propagator matrices are fundamental to characterizing both the contraction rates and the discretization error under the $\ell^\infty$ norm. In the absence of continuous-time differentiation, we track the evolution of path differences by introducing a framework that expresses these propagators as \emph{matrix polynomials}. The underlying polynomial coefficients then facilitate the derivation of the required $\ell^\infty$ bounds. We illustrate this structure with a Gaussian distribution.
\subsection{Propagators: Gaussian examples}
The propagators of the leap-frog scheme are closely related to Chebyshev polynomials, as illustrated by the following proposition.
\begin{proposition}
\label{prop-hmc-Gaussian-propagators}
For $\pi = \mathcal{N}(0, \Sigma)$, after $k$ leap-frog steps, the difference between the position components $x_k = \sfPi_1\sfU_{\hmc,h}^k(x_0, p_0)$ and $y_k = \sfPi_1\sfU_{\hmc, h}^k(y_0, p'_0)$ satisfies
\begin{equation}
\label{eqn-hmc-Gaussian-propagators}
x_k - y_k = T_k(I - \frac{h^2}{2}\Sigma^{-1}) (x_0 - y_0) + hU_{k-1}(I - \frac{h^2}{2}\Sigma^{-1})(p_0 - p'_0)\,,
\end{equation}
where $T_k$ is the $k$-th Chebyshev polynomial of the first kind, and $U_k$ is the $k$-th Chebyshev polynomial of the second kind. We adopt the convention $U_{-1} = 0$, which is consistent with the standard three-term recurrence relation $U_{k+1}(x) = 2xU_k(x) - U_{k-1}(x)$ for $k\geq 0$ when initializing the sequence at $U_0(x) = 1$.
\end{proposition}
\begin{proof}[Proof sketch]
Writing down the iteration \eqref{eqn-leap-frog} in two consecutive steps and eliminating the momentum
difference yields, for $\Delta q_k \coloneqq x_k - y_k$, the three-term recurrence
\begin{equation}
\label{eqn-three-term-recurr-in-x-gaussian}
    \Delta q_{k+1} = (2I-h^2 \Sigma^{-1})\Delta q_k - \Delta q_{k-1}\,,\quad k \geq 1.
\end{equation} Deriving the corresponding recurrence for the propagator matrices reveals the Chebyshev polynomials.
\end{proof}

This is a special case of Proposition~\ref{prop-matrix-polynomial} for which we provide the complete proof. It is natural for $\Delta q_k$ to satisfy a three-term recurrence, as the leap-frog scheme is a discretization of a second-order ordinary differential equation in the position variable. In fact, the position-only iteration \eqref{eqn-three-term-recurr-in-x-gaussian} is the St\"ormer form of the leap-frog scheme~\cite{leimkuhler2015moleculardynamics}. The relation to the Chebyshev polynomials is also well-motivated. For instance,  in a one-dimensional system with the quadratic potential $V(q) = \frac{1}{2} q^2$, the exact solution is $q(t) = q(0)\cos(t) + p(0)\sin(t)$. The Chebyshev polynomials provide the discrete analogue to these trigonometric solutions via the identities $T_k(\cos\theta) = \cos(k\theta)$ and $U_{k-1}(\cos\theta) = \sin(k\theta)/\sin\theta$. 

With momentum refreshment between outer loops, multiplying their matrix polynomials
gives the following HMC propagator across multiple loops. 
\begin{proposition}
\label{prop-hmc-Gaussian-outer-propagators}
For $\pi = \mathcal{N}(0, \Sigma)$, let $x_{Nm+k}$ and $y_{Nm+k}$ denote the position components after initial $k$ leap-frog steps, followed by $N$ full HMC outer loops (each consisting of $m$ steps). Specifically, let $x_{Nm+k} \coloneqq \sfPhi_{h,\xi_{N-1}}^m \cdots \sfPhi_{h,\xi_{0}}^m\sfPi_1 \sfU_{\hmc,h}^{k}(x_0, p_0)$ and $y_{Nm+k} \coloneqq \sfPhi_{h,\xi_{N-1}}^m \cdots \sfPhi_{h,\xi_{0}}^m\sfPi_1 \sfU_{\hmc,h}^{k}(y_0, p'_0)$. The difference $\Delta q_{Nm+k} \coloneqq x_{Nm+k} - y_{Nm+k}$ satisfies
\begin{equation}
\label{eqn-hmc-Gaussian-outer-propagators}
\begin{aligned}
\Delta q_{Nm+k}  &= \left( T_m(I - \frac{h^2}{2}\Sigma^{-1}) \right)^N\\
& \qquad \cdot \left(T_k(I - \frac{h^2}{2}\Sigma^{-1}) (x_0 - y_0) 
 + hU_{k-1}(I - \frac{h^2}{2}\Sigma^{-1})(p_0 - p'_0)\right)\,.
\end{aligned}
\end{equation}
Here, $x_{Nm+k}$ and $y_{Nm+k}$ are coupled by utilizing identical momentum refreshments $\xi_0, \ldots, \xi_{N-1}$ at each outer loop.
\end{proposition}
The UL dynamics are closely related to the Hamiltonian dynamics but crucially involve a friction term. The friction induces a modification to the polynomial structure. Specifically, the resulting three-term recurrence for the BAOAB scheme incorporates damped coefficients, which recovers the leap-frog coefficients by taking $\gamma =0$.
\begin{proposition}
\label{prop-underdamped-Gaussian-propagators}
For $\pi = \mathcal{N}(0, \Sigma)$, let $x_{k}$ and $y_{k}$ denote the position components after $k$ BAOAB steps. Specifically, let $x_{k} \coloneqq \sfPi_1\sfU_{\baoab,h}^{\xi_{k-1}}\cdots\sfU_{\baoab,h}^{\xi_0}(x_0, p_0)$ and $y_{k} \coloneqq \sfPi_1\sfU_{\baoab,h}^{\xi_{k-1}}\cdots\sfU_{\baoab,h}^{\xi_0} (y_0, p'_0)$. The difference $\Delta q_k \coloneqq x_k - y_k$ then satisfies
\begin{equation}
\label{eqn-underdamped-Gaussian-propagators}
\Delta q_{k+1} = (1+\eta)(I -\frac{h^2}{2}\Sigma^{-1})\Delta q_k - \eta \Delta q_{k-1}\,, \quad k \geq 1.
\end{equation}
\end{proposition}
\subsection{Propagators: general cases}
\label{sec-matrix-polynomial-general-cases}
Inspired by the Gaussian case, we establish that the propagator matrices can be represented as matrix polynomials, whose coefficients can be accurately characterized.  However, the Hessian varies along the trajectory, and the
resulting matrices do not commute. The polynomial must therefore be generalized to
one in several ordered, non-commuting arguments, each monomial being a product
taken in a fixed order. We  view the general propagator as a ``path-dependent'' polynomial, where the recurrence remains identical to the Gaussian case but the arguments are updated in order at each step. We now define these polynomials.

\begin{definition}[Multivariate Damped Chebyshev Polynomials]
\label{def-multivariate-damped-Chebyshev}
    For a sequence of matrices $M_i \in \mathbb{R}^{d\times d}$, we define the
multivariate damped Chebyshev polynomial of the first kind,
$\tilde T_k^\eta(M_1,\dots, M_k)$, via the three-term recurrence
  \begin{equation}
\label{eqn-scalar-multivar-damped-chebyshev-first-kind}
    \tilde T_{k+1}^\eta(M_1,\dots, M_{k+1}) = (1+\eta)\,M_{k+1}\tilde T_k^\eta(M_1,\dots, M_k) - \eta\, \tilde T_{k-1}^\eta(M_1,\dots, M_{k-1})
\end{equation}
     for $k\geq 1$, with the initial conditions
     $\tilde T_0^\eta \equiv I$ and
$\tilde T_1^\eta(M_1) = \tfrac{1+\eta}{2} M_1 + \tfrac{1-\eta}{2}I$. Note that the new
argument $M_{k+1}$ enters on the left. 
The multivariate damped Chebyshev polynomial of
the second kind, $\tilde U_k^\eta(M_1,\dots,M_k)$, is defined by the same
recurrence
\begin{equation}
    \label{eqn-scalar-multivar-damped-chebyshev-second-kind}
    \tilde U_{k+1}^\eta(M_1,\dots, M_{k+1}) = (1+\eta)M_{k+1}\tilde U_k^\eta(M_1,\dots, M_k) - \eta \tilde U_{k-1}^\eta(M_1,\dots, M_{k-1})
    \end{equation}
    for $k\geq 0$, with different initial conditions $\tilde U_{-1}^\eta \equiv 0$, $\tilde U_{0}^\eta \equiv \tfrac{1+\eta}{2}I$. For notation brevity, we also adopt the sequence notation $\tilde T_{k}^\eta(\{M_i\}_{i=1}^{k})$ and $\tilde U_{k}^\eta(\{M_i\}_{i=1}^{k})$.
\end{definition}
We note that when $\eta =1$, which corresponds to the frictionless Hamiltonian regime $\gamma = 0$, setting all $M_i \coloneqq xI$ with a scalar $x$ gives $\tilde T_k^\eta(xI,\dots,xI) = T_k(x)I$
and $\tilde U_k^\eta(xI,\dots,xI) = U_k(x)I$, recovering the classical Chebyshev
polynomials $T_k$ and $U_k$. We also note that, even for scalar arguments, our polynomials differ from the
conventional multivariate Chebyshev polynomials on $\mathbb{R}^d$, which are built
from tensor products of one-dimensional polynomials rather than a recurrence in
ordered arguments. In the following proposition, we show that the multivariate damped Chebyshev polynomials govern the general propagator matrices in the leap-frog and BAOAB schemes. In particular, it includes the Gaussian propositions stated above. We present its proof in Appendix~\ref{appendix-derivation-matrix-polynomial}.

\begin{proposition}
\label{prop-matrix-polynomial}
The propagators of the leap-frog and BAOAB schemes are characterized by the
multivariate damped Chebyshev polynomials $\tilde T_k^\eta$ and $\tilde U_{k-1}^\eta$. Throughout, let $(x_0, p_0)$ and $(y_0, p'_0)$ be two initial states evolved
under synchronous coupling, sharing the same momentum refreshments in HMC and
the same O-step Gaussian noise in BAOAB, and let $x_i, y_i$ denote the
resulting coupled positions after $i$ steps along the trajectory specified in
each item below. We set 
\begin{equation}
\label{eqn-prop-matrix-polynomial-Mi}
    M_i \coloneqq I - \tfrac{h^2}{2} H_i\,,
    \qquad
    H_i = \int_0^1 \nabla^2 V\big(\tau x_i + (1-\tau) y_i\big)\,{\rm d}\tau\,.
\end{equation}
\begin{enumerate}
    \item \emph{(One scheme, $k$ steps.)} For the leap-frog iterations in HMC
    ($\eta = 1$) and the BAOAB iterations in UL ($\eta = e^{-\gamma h}$), write
    $\sfU^i$ for $i$ steps of the scheme, namely the leap-frog map
    $\sfU_{\hmc,h}^i$ for HMC and the BAOAB composition
    $\sfU_{\baoab,h}^{\xi_{i-1}}\cdots\allowbreak\sfU_{\baoab,h}^{\xi_0}$ for UL, so that
    $x_i = \sfPi_1\sfU^i(x_0,p_0)$ and $y_i = \sfPi_1\sfU^i(y_0,p'_0)$. The
    position difference after $k$ steps satisfies
    \begin{equation}
    \label{eqn-prop-matrix-polynomial-single-scheme}
        x_k - y_k = \tilde T_k^\eta(M_0, \dots, M_{k-1})(x_0 - y_0)
        + h\, \tilde U_{k-1}^\eta(M_1, \dots, M_{k-1})(p_0 - p'_0)\,.
    \end{equation}
 
    \item \emph{(HMC, initial leap-frog run plus $N$ outer loops.)} After $k$
    leap-frog steps followed by $N$ outer loops of $m$ leap-frog steps each, the position difference between
    $x_{Nm+k} \coloneqq \sfPhi_{h,\xi_{N-1}}^m \cdots \sfPhi_{h,\xi_0}^m\, \sfPi_1 \sfU_{\hmc,h}^{k}(x_0, p_0)$
    and
    $y_{Nm+k} \coloneqq \sfPhi_{h,\xi_{N-1}}^m \cdots \sfPhi_{h,\xi_0}^m\, \sfPi_1 \allowbreak\sfU_{\hmc,h}^{k}(y_0, p'_0)$
    satisfies
    \begin{equation}
    \label{eqn-propagator-multiple-outer-loop}
    \begin{aligned}
        x_{Nm+k} - y_{Nm+k}
        &= \left[\prod_{j \in\{k + im\}_{i=0}^{N-1}}^{\longleftarrow}
            \tilde T_m^1(M_j, M_{j+1}, \dots, M_{j+m-1})\right] \\
        & \cdot \Big[\tilde T_k^1(M_0, \dots, M_{k-1})(x_0 - y_0)
            + h\, \tilde U_{k-1}^1(M_1, \dots, M_{k-1})(p_0 - p'_0)\Big]\,,
    \end{aligned}
    \end{equation}
    where $x_j, y_j$ are the coupled positions after $j$ total steps (the first
    $k$ within the initial leap-frog run, the remainder within subsequent outer
    loops), and $\prod^{\longleftarrow}$ orders the factors so that those with
    smaller index $j$ stand to the right.
\end{enumerate}
\end{proposition}

\section{Bounds for Propagators}
\label{sec-propagator-bounds}
In this section, we present $\ell^2$ and $\ell^\infty$ bounds for propagators, instantiating the technical framework sketched in Section \ref{sec-sketch-techniques}. While the $\ell^2$ bounds are obtained by adapting existing techniques to our setting, the $\ell^\infty$ bounds represent our primary analytical contribution. These estimates hinge on the structural properties of the multivariate damped Chebyshev polynomials, which allow us to bridge from $\ell^2$ stability to $\ell^\infty$ guarantees under the assumption of weak or sparse interactions.

\subsection{$\ell^2$ bounds}
\label{sec-l2-bounds}
The relevant $\ell^2$ bounds for HMC are as follows.
\begin{proposition}[$\ell^2$ bounds for HMC propagators]
\label{prop-l2bounds-hmc}
Suppose Assumption \ref{assumption-V-log-concave} hold. If $H_i = \int \nabla^2 V {\rm d} \nu_i$  for a probability measure $\nu_i$ on $\mathbb{R}^d$,
we have the following $\ell^2$ bounds for the leap-frog propagators $\tilde T_k$ and $\tilde U_{k-1}$ provided $kh \leq 1/\sqrt{20\beta}$: 
\begin{equation}
\label{eqn-l2bounds-hmc-Tk}
%\label{eqn-l2bounds-hmc-Uk} 
        |\tilde T_k^1(\{I - \frac{h^2}{2}H_i\}_{i=0}^{k-1})|_{\ell^2} \leq 1 - \frac{\alpha}{10}(kh)^2\,,\quad 
    |h\tilde U_{k-1}^1(\{I - \frac{h^2}{2}H_i\}_{i=1}^{k-1})|_{\ell^2} \leq \sqrt{\frac{2}{5\beta}}\,.
\end{equation}
\end{proposition}
The proof is provided in Appendix~\ref{appendix-bounds-leap-frog}. Combined with Proposition~\ref{prop-matrix-polynomial}, Proposition~\ref{prop-l2bounds-hmc} states that the leap-frog scheme induces a contraction in the position components, provided the trajectories are initialized with identical momentum and the total integration time remains within a specific threshold scaled as $1/\sqrt{\beta}$.  
The momentum difference may expand, though this growth remains well-controlled by a magnitude scaled as $1/\sqrt{\beta}$. These scalings are consistent with the Gaussian setting under the continuous-time dynamics within the time range. We adapt the proof of an existing discrete-time contraction result~\cite[Lemma~19]{bou-rabee203meanfield} to obtain the bound for $\tilde T_k^1$. %\eqref{eqn-l2bounds-hmc-Tk}.
Applying a similar framework for the momentum components, we establish the  $\tilde U_k^1$ bound. 

For UL, contraction is induced in the weighted norms of
Section~\ref{sec-sketch-techniques}, so we build the weight matrix $W$ into the
propagator representation. For two coupled trajectories with initial states
$z_0 = (x_0,p_0)$ and $z'_0 = (y_0, p'_0)$, viewed as column vectors in
$\mathbb{R}^{2d}$, a one-step update $\sfU$ (any of BAOAB, ABAO, BAO, AB) gives
$z_1 = \sfU(z_0)$ and $z'_1 = \sfU(z'_0)$. The resulting state difference is
$z_1 - z'_1 = M(H_0)(z_0 - z'_0)$, where $M(H_0)$ is the one-step propagator of
the unweighted difference, parametrized by an integrated Hessian $H_0$.
Inserting $I = W^{-1}W$ gives
\begin{equation}
\label{eqn-weighted-propagators-primitive}
    W(z_1 - z'_1) = \big(W M(H_0) W^{-1}\big)\, W(z_0 - z'_0)\,.
\end{equation}
The weighted propagator $M^w(H_0) \coloneqq WM(H_0)W^{-1}$ acts on the weighted
difference $W(z_0 - z'_0)$, so a weighted $\ell^2$ contraction follows from
bounding $|M^w(H_0)|_{\ell^2}$, and discretization errors are measured in the
 weighted norm $|z_0 - z_0'|_{\ell^2_w} = |W(z_0-z_0')|_{\ell^2}$. Over $k$ steps the propagator is
the composition $M_k \coloneqq M(H_{k-1})\cdots M(H_0)$, with weighted counterpart
$M^w_k \coloneqq W M_k W^{-1}$, and a $k$-step weighted contraction follows from
bounding the full-space
operator norm $|M^w_k|_{\ell^2}$.

Our goal, however, is the bias of the \emph{unweighted} position marginal,
recovered by the projection $\sfPi_1 = \begin{bmatrix} I & 0 \end{bmatrix}
\in \mathbb{R}^{d\times 2d}$. After one step, we have
\begin{equation}
\label{eqn-weighted-position-propagators}
    \sfPi_1(z_1 - z'_1) = \begin{bmatrix} I & 0 \end{bmatrix} W^{-1} M^w(H_0)\, W(z_0 - z'_0)\,.
\end{equation}
The factor $\begin{bmatrix} I & 0 \end{bmatrix} M(H_0)$ captures the dependence of
$x_1 - y_1$ on the initial position and momentum differences. The same structure
holds over $k$ steps, with $\tilde T_k^\eta$ and $h\tilde U_{k-1}^\eta$ as the
position-row blocks of $M_k$. We bound both the full-space propagator and these weighted
blocks for the BAOAB scheme as follows.

\begin{proposition}[$\ell^2$ bounds for UL propagators]
\label{prop-l2bounds-baoab} 
Suppose $H_i$ satisfies the same assumptions as in Proposition \ref{prop-l2bounds-hmc}. For a step size $h \leq \frac{1-\eta}{2\sqrt{\beta}}$, and letting $c(h) = \frac{\alpha h^2}{4(1-\eta)}$, the full-space weighted BAOAB propagator satisfies 
\begin{equation}
\label{eqn-prop-l2bounds-baoab-full-space-weighted}
    |M_{\baoab,k}^w|_{\ell^2} \leq 7(1-c(h))^{\frac{k-1}{2}}\,,
\end{equation}

and its weighted position-row blocks satisfy
\begin{equation}
\label{eqn-prop-l2bounds-baoab-11block}
    |\tilde T_k^\eta(\{I - \frac{h^2}{2}H_i\}_{i=0}^{k-1})|_{\ell^2} \leq 7 \sqrt{2} \left(1 - c(h)\right)^{\frac{k-1}{2}}\,,
\end{equation}
\begin{equation}
\label{eqn-prop-l2bounds-baoab-12block}
\begin{aligned}
    &\frac{1}{\sqrt{a-b^2}}\left|-b \tilde T_k^\eta(\{I - \frac{h^2}{2}H_i\}_{i=0}^{k-1}) + h \tilde U_{k-1}^\eta(\{I - \frac{h^2}{2}H_i\}_{i=1}^{k-1})\right|_{\ell^2}  \leq 7\sqrt{2}\left( 1 - c(h)\right)^{\frac{k-1}{2}}\,,
\end{aligned}
\end{equation}
\end{proposition}
The proof is detailed in Appendix~\ref{appendix-bounds-baoab}. Proposition \ref{prop-l2bounds-baoab} is an adaptation of the existing contraction result~\cite[Theorem~5.1]{leimkuhler2024contraction}. The specific combinations of matrices are of interest because we have
\begin{equation}
\label{eqn-weight-matrix-polynomial-baoab}
\begin{bmatrix}
    \tilde T_k^\eta & h \tilde U_{k-1}^\eta
\end{bmatrix} W^{-1} =
\begin{bmatrix} 
\tilde T_k^\eta & - \frac{b}{\sqrt{a-b^2}} \tilde T_k^\eta + \frac{1}{\sqrt{a-b^2}}h\tilde U_{k-1}^\eta
\end{bmatrix}\,.
\end{equation}

\subsection{$\ell^\infty$ bounds under weak interactions}
The weak interactions defined in Assumption \ref{assumption-weak-interactions} indicate that $\nabla^2 V$ is a small perturbation of a diagonal matrix. The resulting
propagators therefore have $\ell^\infty$ norms close to their $\ell^2$ norms,
since the two norms coincide for a diagonal matrix, and a contraction is preserved in $\ell^\infty$ whenever the interactions are weak. For HMC, since one outer loop is $m$ leap-frog steps, the relevant object is the multi-step propagator $\tilde T_k^1$.
\begin{proposition}[$\ell^\infty$ bounds under weak interactions, HMC]
\label{prop-linftybounds-hmc-weak} Suppose Assumptions \ref{assumption-V-log-concave} and \ref{assumption-weak-interactions} hold with the weak interaction parameter $C^{(O)} = 1/50$. Let the integrated Hessians $H_i$ be defined as  in Proposition \ref{prop-l2bounds-hmc}. For $kh \leq 1/\sqrt{20\beta}$, we have
\begin{equation}
|\tilde T_k^1 (\{I - \frac{h^2}{2}H_i\}_{i=0}^{k-1})|_{\ell^\infty} \leq 1 - \frac{\alpha}{10}(kh)^2 + \frac{kh\alpha}{100\sqrt{\beta}}\,,\quad
|h\tilde U_{k-1}^1(\{I - \frac{h^2}{2}H_i\}_{i=1}^{k-1})|_{\ell^\infty} \leq \frac{3}{2\sqrt\beta}\,.
\end{equation}
In particular, for $kh = \frac{1}{\sqrt{20\beta}}$, we have the contraction $|\tilde T_k^1 (\{I - \frac{h^2}{2}H_i\}_{i=0}^{k-1})|_{\ell^\infty} \leq 1 - \frac{\alpha}{400\beta}$.
\end{proposition}
The proof can be found in Appendix~\ref{appendix-bounds-leap-frog}. 

For UL, by contrast, the contraction is obtained at the level of a single step
rather than through a multi-step polynomial. We decompose the full $\sfB\sfA\sfO\sfA\sfB$ operator into $\sfB\sfA\sfO$, $\sfA\sfB\sfA\sfO$, and $\sfA\sfB$, and establish a one-step contraction for the $\sfA\sfB\sfA\sfO$ operator.
As in the general
construction~\eqref{eqn-weighted-propagators-primitive},
$M_{\abao}^w(H_0)$ is the weighted one-step propagator of the
$\sfU_{\abao,h}^\xi$ map, satisfying
$W(\sfU_{\abao,h}^\xi(z_0) - \sfU_{\abao,h}^\xi(z'_0)) = M_{\abao}^w(H_0)\, W(z_0 - z'_0)$
for states $z_0 = (x_0,p_0)$, $z'_0 = (y_0,p'_0)$ coupled through the same
O-step noise. 
Here $H_0 = \int_0^1 \nabla^2 V\big(\tau(x_0+\tfrac{h}{2}p_0)+(1-\tau)(y_0 + \tfrac{h}{2}p'_0)\big)\,{\rm d}\tau$, and
\begin{equation}
\label{eqn-def-abao-weighted-matrix}
    M_{\abao}^w(H_0)= \begin{bmatrix}
        I - \frac{h^2}{2} H_0 - \frac{\eta}{1-\eta}h^2 H_0
        & h^3\frac{\frac14 + \frac{\eta}{(1-\eta)^2}}{\sqrt{1/\beta-h^2/(1-\eta)^2}}H_0\\[4pt]
        - \eta h \sqrt{\frac{1}{\beta}-\frac{h^2}{(1-\eta)^2}}H_0
        & \eta\big(I + (-\frac12 + \frac{1}{1-\eta})h^2H_0\big)
    \end{bmatrix}\,.
\end{equation}
We work with this full phase-space propagator rather than its position
projection because the $\sfA\sfB\sfA\sfO$ operator drives contraction in the joint
space over several iterations before the position marginal is extracted. The
maps $M_{\bao}^w$ and $M_{\ab}^w$ are defined analogously by their one-step
propagation of the weighted state differences, though the state combination over
which $H_0$ integrates may differ between schemes.

\begin{proposition}[$\ell^\infty$ bounds under weak interactions, UL]
\label{prop-propagators-weak-interactions-baoab}
Suppose Assumptions \ref{assumption-V-log-concave} and \ref{assumption-weak-interactions} hold with $C^{(O)} = 1/20$. Let the integrated Hessian $H_0$ be defined as  in Proposition \ref{prop-l2bounds-hmc}. For $h \leq \frac{1-\eta}{2\sqrt\beta}$ and recalling $c(h) = \frac{\alpha h^2}{4(1-\eta)}$ from Proposition \ref{prop-l2bounds-baoab},  we have 
\begin{equation}
    |M_{\abao}^w(H_0)|_{\ell^\infty} \leq  1 - \frac12c(h)\,, \quad |M_{\bao}^w(H_0)|_{\ell^\infty} \leq 3\,, \quad |M_{\ab}^w(H_0)|_{\ell^\infty} \leq 2\,.
\end{equation}
\end{proposition}
These bounds are the $\ell^\infty$ analogues of the $\ell^2$ ingredients used to
prove BAOAB contraction in \cite{leimkuhler2024contraction}. See Appendix~\ref{appendix-bounds-baoab}  for its proof.

\subsection{$\ell^\infty$ bounds under sparse interactions}
\label{sec-bound-sparse}
We now turn to the propagator bounds under sparse interactions. 
Recall from
Assumption~\ref{assumption-sparse-interactions} that each term $V_i$ depends only on the
coordinates in the neighborhood $\sfN(i)$. A single Hessian $\nabla^2 V$ therefore has a nonzero $(i,j)$ entry only when some
term $V_l$ involves both $i$ and $j$, that is, $i, j \in \sfN(l)$, so $i$ and $j$ are within two edges in the interaction graph. More generally, Lemma~\ref{lem-sparse-instantiation} in Appendix~\ref{appendix-bounds-leap-frog} shows that a product of $r$ Hessians, the basic building block of the matrix
polynomials, has its $(i,j)$ entry supported on $j \in \sfN_{2r}(i)$, with at most $s_r$ nonzero entries per
row. 

This locality is what yields the $\ell^\infty$ bounds. To bound a row of a matrix
polynomial, we fix a degree threshold $r$ and split the entries by neighborhood.
For $j \in \sfN_{2r}(i)$, there are at most $s_r$ such indices, and a
Cauchy--Schwarz estimate over them bounds their contribution by $\sqrt{s_r}$
times the $\ell^2$ norm. For $j \notin \sfN_{2r}(i)$,
the entry only arises from
products of degree above $r$, which carry high powers of the step size and
decay geometrically. Bounding $s_r$ together with this decay converts the
$\ell^2$ propagator estimates into $\ell^\infty$ ones and controls error
accumulation. The following proposition is proved %
in Appendix~\ref{appendix-bounds-leap-frog}.

\begin{proposition}[$\ell^\infty$ bounds under sparse interactions, HMC] 
\label{prop-linf-bounds-sparse-hmc}
Suppose Assumptions \ref{assumption-V-log-concave} and \ref{assumption-sparse-interactions} hold. Let the integrated Hessians $H_i$ be defined as  in Proposition \ref{prop-l2bounds-hmc}. Take $r_i = \lceil ih\sqrt\beta e + \frac{\log\sqrt{d}}{\log(5/3)}\rceil$. For $mh \leq 1/\sqrt{20\beta}$, integers $N \geq 0$ and $k < m$, we have
\begin{equation}
\label{eqn-prop-linf-bounds-sparse-hmc-position}
    \Big|\Big(\prod_{j \in\{k+im\}_{i=0}^{N-1}}^{\longleftarrow}\tilde T_m^1(\{I - \frac{h^2}{2}H_{j+l} \}_{l=0}^{m-1})\Big)\tilde T_k^1(\{I-\frac{h^2} {2}H_i\}_{i=0}^{k-1})\Big|_{\ell^\infty} \leq 4 \sqrt{s_{r_{Nm+k}}}\,,
\end{equation}
\begin{equation}
    \Big|\Big(\prod_{j \in\{k+im\}_{i=0}^{N-1}}^{\longleftarrow}\tilde T_m^1(\{I - \frac{h^2}{2}H_{j+l} \}_{l=0}^{m-1})\Big) h\tilde U_{k-1}^1(\{I-\frac{h^2} {2}H_i\}_{i=1}^{k-1})\Big|_{\ell^\infty} \leq \frac{4}{\sqrt\beta} \sqrt{s_{r_{Nm+k}}}\,.
\end{equation}
\end{proposition}
For the BAOAB scheme, the weighted matrices relevant to the position projection have the following $\ell^\infty$ bounds. The proof is provided in Appendix~\ref{appendix-bounds-baoab}.
\begin{proposition}[$\ell^\infty$ bounds under sparse interactions, UL] 
\label{prop-linf-bounds-sparse-baoab}
Suppose Assumptions \ref{assumption-V-log-concave} and \ref{assumption-sparse-interactions} hold. Let the integrated Hessians $H_i$ be defined as  in Proposition \ref{prop-l2bounds-hmc}, and let $r_i = \lceil \frac{ie^2\beta}{1-\eta}h^2+\log(\sqrt{d}) \rceil$. 
Recall that $c(h) = \frac{\alpha h^2}{4(1-\eta)}$. For $h \leq \frac{1-\eta}{2 \sqrt{\beta}}$, we have 
\begin{equation}
    |\tilde T_k^\eta(\{I-\frac{h^2} {2}H_i\}_{i=0}^{k-1})|_{\ell^\infty} \leq (7\sqrt{2}+1)\sqrt{s_{r_k}}\exp(-\frac{k-1}{2}c(h))\,,
\end{equation}
\begin{equation}
\begin{aligned}
&\frac{1}{\sqrt{a-b^2}}\left|-b \tilde T_k^\eta(\{I - \frac{h^2}{2}H_i\}_{i=0}^{k-1}) + h \tilde U_{k-1}^\eta(\{I - \frac{h^2}{2}H_i\}_{i=1}^{k-1})\right|_{\ell^\infty} \\
    \leq &  (7\sqrt{2}+ \frac{2}{\sqrt{3}})\sqrt{s_{r_k}}\exp(-\frac{k-1}{2}c(h))\,.
\end{aligned}
\end{equation}
\end{proposition}
A contraction factor decaying exponentially in $N$ could also be
included in~\eqref{eqn-prop-linf-bounds-sparse-hmc-position} of
Proposition~\ref{prop-linf-bounds-sparse-hmc}, as in
Proposition~\ref{prop-linf-bounds-sparse-baoab}. 
We omit it, as our analysis does
not require it.
\section{Discretization Error and Sampling Bias}
\label{sec-discretization-error-main-text}
Besides the propagator bounds of Section \ref{sec-propagator-bounds}, the other critical component of the coupling framework in Section \ref{sec-sketch-techniques} is the discretization error. In this section
we bound it for HMC and UL, then combine it with the
propagator bounds 
to obtain the sampling bias.
\subsection{Discretization error in HMC}
To analyze the position and momentum errors separately, we split each solution map
into its position and momentum components, writing $\sfU_{\hmc}^t = (\sfQ^t, \sfP^t)$ for
the continuous Hamiltonian dynamics~\eqref{eqn-continuous-time-Hamiltonian}, and
$\sfU_{\hmc,h}^m = (\sfQ_h^m, \sfP_h^m)$ for $m$ leap-frog
steps~\eqref{eqn-leap-frog}. The single-step discretization error then splits as
\begin{equation}
\label{eqn-fine-discretization-error-hmc-decomp}
    (\sfU_{\hmc,h}^1 - \sfU_{\hmc}^h)(q^*, p^*)
    = \big((\sfQ_h^1-\sfQ^h)(q^*, p^*),\ (\sfP_h^1-\sfP^h)(q^*, p^*)\big)\,,
\end{equation}
which is the common building block of
\eqref{eqn-one-step-hmc-discretization-error-decomposition-propagation} and
\eqref{eqn-multi-step-hmc-discretization-error-decomposition-propagation}. The
following proposition bounds each component. We provide the proof to a generalized version in
Appendix~\ref{appendix-discretization-error}.
\begin{proposition}
\label{prop-discretization-error-hmc-main-text}
Suppose for any probability measures $\nu_1, \nu_2$ on $\mathbb{R}^d$, it holds that $| \int \nabla^2 V {\rm d} \nu_1|_{\ell^\infty} < C_3$ and $|(\int  \nabla^2 V {\rm d} \nu_1)(\int \nabla^2 V {\rm d} \nu_2)|_{\ell^\infty} < C_4$ for constants $C_3, C_4 > 0$ that may depend on $\beta$. Then for $h \leq 1/\sqrt{20\beta}$ and $(q^*, p^*) \sim \pi \otimes \mathcal{N}(0, I)$, we have  
\begin{equation}
\label{eqn-one-step-position-err-main-text}
|\sfQ_h^1(q^*,p^*)-\sfQ^h(q^*,p^*)|_{2,\ell^\infty} \leq  \frac1{\sqrt{2}}h^3C_3\sqrt{\log(2d)}\,,
\end{equation}
\begin{equation}
\label{eqn-one-step-momentum-err-main-text}
|\sfP_h^1(q^*,p^*)-\sfP^h(q^*,p^*)|_{2,\ell^\infty}\leq 2 h^2C_3 \sqrt{\log(2d)}  + \frac{1}{2\sqrt{5}}h^4 C_4 \sqrt{\log(2d)}\,.
\end{equation}
\end{proposition}

\subsection{Discretization error in UL} For the BAOAB integrator we adapt the framework of \cite{leimkuhler2024contraction}, comparing the BAOAB chain with an HOH chain for an accurate discretization error estimate even in the large-$\gamma$ limit. 
Our analysis differs mainly in measuring the error in the $\ell^\infty$ norm
rather than $\ell^2$.
The key component of
\eqref{eqn-multi-step-baoab-discretization-error-decomposition-propagation} is
$(\sfU_{\baoab,h}^l - \sfU_{\hoh,h}^{l})(q^*, p^*)$. Run $l$ steps of
each scheme from the \emph{same} initial state $(q^*,p^*)$ under the \emph{same}
O-step noise $\xi_0, \dots, \xi_{l-1}$:
\begin{equation}
    (x_i, p_i) = \sfU_{\hoh,h}^{\xi_{i-1}}(x_{i-1}, p_{i-1})\,, \qquad
    (y_i, p'_i) = \sfU_{\baoab,h}^{\xi_{i-1}}(y_{i-1}, p'_{i-1})\,,
\end{equation}
with $(x_0, p_0) = (y_0, p_0') = (q^*,p^*)$. Writing
$\Delta_q^i = x_i - y_i$ and $\Delta_p^i = p_i - p_i'$ for the accumulated
discrepancies, the difference above is exactly $(\Delta_q^l, \Delta_p^l)$. The
following proposition bounds it in the weighted norm. See
Appendix~\ref{appendix-discretization-error} for the proof.
\begin{proposition}
\label{prop-discretization-error-underdamped-main-text}    
Assume the shared initial state is $(q^*,p^*)\sim \pi \otimes \mathcal{N}(0,I)$ and suppose Assumption \ref{assumption-V-log-concave} holds. If for any probability measures $\nu_1, \nu_2$ on $\mathbb{R}^d$, we have $|\int \nabla^2 V {\rm d}\nu_1|_{\ell^\infty} \leq C_5$ and $|\left(\int \nabla^2 V {\rm d}\nu_1\right)\left(\int \nabla^2 V {\rm d}\nu_2\right)|_{\ell^\infty} \leq C_6$, and the step size satisfies $h \leq \frac{1-\eta}{2\sqrt\beta}$, then the accumulated discretization error across $l$ steps can be bounded by  
\begin{equation}
\label{eqn-prop-l-step-error-gronwall-main-text}
    |(\Delta_q^l, \Delta_p^l)|_{2,\ell^\infty_w} \leq \exp\left[(l-1)h(\frac32\frac{C_5}{\sqrt\beta}  + \frac{1}{16}\frac{C_6}{\beta^{3/2}} + \sqrt\beta)\right] E_l\,,
\end{equation}
where
\begin{equation}
    E_l = \left( \frac{27}{64}C_5\, h^3 l
    + {\frac{27}{8}}\frac{h^2 C_5}{\sqrt{\beta}(1-\eta)}
    + {\frac38}\frac{h^4 C_6}{\sqrt{\beta}(1-\eta)} \right) \sqrt{\log(2d)}\,.
\end{equation}
\end{proposition}

\subsection{Assembling the bias bounds}
Combining the discretization-error estimates above with the $\ell^\infty$
propagator bounds of Section~\ref{sec-propagator-bounds}, under the coupling
framework of Section~\ref{sec-sketch-techniques}, yields iterative inequalities
for HMC and its underdamped counterpart. Iterating these to the stationary limit
gives the bias bounds of Theorems~\ref{thm-weak-formal}
and~\ref{thm-sparse-formal}. The full assembly is carried out in
Appendix~\ref{appendix-sampling-bias-bounds}.
 
\section{Conclusions and Discussions} 
\label{sec-conclusion}
We study the delocalization of bias phenomenon in unadjusted Hamiltonian Monte
Carlo and underdamped Langevin algorithms. Although these algorithms are biased at
finite step size $h$, we show that the $W_2$ bias of low-dimensional marginals
scales with the marginal dimension rather than the full dimension $d$, when the
target potential has weak or sparse interactions, or is Gaussian. Consequently, {$O(\sqrt{K})$ integration steps} suffice for a bounded $W_2$ error across all
$K$-marginals, up to logarithmic dependence on $d$. This offers a counterpoint to common practice in
Bayesian statistics, where Metropolis--Hastings adjustment is incorporated to
eliminate discretization bias. The adjusted variants instead require a step size scaling with the full dimension $d$, leading to a large number of iterations.

{More broadly, we introduce a general matrix-polynomial framework for the
propagators of discretized Hamiltonian and Langevin dynamics. The polynomial
representation is well-suited to the $\ell^\infty$ norm, which is not amenable to
continuous-time tools such as differentiation. The
framework is also not specific to the leap-frog and BAOAB schemes studied here. It
applies to other splitting integrators, including OBABO and ABOBA, where, possibly
after regrouping the iterates, a single repeating core block satisfies a similar
three-term recurrence and the boundary factors contribute only bounded prefactors
that do not affect the coefficient estimates. We thus expect the framework to
facilitate the study of delocalization of bias in other dynamics as well. }

{Another interesting direction is randomized HMC~\cite{bou2017randomized,bou-rabee2025unadjustedhamiltonian}, where the integration time
is randomized to avoid periodicity. Our proof strategy does not directly apply there, as the number of
leap-frog steps per outer loop is random rather than fixed.}

We follow assumptions on the potential similar to
those in~\cite{chen2024convergence,lacker2025hierarchical}, and {relaxing them toward
conditions met by real physical systems is a significant future direction.} One such setting is non-log-concave target
distributions. Another is more general
interaction patterns, such as the decaying but long-range interactions of
physical systems rather than the sparse interactions considered here. These extensions are left for future work.

\vspace{0.2in}
\noindent {\bf Acknowledgments}
We thank Nawaf Bou-Rabee, Aaron Dinner, Daniel Lacker, Peter A.\ Whalley,  and Fuzhong Zhou for helpful discussions. Y.~Chen acknowledges the support from the UCLA DataX Pilot Projects Grant Program. J.~Niles-Weed is supported in part by National Science Foundation award DMS-2339829. J.~Weare and X.~Cheng are supported in part by National Science Foundation award DMS-2425899. X.~Cheng is also supported in part by a Dean's Dissertation Fellowship from the Graduate School of Arts and Science at New York University. AI tools were used only to edit language, flag typographical errors and minor inconsistencies, and generate the TikZ code for the figure.

  % branched from Michael Forbes' macro document
  % September 2014

  \newcommand{\cSTOC}[1]{\nth{\intcalcSub{#1}{1968}}\ Annual\ ACM\ Symposium\ on\ Theory\ of\ Computing\ (STOC)}
  \newcommand{\cFSTTCS}[1]{\nth{\intcalcSub{#1}{1980}}\ International\ Conference\ on\ Foundations\ of\ Software\ Technology\ and\ Theoretical\ Computer\ Science\ (FSTTCS)}
  \newcommand{\cCCC}[1]{\nth{\intcalcSub{#1}{1985}}\ Annual\ IEEE\ Conference\ on\ Computational\ Complexity\ (CCC)}
  \newcommand{\cFOCS}[1]{\nth{\intcalcSub{#1}{1959}}\ Annual\ IEEE\ Symposium\ on\ Foundations\ of\ Computer\ Science\ (FOCS)}
  \newcommand{\cRANDOM}[1]{\nth{\intcalcSub{#1}{1996}}\ International\ Workshop\ on\ Randomization\ and\ Computation\ (RANDOM)}
  \newcommand{\cISSAC}[1]{#1\ International\ Symposium\ on\ Symbolic\ and\ Algebraic\ Computation\ (ISSAC)}
  \newcommand{\cICALP}[1]{\nth{\intcalcSub{#1}{1973}}\ International\ Colloquium\ on\ Automata,\ Languages and\ Programming\ (ICALP)}
    \newcommand{\cCOLT}[1]{\nth{\intcalcSub{#1}{1987}}\ Conference\ on\ Learning\ Theory\ (COLT)}
  \newcommand{\cCSR}[1]{\nth{\intcalcSub{#1}{2005}}\ International\ Computer\ Science\ Symposium\ in\ Russia\ (CSR)}
  \newcommand{\cMFCS}[1]{\nth{\intcalcSub{#1}{1975}}\ International\ Symposium\ on\ the\ Mathematical\ Foundations\ of\ Computer\ Science\ (MFCS)}
  \newcommand{\cPODS}[1]{\nth{\intcalcSub{#1}{1981}}\ Symposium\ on\ Principles\ of\ Database\ Systems\ (PODS)}
  \newcommand{\cSODA}[1]{\nth{\intcalcSub{#1}{1989}}\ Annual\ ACM-SIAM\ Symposium\ on\ Discrete\ Algorithms\ (SODA)}
  \newcommand{\cNIPS}[1]{Advances\ in\ Neural\ Information\ Processing\ Systems\ \intcalcSub{#1}{1987} (NeurIPS)}
  \newcommand{\cWALCOM}[1]{\nth{\intcalcSub{#1}{2006}}\ International\ Workshop\ on\ Algorithms\ and\ Computation\ (WALCOM)}
  \newcommand{\cSoCG}[1]{\nth{\intcalcSub{#1}{1984}}\ Annual\ Symposium\ on\ Computational\ Geometry\ (SCG)}
  \newcommand{\cKDD}[1]{\nth{\intcalcSub{#1}{1994}}\ ACM\ SIGKDD\ International\ Conference\ on\ Knowledge\ Discovery\ and\ Data\ Mining\ (KDD)}
  \newcommand{\cICML}[1]{\nth{\intcalcSub{#1}{1983}}\ International\ Conference\ on\ Machine\ Learning\ (ICML)}
  \newcommand{\cAISTATS}[1]{\nth{\intcalcSub{#1}{1997}}\ International\ Conference\ on\ Artificial\ Intelligence\ and\ Statistics\ (AISTATS)}
  \newcommand{\cITCS}[1]{\nth{\intcalcSub{#1}{2009}}\ Conference\ on\ Innovations\ in\ Theoretical\ Computer\ Science\ (ITCS)}
  \newcommand{\cPODC}[1]{{#1}\ ACM\ Symposium\ on\ Principles\ of\ Distributed\ Computing\ (PODC)}
  \newcommand{\cAPPROX}[1]{\nth{\intcalcSub{#1}{1997}}\ International\ Workshop\ on\ Approximation\ Algorithms\ for\  Combinatorial\ Optimization\ Problems\ (APPROX)}
  \newcommand{\cSTACS}[1]{\nth{\intcalcSub{#1}{1983}}\ International\ Symposium\ on\ Theoretical\ Aspects\ of\  Computer\ Science\ (STACS)}
  \newcommand{\cMTNS}[1]{\nth{\intcalcSub{#1}{1991}}\ International\ Symposium\ on\ Mathematical\ Theory\ of\  Networks\ and\ Systems\ (MTNS)}
  \newcommand{\cICM}[1]{International\ Congress\ of\ Mathematicians\ {#1} (ICM)}
  \newcommand{\cWWW}[1]{\nth{\intcalcSub{#1}{1991}}\ International\ World\ Wide\ Web\ Conference\ (WWW)}
  \newcommand{\cICLR}[1]{\nth{\intcalcSub{#1}{2012}}\ International\ Conference\ on\ Learning\ Representations\ (ICLR)}
  \newcommand{\cICCV}[1]{\nth{\intcalcSub{#1}{1994}}\ IEEE\ International\ Conference\ on\ Computer\ Vision\ (ICCV)}
  \newcommand{\cICASSP}[1]{#1\ International\ Conference\ on\ Acoustics,\ Speech,\ and\ Signal\ Processing\ (ICASSP)}
  \newcommand{\cUAI}[1]{\nth{\intcalcSub{#1}{1984}}\ Annual\ Conference\ on\ Uncertainty\ in\ Artificial\ Intelligence\ (UAI)}
  \newcommand{\cSOSA}[1]{\nth{\intcalcSub{#1}{2017}}\ Symposium\ on\ Simplicity\ in\ Algorithms\ (SOSA)}
  \newcommand{\cISIT}[1]{#1\ IEEE\ International\ Symposium\ on\ Information\ Theory\ (ISIT)}  
  \newcommand{\cALT}[1]{\nth{\intcalcSub{#1}{1989}}\ International\ Conference\ on\ Algorithmic\ Learning\ Theory\ (ALT)} 
  \newcommand{\cWSDM}[1]{\nth{\intcalcSub{#1}{2007}}\ International\ Conference\ on\ Web\ Search\ and\ Data\ Mining\ (WSDM)} 
  \newcommand{\cICDM}[1]{#1\ IEEE\ International\ Conference\ on\ Data\ Mining\ (ICDM)}  
  \newcommand{\cEC}[1]{\nth{\intcalcSub{#1}{1999}}\ ACM\ Conference\ on\ Economics\ and\ Computation\ (EC)} 
  \newcommand{\cDISC}[1]{\nth{\intcalcSub{#1}{1986}}\ International\ Symposium\ on\ Distributed\ Computing\ (DISC)}
  \newcommand{\cSIGMOD}[1]{{#1} ACM SIGMOD International Conference on Management of Data}
  \newcommand{\cCIKM}[1]{\nth{\intcalcSub{#1}{1991}}\ ACM\ International\ Conference\ on\ Information\ and\ Knowledge\ Management\ (CIKM)}
 \newcommand{\cAAAI}[1]{AAAI\ Conference\ on\ Artificial\ Intelligence\ (AAAI)}
  \newcommand{\cICDE}[1]{\nth{\intcalcSub{#1}{1984}}\ IEEE\ International\ Conference\ on\ Data\ Engineering\ (ICDE)}  
  \newcommand{\cSPAA}[1]{\nth{\intcalcSub{#1}{1988}}\ ACM\ Symposium\ on\ Parallel\ Algorithms\ and\ Architectures\ (SPAA)}  
  \newcommand{\cESA}[1]{\nth{\intcalcSub{#1}{1992}}\ European\ Symposium\ on\ Algorithms\ (ESA)}  

  \newcommand{\pSTOC}[1]{Preliminary\ version\ in\ the\ \cSTOC{#1}}
  \newcommand{\pFSTTCS}[1]{Preliminary\ version\ in\ the\ \cFSTTCS{#1}}
  \newcommand{\pCCC}[1]{Preliminary\ version\ in\ the\ \cCCC{#1}}
  \newcommand{\pFOCS}[1]{Preliminary\ version\ in\ the\ \cFOCS{#1}}
  \newcommand{\pRANDOM}[1]{Preliminary\ version\ in\ the\ \cRANDOM{#1}}
  \newcommand{\pISSAC}[1]{Preliminary\ version\ in\ the\ \cISSAC{#1}}
  \newcommand{\pICALP}[1]{Preliminary\ version\ in\ the\ \cICALP{#1}}
  \newcommand{\pCOLT}[1]{Preliminary\ version\ in\ the\ \cCOLT{#1}}
  \newcommand{\pCSR}[1]{Preliminary\ version\ in\ the\ \cCSR{#1}}
  \newcommand{\pMFCS}[1]{Preliminary\ version\ in\ the\ \cMFCS{#1}}
  \newcommand{\pPODS}[1]{Preliminary\ version\ in\ the\ \cPODS{#1}}
  \newcommand{\pSODA}[1]{Preliminary\ version\ in\ the\ \cSODA{#1}}
  \newcommand{\pNIPS}[1]{Preliminary\ version\ in\ \cNIPS{#1}}
  \newcommand{\pWALCOM}[1]{Preliminary\ version\ in\ the\ \cWALCOM{#1}}
  \newcommand{\pSoCG}[1]{Preliminary\ version\ in\ the\ \cSoCG{#1}}
  \newcommand{\pKDD}[1]{Preliminary\ version\ in\ the\ \cKDD{#1}}
  \newcommand{\pICML}[1]{Preliminary\ version\ in\ the\ \cICML{#1}}
  \newcommand{\pAISTATS}[1]{Preliminary\ version\ in\ the\ \cAISTATS{#1}}
  \newcommand{\pITCS}[1]{Preliminary\ version\ in\ the\ \cITCS{#1}}
  \newcommand{\pPODC}[1]{Preliminary\ version\ in\ the\ \cPODC{#1}}
  \newcommand{\pAPPROX}[1]{Preliminary\ version\ in\ the\ \cAPPROX{#1}}
  \newcommand{\pSTACS}[1]{Preliminary\ version\ in\ the\ \cSTACS{#1}}
  \newcommand{\pMTNS}[1]{Preliminary\ version\ in\ the\ \cMTNS{#1}}
  \newcommand{\pICM}[1]{Preliminary\ version\ in\ the\ \cICM{#1}}
  \newcommand{\pWWW}[1]{Preliminary\ version\ in\ the\ \cWWW{#1}}
  \newcommand{\pICLR}[1]{Preliminary\ version\ in\ the\ \cICLR{#1}}
  \newcommand{\pICCV}[1]{Preliminary\ version\ in\ the\ \cICCV{#1}}
  \newcommand{\pICASSP}[1]{Preliminary\ version\ in\ the\ \cICASSP{#1}}
  \newcommand{\pUAI}[1]{Preliminary\ version\ in\ the\ \cUAI{#1}, #1}
  \newcommand{\pSOSA}[1]{Preliminary\ version\ in\ the\ \cSOSA{#1}, #1}
  \newcommand{\pISIT}[1]{Preliminary\ version\ in\ the\ \cISIT{#1}}
  \newcommand{\pALT}[1]{Preliminary\ version\ in\ the\ \cALT{#1}}
  \newcommand{\pWSDM}[1]{Preliminary\ version\ in\ the\ \cWSDM{#1}}
  \newcommand{\pICDM}[1]{Preliminary\ version\ in\ the\ \cICDM{#1}}
  \newcommand{\pEC}[1]{Preliminary\ version\ in\ the\ \cEC{#1}}
  \newcommand{\pDISC}[1]{Preliminary\ version\ in\ the\ \cDISC{#1}, #1}
  \newcommand{\pSIGMOD}[1]{Preliminary\ version\ in\ the\ \cSIGMOD{#1}, #1}
  \newcommand{\pCIKM}[1]{Preliminary\ version\ in\ the\ \cCIKM{#1}, #1}
  \newcommand{\pAAAI}[1]{Preliminary\ version\ in\ the\ \cAAAI{#1}, #1}
  \newcommand{\pICDE}[1]{Preliminary\ version\ in\ the\ \cICDE{#1}, #1}
  \newcommand{\pSPAA}[1]{Preliminary\ version\ in\ the\ \cSPAA{#1}, #1}
    \newcommand{\pESA}[1]{Preliminary\ version\ in\ the\ \cESA{#1}, #1}

  \newcommand{\STOC}[1]{Proceedings\ of\ the\ \cSTOC{#1}}
  \newcommand{\FSTTCS}[1]{Proceedings\ of\ the\ \cFSTTCS{#1}}
  \newcommand{\CCC}[1]{Proceedings\ of\ the\ \cCCC{#1}}
  \newcommand{\FOCS}[1]{Proceedings\ of\ the\ \cFOCS{#1}}
  \newcommand{\RANDOM}[1]{Proceedings\ of\ the\ \cRANDOM{#1}}
  \newcommand{\ISSAC}[1]{Proceedings\ of\ the\ \cISSAC{#1}}
  \newcommand{\ICALP}[1]{Proceedings\ of\ the\ \cICALP{#1}}
  \newcommand{\COLT}[1]{Proceedings\ of\ the\ \cCOLT{#1}}
  \newcommand{\CSR}[1]{Proceedings\ of\ the\ \cCSR{#1}}
  \newcommand{\MFCS}[1]{Proceedings\ of\ the\ \cMFCS{#1}}
  \newcommand{\PODS}[1]{Proceedings\ of\ the\ \cPODS{#1}}
  \newcommand{\SODA}[1]{Proceedings\ of\ the\ \cSODA{#1}}
  \newcommand{\NIPS}[1]{\cNIPS{#1}}
  \newcommand{\WALCOM}[1]{Proceedings\ of\ the\ \cWALCOM{#1}}
  \newcommand{\SoCG}[1]{Proceedings\ of\ the\ \cSoCG{#1}}
  \newcommand{\KDD}[1]{Proceedings\ of\ the\ \cKDD{#1}}
  \newcommand{\ICML}[1]{Proceedings\ of\ the\ \cICML{#1}}
  \newcommand{\AISTATS}[1]{Proceedings\ of\ the\ \cAISTATS{#1}}
  \newcommand{\ITCS}[1]{Proceedings\ of\ the\ \cITCS{#1}}
  \newcommand{\PODC}[1]{Proceedings\ of\ the\ \cPODC{#1}}
  \newcommand{\APPROX}[1]{Proceedings\ of\ the\ \cAPPROX{#1}}
  \newcommand{\STACS}[1]{Proceedings\ of\ the\ \cSTACS{#1}}
  \newcommand{\MTNS}[1]{Proceedings\ of\ the\ \cMTNS{#1}}
  \newcommand{\ICM}[1]{Proceedings\ of\ the\ \cICM{#1}}
  \newcommand{\WWW}[1]{Proceedings\ of\ the\ \cWWW{#1}}
  \newcommand{\ICLR}[1]{Proceedings\ of\ the\ \cICLR{#1}}
  \newcommand{\ICCV}[1]{Proceedings\ of\ the\ \cICCV{#1}}
  \newcommand{\ICASSP}[1]{Proceedings\ of\ the\ \cICASSP{#1}}
  \newcommand{\UAI}[1]{Proceedings\ of\ the\ \cUAI{#1}}
  \newcommand{\SOSA}[1]{Proceedings\ of\ the\ \cSOSA{#1}}
  \newcommand{\ISIT}[1]{Proceedings\ of\ the\ \cISIT{#1}}
  \newcommand{\ALT}[1]{Proceedings\ of\ the\ \cALT{#1}}
  \newcommand{\WSDM}[1]{Proceedings\ of\ the\ \cWSDM{#1}}
  \newcommand{\ICDM}[1]{Proceedings\ of\ the\ \cICDM{#1}}
  \newcommand{\EC}[1]{Proceedings\ of\ the\ \cEC{#1}}
  \newcommand{\DISC}[1]{Proceedings\ of\ the\ \cDISC{#1}}
  \newcommand{\SIGMOD}[1]{Proceedings\ of\ the\ \cSIGMOD{#1}}
  \newcommand{\CIKM}[1]{Proceedings\ of\ the\ \cCIKM{#1}}
  \newcommand{\AAAI}[1]{Proceedings\ of\ the\ \cAAAI{#1}}
  \newcommand{\ICDE}[1]{Proceedings\ of\ the\ \cICDE{#1}}
  \newcommand{\SPAA}[1]{Proceedings\ of\ the\ \cSPAA{#1}}
  \newcommand{\ESA}[1]{Proceedings\ of\ the\ \cESA{#1}}

  \newcommand{\arXiv}[1]{\href{http://arxiv.org/abs/#1}{arXiv:#1}}
  \newcommand{\farXiv}[1]{Full\ version\ at\ \arXiv{#1}}
  \newcommand{\parXiv}[1]{Preliminary\ version\ at\ \arXiv{#1}}
  \newcommand{\CoRR}{Computing\ Research\ Repository\ (CoRR)}

  \newcommand{\cECCC}[2]{\href{http://eccc.hpi-web.de/report/20#1/#2/}{Electronic\ Colloquium\ on\ Computational\ Complexity\ (ECCC),\ Technical\ Report\ TR#1-#2}}
  \newcommand{\ECCC}{Electronic\ Colloquium\ on\ Computational\ Complexity\ (ECCC)}
  \newcommand{\fECCC}[2]{Full\ version\ in\ the\ \cECCC{#1}{#2}}
  \newcommand{\pECCC}[2]{Preliminary\ version\ in\ the\ \cECCC{#1}{#2}}
\bibliographystyle{plain}
\bibliography{ref}
\vspace{2em}
\appendix
 The appendices contain all technical details. Section \ref{appendix-gaussian} gives the proof for Gaussian distributions. Section \ref{appendix-derivation-matrix-polynomial} derives the matrix polynomials. Sections \ref{appendix-bounds-leap-frog} and \ref{appendix-bounds-baoab} provide bounds for the propagators of the leap-frog and BAOAB integrators in $\ell^2$ and $\ell^\infty$. Discretization error analysis is laid out in Section \ref{appendix-discretization-error}. The sampling bias bounds are concluded in Section \ref{appendix-sampling-bias-bounds}.

\section{Proof for Gaussian Distributions}
\label{appendix-gaussian}
\begin{proof}[Proof of Example \ref{thm-Gaussian-hmc-main-text}]
{We assume $m$ is chosen so that $mh$ avoids the discrete set of resonant values at
which some eigenmode fails to contract. Away from these, the leap-frog iteration
contracts in every mode and the chain converges to $\pi_h$. As these values are
exceptional, we omit the condition from the example.}
According to \cite{apers2022hamiltonianmonte}, the invariant distribution for a target distribution $\pi=\mathcal{N}(\mu, \Sigma)$ is given by $\pi_h = \mathcal{N}(\mu, \Sigma_{\infty})$, where $  \Sigma_\infty^{-1} = \Sigma^{-1}(I- \frac{h^2}4 \Sigma^{-1})$.  
Consider the coupling between $\pi_h$ and $\pi$ as $X = \Sigma^{1/2}Z + \mu$ and $Y  = \Sigma_\infty^{1/2}Z+\mu$ with $Z \sim \mathcal{N}(0,I)$. Then,
    \begin{equation}
    W_{2,\ell^\infty}^2(\pi_h, \pi) \leq \mathbb{E}[|X-Y|_{\ell^\infty}^2] = \mathbb{E}[|(\Sigma^{1/2} - \Sigma_\infty^{1/2})Z|_{\ell^\infty}^2].
    \end{equation}
    Entries of $(\Sigma^{1/2}-\Sigma_\infty^{1/2})Z$ are $|\Sigma^{1/2}-\Sigma_\infty^{1/2}|_{\ell^2}^2$-subgaussian. And
    \begin{equation*}
    |\Sigma^{1/2}-\Sigma_\infty^{1/2}|_{\ell^2} = \max_{1\leq i \leq d}\left|\frac{1}{\omega_i} - \frac{1}{\hat \omega_i}\right| = O(\sqrt{\beta}h^2)\,.
    \end{equation*}
    By Lemma \ref{lem-Gaussian-linf}, we have $W_{2,\ell^\infty}(\pi_h, \pi) = O(\sqrt{\beta}h^2 \sqrt{\log(2d)})=O(h\sqrt{\log(2d)})$ as $h < 1/\sqrt{\beta}$.
\end{proof}
\section{Derivation of Matrix Polynomials}
\label{appendix-derivation-matrix-polynomial}
Throughout the appendices, when a multivariate damped Chebyshev polynomial is
evaluated at the canonical argument determined by its index, with $H_i$ as
specified in the relevant statement, we suppress the argument and write
\begin{equation*}
    \tilde T_k^\eta \coloneqq \tilde T_k^\eta\big(\{I - \tfrac{h^2}{2}H_i\}_{i=0}^{k-1}\big)\,,
    \qquad
    \tilde U_{k-1}^\eta \coloneqq \tilde U_{k-1}^\eta\big(\{I - \tfrac{h^2}{2}H_i\}_{i=1}^{k-1}\big)\,.
\end{equation*}
Arguments are written explicitly whenever they differ from this convention, as in
the shifted windows of the outer-loop blocks and in the scalar reductions below.
\begin{proof}[Proof of Proposition \ref{prop-matrix-polynomial}] 
\noindent\textbf{Part (a).} We derive the correspondence with multivariate damped matrix polynomials from three-term recurrence relations, similar to Proposition \ref{prop-hmc-Gaussian-propagators} but with varying Hessians.
With the same noise applied to both trajectories, the BAOAB
iteration~\eqref{eqn-baoab} for the differences
$\Delta q_k \coloneqq x_k - y_k$ and $\Delta p_k \coloneqq p_k - p'_k$ is
\begin{equation}
\label{eqn-inner-loop-difference}
\begin{aligned}
    \Delta q_{k+1} &= \Delta q_k + \tfrac{h}{2}(1+\eta)\Delta p_k - \tfrac{h^2}{4}(1+\eta)H_k \Delta q_k\,,\\
    \Delta p_{k+1} &= \eta\big(\Delta p_k - \tfrac{h}{2}H_k\Delta q_k\big) - \tfrac{h}{2}H_{k+1} \Delta q_{k+1}\,.
\end{aligned}
\end{equation}
The leap-frog iteration~\eqref{eqn-leap-frog} is the special case $\eta = 1$, so
it suffices to treat~\eqref{eqn-inner-loop-difference} for general $\eta$. We write down \eqref{eqn-inner-loop-difference} also for $(\Delta q_k, \Delta p_k)$ and eliminate the momentum $\Delta p_k$ in two consecutive steps by matching its coefficients. We then obtain
\begin{equation}
\label{eqn-three-term-recurrence-baoab-general}
    \Delta q_{k+1} = (1+\eta)(I - \tfrac{h^2}{2}H_k)\Delta q_k - \eta \Delta q_{k-1}\,,
    \qquad k \geq 1\,.
\end{equation}
At $\eta = 1$, the
recurrence recovers the relation for the leap-frog scheme,
$\Delta q_{k+1} = (2I - h^2 H_k)\Delta q_k - \Delta q_{k-1}$.

We prove~\eqref{eqn-prop-matrix-polynomial-single-scheme} by induction. It holds
for $k=0,1$ by the initial conditions of $\tilde T_k^\eta$ and $\tilde U_k^\eta$
in Definition~\ref{def-multivariate-damped-Chebyshev}. For $k \geq 1$, the
induction hypothesis and~\eqref{eqn-three-term-recurrence-baoab-general} give
\begin{equation}
\begin{aligned}
    \Delta q_{k+1} &= \big((1+\eta)(I -\tfrac{h^2}{2}H_k)\tilde T_k^\eta - \eta \tilde T_{k-1}^\eta\big)\Delta q_0\\
    & \quad + \big((1+\eta)(I - \tfrac{h^2}{2}H_k)h\tilde U_{k-1}^\eta - \eta h\tilde U_{k-2}^\eta\big)\Delta p_0\,,
\end{aligned}
\end{equation}
and the three-term recurrences for $\tilde T_k^\eta$ and $\tilde U_k^\eta$
in~\eqref{eqn-scalar-multivar-damped-chebyshev-first-kind}
and~\eqref{eqn-scalar-multivar-damped-chebyshev-second-kind}
identify the right-hand side as~\eqref{eqn-prop-matrix-polynomial-single-scheme},
completing the induction.

    \noindent\textbf{Part (b).} We iterate the single-scheme identity at $\eta = 1$ across the HMC outer loops.
For $N \geq 1$, since both trajectories share the refreshment $\xi_{N-1}$, the
identity gives
\begin{equation}
\begin{aligned}
    \Delta q_{Nm+k}
    &= \tilde T_m^1\big(\{I -\tfrac{h^2}{2}H_i\}_{i=(N-1)m+k}^{Nm+k-1}\big)\,\Delta q_{(N-1)m + k}\\
    &= \prod_{j \in\{k + im\}_{i=0}^{N-1}}^{\longleftarrow}
        \tilde T_m^1\big(\{I - \tfrac{h^2}{2}H_l\}_{l=j}^{j+m-1}\big)\,\Delta q_k\,,
\end{aligned}
\end{equation}
and the conclusion follows by substituting the expression for $\Delta q_k$
from~\eqref{eqn-prop-matrix-polynomial-single-scheme}.
\end{proof}
A property of these matrix polynomials, which we use when bounding the
propagators, is that their coefficients have alternating signs depending only on
the total order of each term.
\begin{lemma}
\label{lem-matrix-coefficients-sign}
For any $\eta \in [0,1]$, the multivariate damped Chebyshev matrix polynomials
$\tilde T_k^\eta$ and $\tilde U_{k}^\eta$ expand as
\begin{equation}
\label{eqn-first-kind-expansion}
    \tilde T_k^\eta = t^{(k)}_0 I + \sum_{j=1}^{k}(-1)^j h^{2j}\sum_{k-1\geq i_1 > \dots > i_j \geq 0} \tilde t_{i_1,\dots,i_j}^{(k)} H_{i_1}\cdots H_{i_j}\,,
\end{equation}
\begin{equation}
\label{eqn-second-kind-expansion}
    \tilde U_{k}^\eta = u_0^{(k)}I + \sum_{j=1}^{k} (-1)^j h^{2j}\sum_{k\geq i_1>\dots>i_j\geq1}\tilde u_{i_1,\dots,i_j}^{(k)}H_{i_1}\cdots H_{i_j}\,,
\end{equation}
with $\tilde U_{-1}^\eta \equiv 0$. For $k \geq 0$,
\begin{equation}
    t_0^{(k)} = 1\,, \qquad u_0^{(k)} = \tfrac12(1+\eta)\frac{\eta^{k+1} - 1}{\eta - 1}\,,
\end{equation}
and, for $k \geq i_1+1$,
\begin{equation}
    \tilde t_{i_1,\dots,i_j}^{(k)} = \frac{\eta^{k-i_1}-1}{\eta-1}\,\tilde t_{i_1,\dots,i_{j}}^{(i_1+1)}\,, \qquad
    \tilde u_{i_1,\dots i_j}^{(k-1)} = \frac{\eta^{k-i_1}-1}{\eta-1}\,\tilde u_{i_1,\dots,i_j}^{(i_1)}\,,
\end{equation}
where for $\eta = 1$ the quotients are read as the limit
$\lim_{\eta \to 1}(\eta^k - 1)/(\eta-1)= k$. All of these coefficients are
nonnegative:
\begin{equation}
    t_0^{(k)},\ u_0^{(k)},\ \tilde t_{i_1,\dots,i_j}^{(k)},\ \tilde u_{i_1,\dots,i_j}^{(k)} \geq 0\,.
\end{equation}
\end{lemma}
\begin{proof}[Proof of Lemma \ref{lem-matrix-coefficients-sign}]
The definitions of $\tilde T_k^\eta$ and $\tilde U_k^\eta$
in~\eqref{eqn-scalar-multivar-damped-chebyshev-first-kind}
and~\eqref{eqn-scalar-multivar-damped-chebyshev-second-kind} give the matrix
recurrences
\begin{align}
\label{eqn-matrix-polynomial-first-kind-recurrence}
    \tilde T_{k+1}^\eta &= (1+\eta)(I - \tfrac{h^2}{2}H_{k})\tilde T_k^\eta - \eta \tilde T_{k-1}^\eta\,,\\
\label{eqn-matrix-polynomial-second-kind-recurrence}
    \tilde U_{k+1}^\eta &= (1+\eta)(I - \tfrac{h^2}{2}H_{k+1})\tilde U_k^\eta - \eta \tilde U_{k-1}^\eta\,.
\end{align}
We establish the expansions~\eqref{eqn-first-kind-expansion}
and~\eqref{eqn-second-kind-expansion} by induction, the cases $k=0,1$ following
from the initial conditions.
 
\noindent\textbf{Constant terms.} The constant coefficients satisfy
$t_0^{(k+1)} = (1+\eta) t_0^{(k)} - \eta t_0^{(k-1)}$ and
$u_0^{(k+1)} = (1+\eta) u_0^{(k)} - \eta u_0^{(k-1)}$, with $t_0^{(0)} = t_0^{(1)} = 1$,
$u_{0}^{(-1)} = 0$, and $u_0^{(0)} = (1+\eta)/2$. Solving these gives $t_0^{(k)} = 1$
and $u_0^{(k)} = \tfrac12 (1+\eta)\frac{\eta^{k+1}-1}{\eta-1}$, the latter read as
$k+1$ when $\eta = 1$.
 
\noindent\textbf{Higher-order terms.} Assume~\eqref{eqn-first-kind-expansion}
and~\eqref{eqn-second-kind-expansion} hold up to index $k$. Passing to $k+1$
via~\eqref{eqn-matrix-polynomial-first-kind-recurrence}
and~\eqref{eqn-matrix-polynomial-second-kind-recurrence}, the only new matrix products
arise from $-(1+\eta)\tfrac{h^2}{2}H_k\tilde T_k^\eta$ and
$-(1+\eta)\tfrac{h^2}{2}H_{k+1}\tilde U_k^\eta$, which prepend $H_k$ (resp.\
$H_{k+1}$) to the existing terms. Applied to the order-$j$ terms of
$\tilde T_k^\eta$ and $\tilde U_k^\eta$, this produces the order-$(j+1)$
contributions
\begin{equation}
    (-1)^{j+1}h^{2(j+1)}\tfrac{1+\eta}{2}\!\!\sum_{k-1\geq i_1>\dots>i_j \geq 0}\!\! \tilde t_{i_1,\dots,i_j}^{(k)} H_k H_{i_1}\cdots H_{i_j}\,,
\end{equation}
\begin{equation}
    (-1)^{j+1}h^{2(j+1)}\tfrac{1+\eta}{2}\!\!\sum_{k\geq i_1>\dots>i_j\geq1}\!\!\tilde u_{i_1,\dots,i_j}^{(k)}H_{k+1}H_{i_1}\cdots H_{i_j}\,,
\end{equation}
together with the new order-one terms $-(1+\eta)\tfrac{h^2}{2}t_0^{(k)}H_k$ and
$-(1+\eta)\tfrac{h^2}{2}u_0^{(k)}H_{k+1}$. In each new product the prepended
factor carries a strictly larger index than all of $i_1 > \dots > i_j$, so the
matrices remain in descending index order, and~\eqref{eqn-first-kind-expansion}
and~\eqref{eqn-second-kind-expansion} hold at $k+1$.
 
\noindent\textbf{Nonnegativity.} We match matrix products across the
recurrences and induct on $j$, the number of factors. The case $j=0$ was settled
above. A product $H_{i_1}\cdots H_{i_j}$ first appears in $\tilde T_{i_1+1}^\eta$
and $\tilde U_{i_1}^\eta$, with coefficients, {for $i_1\geq 1$,}
\begin{equation}
\label{eqn-coeff-first-emerge}
    \tilde t_{i_1,\dots,i_j}^{(i_1+1)} = \tfrac{1+\eta}{2}\,\tilde t_{i_2,\dots,i_j}^{(i_1)}\,, \qquad
    \tilde u_{i_1,\dots,i_j}^{(i_1)} = \tfrac{1+\eta}{2}\,\tilde u_{i_2,\dots,i_j}^{(i_1-1)}\,,
\end{equation}
where for $j=1$ the right-hand factors are read as the constant coefficients
$t_0^{(i_1)}$ and $u_0^{(i_1-1)}$. {Besides the induction, we also note that the first-kind base case is $\tilde t_0^{(1)}=\frac{1+\eta}{4} \geq 0$, the coefficient of $-h^2H_0$ in $\tilde T_1^\eta$.} By the induction hypothesis on $j-1$, these
initial coefficients are nonnegative.  Tracking the coefficient of this same
product $H_{i_1}\cdots H_{i_j}$ in the higher-index polynomials, it satisfies
\begin{equation}
    \tilde t_{i_1,\dots,i_j}^{(i_1+2)} = (1+\eta)\tilde t_{i_1,\dots,i_j}^{(i_1+1)}\,, \qquad
    \tilde u_{i_1,\dots,i_j}^{(i_1+1)} = (1+\eta)\tilde u_{i_1,\dots,i_j}^{(i_1)}\,,
\end{equation}
and, for $k \geq i_1+2$,
\begin{equation}
    \tilde t_{i_1,\dots,i_j}^{(k+1)} = (1+\eta)\tilde t_{i_1,\dots,i_j}^{(k)} - \eta\tilde t_{i_1,\dots,i_j}^{(k-1)}\,, \qquad
    \tilde u_{i_1,\dots,i_j}^{(k)} = (1+\eta)\tilde u_{i_1,\dots,i_j}^{(k-1)} - \eta\tilde u_{i_1,\dots,i_j}^{(k-2)}\,.
\end{equation}
Solving these recurrences gives, for $k \geq i_1+1$,
\begin{equation}
    \tilde t_{i_1,\dots,i_j}^{(k)} = \frac{\eta^{k-i_1}-1}{\eta-1}\,\tilde t_{i_1,\dots,i_j}^{(i_1+1)}\,, \qquad
    \tilde u_{i_1,\dots,i_j}^{(k-1)} = \frac{\eta^{k-i_1}-1}{\eta-1}\,\tilde u_{i_1,\dots,i_j}^{(i_1)}\,,
\end{equation}
which are nonnegative, with the $\eta=1$ values read as limits as before.
\end{proof}
To bound the propagators we need the magnitude of these coefficients. Since all
coefficients of a given order share one sign, it suffices to track their
aggregate sums
\begin{equation}
    t_j^{(k)} = \sum_{k-1\geq i_1>\dots>i_j \geq 0} \tilde t_{i_1,\dots,i_j}^{(k)}\,, \qquad
    u_j^{(k)}=\sum_{k \geq i_1>\dots>i_j \geq 1}\tilde u_{i_1,\dots,i_j}^{(k)}\,.
\end{equation}
We characterize these aggregate coefficients in two regimes, exactly when
$\eta = 1$, and by upper bounds when $\eta \in [0,1)$.
\begin{proposition}
\label{prop-matrix-polynomial-agg-coeff}
When $\eta = 1$,
\begin{equation}
    t_j^{(k)} = k\,\frac{(k+j-1)!}{(k-j)!(2j)!}\,, \quad k\geq 1,\ 0\leq j \leq k\,,
\end{equation}
\begin{equation}
    u_j^{(k)} = \frac{(k+j+1)!}{(k-j)!(2j+1)!}\,, \quad k \geq 0,\ 0 \leq j \leq k\,,
\end{equation}
together with $t_0^{(0)} = 1$.
\end{proposition}
 
\begin{proof}[Proof of Proposition \ref{prop-matrix-polynomial-agg-coeff}]
Taking $H_i \coloneqq I$ in Lemma~\ref{lem-matrix-coefficients-sign} collapses the
expansions onto their aggregate coefficients,
\begin{equation}
\begin{aligned}
    \tilde T_k^\eta(\{(1- \tfrac{h^2}{2})I\}_{i=0}^{k-1}) &= t_0^{(k)} I + \sum_{j=1}^k (-1)^j h^{2j}t_j^{(k)}I\,, \\
    \tilde U_k^\eta(\{(1 - \tfrac{h^2}{2})I\}_{i=1}^{k}) &= u_0^{(k)}I + \sum_{j=1}^k (-1)^jh^{2j}u_j^{(k)}I\,.
\end{aligned}
\end{equation}
At $\eta = 1$ with equal arguments, the multivariate polynomials reduce to the
classical Chebyshev polynomials $T_k$ and $U_k$ (as noted after
Definition~\ref{def-multivariate-damped-Chebyshev}), so
\begin{equation}
\label{eqn-reduced-matrix-polynomial-coefficients}
    T_k(1 - \tfrac{h^2}{2})\,I = \sum_{j=0}^k(-1)^j h^{2j}t_j^{(k)}I\,, \qquad
    U_k(1-\tfrac{h^2}{2})\,I = \sum_{j=0}^k(-1)^jh^{2j}u_j^{(k)}I\,.
\end{equation}
The hypergeometric forms of $T_k$ and $U_k$~\cite{abramowitz1964handbook} give
explicit expressions for $T_k(1-x)$ and $U_k(1-x)$ at $x = h^2/2$, and matching
coefficients yields the stated formulas.
\end{proof}
For $\eta \in [0,1)$, explicit formulas are no longer available, but the
following upper bounds suffice to control the propagator norms.
\begin{proposition}
\label{prop-matrix-polynomial-agg-coeff-eta}
For $\eta \in [0, 1)$ and the same ranges of $j, k$ as in
Proposition~\ref{prop-matrix-polynomial-agg-coeff},
\begin{equation}
\label{eqn-prop-matrix-polynomial-agg-coeff-eta}
    t_j^{(k)} \leq \left(\frac{1}{1-\eta}\right)^j \binom{k}{j}\,, \quad u_j^{(k)} \leq \left(\frac{1}{1-\eta}\right)^{j+1}\binom{k}{j}\,.
\end{equation}
\end{proposition}
\begin{proof}[Proof of Proposition \ref{prop-matrix-polynomial-agg-coeff-eta}]
We derive iterative inequalities for the coefficients from their recurrences and
conclude by induction. We treat $t_j^{(k)}$, and the argument for $u_j^{(k)}$ is
analogous. For $j \geq 1$, Lemma~\ref{lem-matrix-coefficients-sign} gives
\begin{equation}
    t_j^{(k)} = \sum_{k-1 \geq i_1 > \dots > i_j \geq 0}\tilde t_{i_1,\dots,i_j}^{(k)}
    = \sum_{k-1\geq i_1 > \dots > i_j \geq 0}\frac{\eta^{k-i_1}-1}{\eta - 1}\,\tilde t_{i_1,\dots,i_j}^{(i_1+1)}\,.
\end{equation}
By~\eqref{eqn-coeff-first-emerge} and the initial conditions
$\tilde t_0^{(1)} = \frac{1+\eta}{4}$ (which is the coefficient of $-h^2H_0$) and $t_0^{(0)}= 1$, we have
$\tilde t_{i_1,\dots,i_j}^{(i_1+1)} \leq \tfrac12(1+\eta)\,\tilde t_{i_2,\dots,i_j}^{(i_1)}$,
so
\begin{equation}
\begin{aligned}
    t_j^{(k)} &\leq \tfrac12(1+\eta)\sum_{k-1\geq i_1\geq j-1}\frac{\eta^{k-i_1}-1}{\eta-1}\left(\sum_{i_1-1\geq i_2 > \dots > i_j \geq 0} \tilde t_{i_2,\dots,i_j}^{(i_1)}\right)\\
    & = \tfrac12(1+\eta)\sum_{k-1\geq i_1 \geq j-1}\frac{\eta^{k-i_1}-1}{\eta-1}\,t_{j-1}^{(i_1)}
    \leq \frac{1}{1-\eta}\sum_{k-1\geq i_1 \geq j-1} t_{j-1}^{(i_1)}\,,
\end{aligned}
\end{equation}
using the definition of $t_{j-1}^{(i_1)}$ in the equality. We now induct on $j$.
The case $j = 0$ holds since $t_0^{(k)} = 1$ for all $k$. Assuming
\eqref{eqn-prop-matrix-polynomial-agg-coeff-eta} for some $j$ and all $k \geq j$,
then for $k \geq j+1$,
\begin{equation}
    t_{j+1}^{(k)}\leq \frac{1}{1-\eta}\sum_{k-1\geq i_1 \geq j} \Big(\frac{1}{1-\eta}\Big)^j\binom{i_1}{j}
    = \Big(\frac{1}{1-\eta}\Big)^{j+1}\binom{k}{j+1}\,,
\end{equation}
where the last step uses Pascal's rule
$\sum_{k-1\geq i_1 \geq j}\binom{i_1}{j} = \binom{k}{j+1}$. This proves the bound
for $t_j^{(k)}$.
 
For $u_j^{(k)}$, the same steps give the iterative inequality
$u_j^{(k)} \leq \frac{1}{1-\eta}\sum_{k \geq i_1\geq j} u_{j-1}^{(i_1-1)}$.
Combined with the initial condition
$u_0^{(k)} = \tfrac12(1+\eta)\frac{\eta^{k+1}-1}{\eta-1}$ from
Lemma~\ref{lem-matrix-coefficients-sign} and the identity
$\sum_{k\geq i_1\geq j}\binom{i_1-1}{j-1} = \binom{k}{j}$, this yields the stated
bound.
\end{proof}
Part (b) of Proposition~\ref{prop-matrix-polynomial} concerns the HMC
propagators across multiple outer loops. The corresponding matrix polynomials
again have coefficients of alternating signs, related to the inner-loop
coefficients $t_j^{(k)}$ and $u_j^{(k)}$. We only need that the outer-loop
aggregate coefficients inherit any uniform bound on the inner-loop ones.
 
\begin{proposition}
\label{prop-matrix-polynomial-agg-coeff-outer}
When $\eta = 1$, the propagators
in~\eqref{eqn-propagator-multiple-outer-loop} expand as
\begin{equation}
\label{eqn-propagator-outerloop-sign-a}
\begin{aligned}
    &\Big(\prod_{j \in\{k + im\}_{i=0}^{N-1}}^{\longleftarrow} \tilde T_m^1(\{I - \tfrac{h^2}{2}H_l\}_{l=j}^{j+m-1})\Big)\, \tilde T_k^1\\
    &\qquad = a_0^{(Nm+k)}I + \sum_{j=1}^{Nm+k} (-1)^j h^{2j}\!\!\sum_{i_1>\dots>i_j}\!\! a_{i_1,\dots,i_j}^{(Nm+k)} H_{i_1}\cdots H_{i_j}\,,
\end{aligned}
\end{equation}
\begin{equation}
\label{eqn-propagator-outerloop-sign-b}
\begin{aligned}
    &\Big(\prod_{j \in\{k + im\}_{i=0}^{N-1}}^{\longleftarrow} \tilde T_m^1(\{I - \tfrac{h^2}{2}H_l\}_{l=j}^{j+m-1})\Big)\, \tilde U_{k-1}^1\\
    &\qquad = b_0^{(Nm+k-1)}I + \sum_{j=1}^{Nm+k-1} (-1)^j h^{2j}\!\!\sum_{i_1>\dots>i_j}\!\! b_{i_1,\dots,i_j}^{(Nm+k-1)} H_{i_1}\cdots H_{i_j}\,,
\end{aligned}
\end{equation}
with all coefficients nonnegative. Writing
$a_j^{(Nm+k)} = \sum_{i_1>\dots>i_j} a_{i_1,\dots,i_j}^{(Nm+k)}$ and
$b_j^{(Nm+k-1)} = \sum_{i_1>\dots>i_j} b_{i_1,\dots,i_j}^{(Nm+k-1)}$ for the
aggregate coefficients, we have
\begin{equation}
\label{eqn-outerloop-coeff-bound}
    a_j^{(Nm+k)} \leq \max_{j \leq i \leq Nm+k} t_j^{(i)}\,, \qquad
    b_j^{(Nm+k-1)} \leq \max_{j \leq i \leq Nm+k-1} u_j^{(i)}\,.
\end{equation}
\end{proposition}
 
\begin{proof}[Proof of Proposition~\ref{prop-matrix-polynomial-agg-coeff-outer}]
\noindent\textbf{Alternating signs.} In each factor $\tilde T_k^1$ and $\tilde U_k^1$ with
arguments $I - \tfrac{h^2}{2}H_i$, a product of $j$ matrices carries the sign
$(-1)^j$. This sign is multiplicative across the product of factors
in~\eqref{eqn-propagator-outerloop-sign-a}
and~\eqref{eqn-propagator-outerloop-sign-b}, so every order-$j$ product carries
$(-1)^j$, and the coefficients are nonnegative.
 
\noindent\textbf{Reduction to Chebyshev polynomials.} Taking $H_i \coloneqq I$ collapses the
expansions onto their aggregate coefficients, and at $\eta = 1$ the multivariate
polynomials reduce to the classical $T_k$ and $U_k$, giving
\begin{equation}
\begin{aligned}
    \big(T_m(1-\tfrac{h^2}{2})\big)^N T_k(1-\tfrac{h^2}{2})\,I
    &= a_0^{(Nm+k)}I + \sum_{j=1}^{Nm+k}(-1)^j h^{2j}a_j^{(Nm+k)}I\,,\\
    \big(T_m(1-\tfrac{h^2}{2})\big)^N U_{k-1}(1-\tfrac{h^2}{2})\,I
    &= b_0^{(Nm+k-1)}I + \sum_{j=1}^{Nm+k-1}(-1)^j h^{2j}b_j^{(Nm+k-1)}I\,.
\end{aligned}
\end{equation}
Writing $\theta = \arccos(1-h^2/2)$, so that
$T_k(1-\tfrac{h^2}{2}) = \cos(k\theta)$ and
$U_{k-1}(1-\tfrac{h^2}{2}) = \sin(k\theta)/\sin\theta$, and expanding
\begin{equation}
\label{eqn-cos-power-N}
    \big(\cos(m\theta)\big)^N = \frac{1}{2^N}\sum_{i=0}^N\binom{N}{i}\exp(\mathrm{i}(2i-N)m\theta)\,,
\end{equation}
we pair conjugate exponentials into cosines and apply
$2\cos x\cos y = \cos(x-y)+\cos(x+y)$ and
$2\cos x\sin y = \sin(x+y) - \sin(x-y)$. Rewriting the resulting trigonometric
sums back in terms of Chebyshev polynomials expresses
$\big(T_m\big)^N T_k$ as a binomially weighted sum of $T_{\,\cdot}$ at shifted
indices $2im \pm k$ (for $N$ even) or $(2i-1)m \pm k$ (for $N$ odd), all with
positive weights. For $\big(T_m\big)^N U_{k-1}$ the $\sin(x+y) - \sin(x-y)$
identity instead produces a signed sum of $U_{\,\cdot}$ at indices
$2im \pm k - 1$ (for $N$ even) or $(2i-1)m \pm k - 1$ (for $N$ odd), the
$x-y$ terms entering with a minus sign. Matching coefficients with
$t_j^{(\cdot)}$ and $u_j^{(\cdot)}$ then writes each $a_j^{(Nm+k)}$ and
$b_j^{(Nm+k-1)}$ as a $2^{-N}$-weighted sum of inner-loop aggregate coefficients
$t_j^{(i)}$ and $u_j^{(i)}$, with binomial weights and indicator factors
$\mathbf 1_{j \leq i}$. The sum for $a_j^{(Nm+k)}$ has all positive weights, while
that for $b_j^{(Nm+k-1)}$ inherits the minus signs above. Explicit formulas can
be read off in this way, but the main results need only the bound below.
 
\noindent\textbf{Bound.} Each $a_j^{(Nm+k)}$ is a sum of the form
$\tfrac{1}{2^N}\sum_i w_i\, t_j^{(i)}\mathbf 1_{j \leq i}$ with nonnegative
binomial weights $w_i$ satisfying $\sum_i w_i = 2^N$. Since the $t_j^{(i)}$ are
nonnegative by Lemma~\ref{lem-matrix-coefficients-sign}, bounding each term by
$t_j^{(i)}\mathbf 1_{j \leq i} \leq \max_{j \leq i \leq Nm+k} t_j^{(i)}$ gives
\begin{equation}
    a_j^{(Nm+k)} \leq \frac{1}{2^N}\Big(\sum_i w_i\Big)\max_{j \leq i \leq Nm+k} t_j^{(i)}
    = \max_{j \leq i \leq Nm+k} t_j^{(i)}\,.
\end{equation}
The same argument bounds $b_j^{(Nm+k-1)}$. Here some $u_j^{(i)}$ enter with
negative signs, but since the $u_j^{(i)}$ are nonnegative, discarding the
subtracted terms only increases the bound, and the remaining terms are controlled
as above by $\max_{j \leq i \leq Nm+k-1} u_j^{(i)}$.
\end{proof}

\section{Bounds for Leap-frog Propagators}
\label{appendix-bounds-leap-frog}
\subsection{$\ell^2$ bounds} 
Our $\ell^2$ bound for $\tilde T_k^1$ adapts the discrete-time HMC contraction
result for the position component, Lemma~19 in
\cite{bou-rabee203meanfield}, replacing the state-dependent Hessians there by
general matrices. The underlying argument carries through unchanged. Extending
this to the momentum component gives the bound for $\tilde U_{k-1}^1$.
 
The adaptation rests on decoupling the matrix sequence from the trajectory.
Because the leap-frog scheme is nonlinear in the state-dependent Hessians, a
trajectory contraction bound $|x_k - y_k|_{\ell^2} \leq c|x_0-y_0|_{\ell^2}$
does not by itself give a matrix-norm bound. We therefore analyze a linear system
driven by an arbitrary, state-independent matrix sequence $\{H_i\}$. The
contraction proof of~\cite{bou-rabee203meanfield} does not use that the Hessians
are evaluated along the trajectory, so the adaptation is direct. The resulting
bounds hold for any initial condition and hence translate to matrix-norm bounds.
Instantiating with the integrated Hessians yields
Proposition~\ref{prop-l2bounds-hmc}.

\begin{proposition}[$\ell^2$ bounds for leap-frog, general matrices]
\label{prop-l2bounds-hmc-general}
Let $\{H_i\}$ be symmetric positive definite matrices with
$\alpha I \preceq H_i \preceq \beta I$ for $\alpha, \beta$ as in
Assumption~\ref{assumption-V-log-concave}. For $kh \leq 1/\sqrt{20\beta}$, the
matrix polynomials $\tilde T_k^1$ and $\tilde U_{k-1}^1$, evaluated at the
canonical arguments $\{I - \tfrac{h^2}{2}H_i\}$ with these general $H_i$ (recall
the shorthand introduced in 
Section~\ref{appendix-derivation-matrix-polynomial}), satisfy
\begin{equation}
\label{eqn-l2bounds-hmc-Tk-general}
    |\tilde T_k^1|_{\ell^2} \leq 1 - \frac{\alpha}{10}(kh)^2\,,\qquad
    |h\tilde U_{k-1}^1|_{\ell^2} \leq \sqrt{\frac{2}{5\beta}}\,.
\end{equation}
\end{proposition}

\begin{proof}[Proof of Proposition~\ref{prop-l2bounds-hmc-general}]
We interpret the polynomials with general $\{H_i\}$ as the propagator of a leap-frog
iteration with \emph{varying} quadratic potentials, which lets us bound their
$\ell^2$ norms through the contraction and boundedness of the discrete-time
dynamics, analyzed via their continuous-time interpolation.
 
\noindent\textbf{Connection with dynamics.} Consider the time-inhomogeneous linear
dynamics driven by $\{H_i\}$,
\begin{equation}
\label{eqn-linear-inner-loop}
\begin{aligned}
    q_{k+1} &= q_k + hp_k - \tfrac{h^2}{2}H_k q_k\,,\\
    p_{k+1} &= p_k - \tfrac{h}{2}H_kq_k - \tfrac{h}{2}H_{k+1}q_{k+1}\,,
\end{aligned}
\end{equation}
which is the leap-frog scheme~\eqref{eqn-leap-frog} with the quadratic potential
$V_k(q) = \tfrac12 q^\top H_k q$ at step $k$. For two initial states $(x_0, p_0)$,
$(y_0, p_0')$, the differences $\Delta q_k \coloneqq x_k - y_k$,
$\Delta p_k \coloneqq p_k - p'_k$ satisfy
\begin{equation}
\label{eqn-linear-inner-loop-difference}
\begin{aligned}
    \Delta q_{k+1} &= (I- \tfrac{h^2}{2} H_k)\Delta q_k + h \Delta p_k\,,\\
    \Delta p_{k+1} &= \Delta p_k - \tfrac{h}{2}H_k \Delta q_k - \tfrac{h}{2}H_{k+1}\Delta q_{k+1}\,.
\end{aligned}
\end{equation}
This matches~\eqref{eqn-inner-loop-difference} in the proof of
Proposition~\ref{prop-matrix-polynomial}, now with general $\{H_i\}$ in place of
the integrated Hessians. As that proof does not use the specific Hessian
structure, the same argument gives the three-term recurrence
$\Delta q_{k+1} = (2I - h^2 H_k)\Delta q_k - \Delta q_{k-1}$ for $k \geq 1$ and,
by induction,
\begin{equation}
\label{eqn-prop-proof-q-dependence-q0-p0}
    \Delta q_{k} = \tilde T_{k}^1\,\Delta q_0 + h\tilde U_{k-1}^1\,\Delta p_0\,.
\end{equation}
Thus $\tilde T_k^1$ and $h\tilde U_{k-1}^1$ are the propagators
of~\eqref{eqn-linear-inner-loop}. We bound their $\ell^2$ norms by choosing
initial conditions that isolate each, setting $\Delta p_0 = 0$ for $\tilde T_k^1$, and
$\Delta q_0 = 0$ for $h\tilde U_{k-1}^1$.
 
\noindent\textbf{$\ell^2$ norm of $\tilde T_k^1$.} Adapting the proof of
\cite[Lemma~19]{bou-rabee203meanfield} to a fixed matrix sequence $\{H_i\}$
rather than the leap-frog scheme of the target potential, we show that the position
difference still contracts when the potential is decoupled from the trajectory.
We include it both for completeness and to set up the $\tilde U_k^1$ bound.
Introduce the continuous-time interpolation of~\eqref{eqn-linear-inner-loop},
\begin{equation}
\label{eqn-linear-inner-loop-interpolation}
\begin{aligned}
    \frac{{\rm d} q_t}{{\rm d }t} &= p_{\lfloor t \rfloor_h} - (t - \lfloor t \rfloor_h)H_{\lfloor t \rfloor_h}q_{\lfloor t \rfloor_h}\,,\\
    \frac{{\rm d}p_t}{{\rm d}t} &= -\tfrac12(H_{\lfloor t \rfloor_h} q_{\lfloor t \rfloor_h} + H_{\lceil t \rceil_h} q_{\lceil t \rceil_h})\,,
\end{aligned}
\end{equation}
with $\lfloor t \rfloor_h \coloneqq \lfloor t/h\rfloor h$ and
$\lceil t \rceil_h \coloneqq \lceil t/h \rceil h$. Take two trajectories from
$(x_0, p_0)$ and $(y_0, p_0')$ with $p_0 = p_0'$, set $z_t \coloneqq x_t - y_t$,
$w_t \coloneqq p_t - p_t'$, and at each grid point $t_k = kh$ define
$\Phi_{kh} \coloneqq H_{k} z_{kh}$ and $\varphi_{kh} \coloneqq \Phi_{kh}\cdot z_{kh}$,
with $\cdot$ the inner product. Then $w_0 = 0$ and
$|\Phi_t|_{\ell^2}^2 \leq \beta \varphi_t$, and~\eqref{eqn-linear-inner-loop-interpolation}
gives
\begin{equation}
\label{eqn-linear-inner-loop-interpolation-difference}
\begin{aligned}
    \frac{{\rm d}z_t}{{\rm d}t} &= w_{\lfloor t \rfloor_h} - (t - \lfloor t \rfloor_h) \Phi_{\lfloor t \rfloor_h}\,,\\
    \frac{{\rm d}w_t}{{\rm d}t} &= -\tfrac12(\Phi_{\lfloor t \rfloor_h}+\Phi_{\lceil t \rceil_h})\,.
\end{aligned}
\end{equation}
Let $a_t = |z_t|_{\ell^2}^2$ and $b_t = 2 z_t \cdot w_t$, so $b_0 = 0$ and
$\alpha a_{kh} \leq \varphi_{kh} \leq \beta a_{kh}$ since
$\alpha I \preceq H_k \preceq \beta I$. We track $a_t, b_t$ to bound $a_t$.
From~\eqref{eqn-linear-inner-loop-interpolation-difference},
\begin{equation}
\label{eqn-linear-inner-loop-interpolation-ab}
\begin{aligned}
    \frac{{\rm d} a_t}{{\rm d}t} &= b_t + \delta_t\,,\\
    \frac{{\rm d}b_t}{{\rm d}t} &= - C_1\alpha a_t + \eta_t\,,
\end{aligned}
\end{equation}
where $C_1>0$ is chosen later and the perturbations are
\begin{equation}
\begin{aligned}
    \delta_t &= 2z_t \cdot(w_{\lfloor t \rfloor_h}-w_t - (t - \lfloor t \rfloor_h)\Phi_{\lfloor t \rfloor_h}) = (t - \lfloor t \rfloor_h)z_t \cdot (\Phi_{\lceil t\rceil_h} - \Phi_{\lfloor t \rfloor_h})\,,\\
    \eta_t &= C_1\alpha a_t + 2|w_{\lfloor t \rfloor_h} - (t - \lfloor t \rfloor_h)\Phi_{\lfloor t \rfloor_h}|_{\ell^2}^2 - z_t \cdot(\Phi_{\lfloor t \rfloor_h} + \Phi_{\lceil t \rceil_h})\\
    &\quad -(t - \lfloor t \rfloor_h)(w_{\lfloor t \rfloor_h} - (t - \lfloor t \rfloor_h)\Phi_{\lfloor t \rfloor_h})\cdot(\Phi_{\lceil t \rceil_h}-\Phi_{\lfloor t \rfloor_h})\,.
\end{aligned}
\end{equation}
The term $\delta_t$ is piecewise smooth, with
\begin{equation*}
    \frac{{\rm d} \delta_t}{{\rm d}t} = z_t \cdot(\Phi_{\lceil t \rceil_h}-\Phi_{\lfloor t\rfloor_h})+(t - \lfloor t \rfloor_h)(w_{\lfloor t \rfloor_h} - (t - \lfloor t\rfloor_h)\Phi_{\lfloor t\rfloor_h})\cdot(\Phi_{\lceil t\rceil_h}-\Phi_{\lfloor t\rfloor_h})
\end{equation*}
in the interior of each step, and at grid points $\delta_{kh^+} = 0$ and
\begin{equation*}
    \delta_{kh^-} = h(\varphi_{kh} - \varphi_{(k-1)h}) - h(hw_{(k-1)h}-\tfrac12h^2\Phi_{(k-1)h})\cdot \Phi_{(k-1)h}\,.
\end{equation*}
Solving the perturbed system~\eqref{eqn-linear-inner-loop-interpolation-ab} by
variation of parameters, with $s_{t} = \frac{\sin(\sqrt{C_1\alpha}t)}{\sqrt{C_1\alpha}}$
and $c_t = \cos(\sqrt{C_1\alpha}t)$ (so $s_t$ is increasing and nonnegative on
$0\leq t \leq \frac{\pi}{2\sqrt{C_1\alpha}}$), gives
\begin{equation}
\label{eqn-solution-a}
\begin{aligned}
    a_t & = c_t a_0 + \int_0^t c_{t-r}\delta_r {\rm d}r + \int_0^t s_{t-r}\eta_r {\rm d}r\\
    & = c_t a_0 + \sum_{k=1}^{\lceil t/h\rceil-1}[s_{t-r}\delta_r]|^{r=kh^+}_{r=kh^-} - [s_{t-r}\delta_r]|^{r=t^-}_{r=0^+} + \int_0^t s_{t-r}(\eta_r + \tfrac{{\rm d}\delta_r}{{\rm d}r}) {\rm d} r\\
    & = c_t a_0 -\sum_{k=1}^{\lceil t/h\rceil -1}s_{t-kh}\delta_{kh^-}+\int_0^t s_{t-r}(-2\varphi_{\lfloor r\rfloor_h}+C_1\alpha a_{\lfloor r \rfloor_h} + \varepsilon_r) {\rm d} r\,,
\end{aligned}
\end{equation}
where we integrated by parts and set $\varepsilon_t \coloneqq \varepsilon_t^1 + \varepsilon_t^2 + \varepsilon_t^3$,
\begin{equation*}
\begin{aligned}
    \varepsilon_t^1 &= 2|w_{\lfloor t \rfloor_h} - (t - \lfloor t \rfloor_h)\Phi_{\lfloor t \rfloor_h}|_{\ell^2}^2\,,\\
    \varepsilon_t^2 &= -2(z_t - z_{\lfloor t \rfloor_h})\cdot \Phi_{\lfloor t \rfloor_h}\,,\\
    \varepsilon_t^3 &= C_1\alpha(z_t - z_{\lfloor t \rfloor_h}) \cdot (z_t + z_{\lfloor t \rfloor_h})\,.
\end{aligned}
\end{equation*}
We aim to show $a_t \leq c_t a_0$ for a suitable $C_1$, which requires bounding
the remaining terms in~\eqref{eqn-solution-a}. Using
$|\Phi_{\lfloor t \rfloor_h}|_{\ell^2}^2 \leq \beta \varphi_{\lfloor t \rfloor_h}$,
\begin{equation}
    \varepsilon_t^1 \leq 6|w_{\lfloor t \rfloor_h}|_{\ell^2}^2 + 3(t-\lfloor t \rfloor_h)^2|\Phi_{\lfloor t \rfloor_h}|_{\ell^2}^2 \leq 6|w_{\lfloor t \rfloor_h}|_{\ell^2}^2 + 3\beta h^2\varphi_{\lfloor t \rfloor_h}\,.
\end{equation}
Substituting $z_t - z_{\lfloor t \rfloor_h}$
from~\eqref{eqn-linear-inner-loop-interpolation-difference} into
$\varepsilon_t^2$,
\begin{equation}
\begin{aligned}
    \varepsilon_t^2 &= -2 (t - \lfloor t \rfloor_h) w_{\lfloor t \rfloor_h}\cdot \Phi_{\lfloor t \rfloor_h} + (t - \lfloor t \rfloor _h)^2|\Phi_{\lfloor t\rfloor_h}|_{\ell^2}^2\\
    & \leq 2|w_{\lfloor t \rfloor_h}|_{\ell^2}^2+ \tfrac32 h^2 |\Phi_{\lfloor t \rfloor_h}|_{\ell^2}^2\leq 2|w_{\lfloor t \rfloor_h}|_{\ell^2}^2 + \tfrac32\beta h^2\varphi_{\lfloor t \rfloor_h}\,,
\end{aligned}
\end{equation}
and similarly for $\varepsilon_t^3$,
\begin{equation}
\begin{aligned}
    \varepsilon_t^3 &= C_1 \alpha|(t-\lfloor t \rfloor_h)w_{\lfloor t \rfloor_h} - \tfrac12(t - \lfloor t \rfloor_h)^2 \Phi_{\lfloor t\rfloor_h}|_{\ell^2}^2\\
    & \quad +2 C_1\alpha z_{\lfloor t \rfloor_h} \cdot ((t - \lfloor t \rfloor_h) w_{\lfloor t \rfloor_h} - \tfrac12(t - \lfloor t \rfloor_h)^2\Phi_{\lfloor t\rfloor_h})\,.
\end{aligned}
\end{equation}
The first term is bounded by
$2 C_1\alpha h^2|w_{\lfloor t \rfloor_h}|_{\ell^2}^2 + \tfrac{1}{2}C_1\alpha h^4\beta \varphi_{\lfloor t\rfloor_h}$.
For the second, $|z_{\lfloor t \rfloor_h}|_{\ell^2}^2 = a_{\lfloor t \rfloor_h} \leq \varphi_{\lfloor t \rfloor_h}/\alpha$
gives
\begin{equation*}
    2 C_1\alpha z_{\lfloor t \rfloor_h} \cdot (t - \lfloor t \rfloor_h)w_{\lfloor t\rfloor_h} \leq h^2 C_1^2\alpha^2|z_{\lfloor t \rfloor_h}|_{\ell^2}^2 + |w_{\lfloor t \rfloor_h}|_{\ell^2}^2 \leq h^2C_1^2\alpha \varphi_{\lfloor t\rfloor_h}+|w_{\lfloor t \rfloor_h}|_{\ell^2}^2\,.
\end{equation*}
Since $z_{\lfloor t \rfloor_h} \cdot \Phi_{\lfloor t \rfloor_h} \geq 0$,
\begin{equation}
    \varepsilon_t^3 \leq (1 + 2C_1\alpha h^2)|w_{\lfloor t\rfloor_h}|_{\ell^2}^2 + C_1\alpha(C_1 h^2 + \tfrac12 \beta h^4)\varphi_{\lfloor t \rfloor_h}\,,
\end{equation}
and combining the three estimates,
\begin{equation}
\label{eqn-bound-varepsilont}
    \varepsilon_t \leq (9 + 2C_1\alpha h^2)|w_{\lfloor t\rfloor_h}|_{\ell^2}^2 + \left[ \tfrac92\beta h^2+C_1\alpha(C_1 h^2 + \tfrac12 \beta h^4) \right]\varphi_{\lfloor t \rfloor_h}\,.
\end{equation}
For the boundary sum in~\eqref{eqn-solution-a}, when $t/h\in \mathbb{Z}^+$,
summation by parts gives
\begin{equation}
\label{eqn-bound-boundary-terms-original}
\begin{aligned}
    -\sum_{k=1}^{t/h-1}s_{t-kh}\delta_{kh^-}
    = & \sum_{k=2}^{t/h-1}h(s_{t-kh} - s_{t-(k-1)h})\varphi_{(k-1)h}-h s_h\varphi_{t-h} + h s_{t-h}\varphi_0 \\
    & + \sum_{k=1}^{t/h-1}hs_{t-kh}(hw_{(k-1)h} -\tfrac12h^2\Phi_{(k-1)h})\cdot \Phi_{(k-1)h}\,.
\end{aligned}
\end{equation}
Since $s_{t-kh} \leq s_{t-r}$ for $r \in [(k-1)h, kh)$, $\varphi_{kh} \geq 0$, and
\begin{equation*}
    h^2 w_{(k-1)h} \cdot \Phi_{(k-1)h} \leq \tfrac12h|w_{(k-1)h}|_{\ell^2}^2 + \tfrac12h^3|\Phi_{(k-1)h}|_{\ell^2}^2\leq \tfrac12h|w_{(k-1)h}|_{\ell^2}^2 + \tfrac12\beta h^3\varphi_{(k-1)h}\,,
\end{equation*}
we obtain, for $t/h \in \mathbb{Z}^+$,
\begin{equation}
\label{eqn-bound-boundary-terms-same-momenta}
    -\sum_{k=1}^{t/h-1} s_{t-kh} \delta_{kh^-} \leq \int_0^h s_{t-r}\varphi_{\lfloor r \rfloor_h}{\rm d}r+ \tfrac12\int_0^t s_{t-r}(|w_{\lfloor r \rfloor_h}|_{\ell^2}^2 + \beta h^2 \varphi_{\lfloor r \rfloor_h}){\rm d} r\,.
\end{equation}
From~\eqref{eqn-linear-inner-loop-interpolation-difference} with $w_0=0$,
Cauchy--Schwarz and $|\Phi_{\lfloor \tau\rfloor_h}|_{\ell^2}^2 \leq \beta\varphi_{\lfloor \tau\rfloor_h}$
give
\begin{equation}
\label{eqn-w-integration-no-initial-bound}
    |w_{\lfloor t \rfloor_h}|_{\ell^2}^2 = \left|\tfrac12\int_0^{\lfloor t \rfloor_h}(\Phi_{\lfloor \tau \rfloor_h} + \Phi_{\lceil \tau \rceil_h}){\rm d}\tau\right|_{\ell^2}^2
    \leq \tfrac12\beta \lfloor t \rfloor_h \left(\int_0^{\lfloor t\rfloor_h} \varphi_{\lfloor \tau\rfloor_h} {\rm d} \tau + \int_0^{\lfloor t \rfloor_h} \varphi_{\lceil \tau \rceil_h}{\rm d}\tau\right)\,.
\end{equation}
{
To bound $\int_0^t s_{t-r}|w_{\lfloor r \rfloor_h}|_{\ell^2}^2{\rm d}r$, we note that for $t/h\in \mathbb{Z}^+$,
\begin{equation}
\label{eqn-sw-bound-prep}
    \int_0^t s_{t-r} \lfloor r \rfloor_h \left(\int_0^{\lfloor r \rfloor_h}\varphi_{\lceil \tau \rceil_h}{\rm d}\tau \right) {\rm d}r = \int_h^{t} \left(\int_{\lfloor \tau \rfloor_h}^{t}s_{t-r}\lfloor r \rfloor_h {\rm d}r\right) \varphi_{\lfloor \tau \rfloor_h}{\rm d}\tau\,.
\end{equation}
As $s_t$ is increasing, we decompose the integral in $r$ into sub-intervals of length $h$, and direct calculation of the sum of $\lfloor r \rfloor_h$ terms gives
\begin{equation}
   \int_{\lfloor\tau\rfloor_h}^{t} s_{t-r}\lfloor r\rfloor_h {\rm d}r \leq \left(\int_{\lfloor \tau \rfloor_h}^{\lfloor \tau \rfloor_h + h}s_{t-r}{\rm d}r\right) \frac{t^2}{2h}\,.
\end{equation}
Plugging into \eqref{eqn-sw-bound-prep} and switching the order of integration yield an upper bound
\begin{equation}
     \int_h^t \left(\int_{\lfloor \tau \rfloor_h}^{\lfloor \tau \rfloor_h + h}s_{t-r}{\rm d}r\right) \frac{t^2}{2h} \varphi_{\lfloor \tau\rfloor_h}{\rm d}\tau = \frac{t^2}{2}\int_h^{t} s_{t-r}\varphi_{\lfloor r\rfloor_h} {\rm d}r\,.
\end{equation}
We can similarly derive for the other term in \eqref{eqn-w-integration-no-initial-bound} that
\begin{equation}
    \int_0^ts_{t-r}\lfloor r \rfloor_h \left( \int_0^{\lfloor r \rfloor_h}\varphi_{\lfloor \tau \rfloor_h}{\rm d}\tau\right){\rm d}r \leq  \frac{t^2}{2}\int_0^{t-h}s_{t-r}\varphi_{\lfloor r\rfloor_h}{\rm d}r\,,
\end{equation}
and combining these two bounds yields
\begin{equation}
\label{eqn-sw-bound}
    \int_0^t s_{t-r}|w_{\lfloor r \rfloor_h}|_{\ell^2}^2 {\rm d}r \leq  \tfrac12\beta t^2\int_0^t s_{t-r}\varphi_{\lfloor r \rfloor_h}{\rm d}r\,.
\end{equation}
}
Hence the last two terms of~\eqref{eqn-solution-a} are bounded by
\begin{equation}
\begin{aligned}
    & -\sum_{k=1}^{\lceil t/h\rceil -1}s_{t-kh}\delta_{kh^-}+\int_0^t s_{t-r}(-2\varphi_{\lfloor r\rfloor_h}+C_1\alpha a_{\lfloor r \rfloor_h} + \varepsilon_r) {\rm d} r\\
    &\leq \int_0^h s_{t-r}\varphi_{\lfloor r\rfloor_h}{\rm d}r + \int_0^t s_{t-r}(-2\varphi_{\lfloor r\rfloor_h} + C_1\alpha  a_{\lfloor r\rfloor_h}) {\rm d} r\\
    & \quad +\int_0^t s_{t-r}((5\beta h^2 + C_1^2\alpha h^2 + \tfrac12 C_1 \alpha \beta h^4)\varphi_{\lfloor r\rfloor_h} + (\tfrac{19}{2} + 2C_1\alpha h^2)|w_{\lfloor r\rfloor_h}|_{\ell^2}^2){\rm d}r\,.
\end{aligned}
\end{equation}
Applying~\eqref{eqn-sw-bound} and $\alpha a_{kh} \leq \varphi_{kh}$,
\begin{equation}
\begin{aligned}
    & -\sum_{k=1}^{\lceil t/h\rceil -1}s_{t-kh}\delta_{kh^-}+\int_0^t s_{t-r}(-2\varphi_{\lfloor r\rfloor_h}+C_1\alpha a_{\lfloor r \rfloor_h} + \varepsilon_r) {\rm d} r\\
    \leq & \int_0^t s_{t-r}\left[-1 + C_1 + 5\beta h^2 + C_1^2\alpha h^2 + \tfrac12 C_1 \alpha \beta h^4 + (\tfrac{19}{4}+ C_1\alpha h^2)\beta t^2\right]\varphi_{\lfloor r \rfloor_h}{\rm d}r\\
    \leq & \int_0^t s_{t-r}\left[-1 + C_1 + (\tfrac{39}{4}+C_1^2 + \tfrac32C_1\beta t^2)\beta t^2 \right]{\varphi_{\lfloor r \rfloor_h}}{\rm d}r\,,
\end{aligned}
\end{equation}
using $t/h\in \mathbb{Z}^+$ and $\alpha \leq \beta$ in the last line. Taking
$C_1 = 4/9$ and $\beta t^2 \leq 1/20$ makes the bracketed coefficient negative,
so~\eqref{eqn-solution-a} gives $a_t \leq c_t a_0$. The same choice gives
$t \leq \pi/(2\sqrt{C_1\alpha})$, ensuring $s_t$ is monotone on the interval, and
\begin{equation*}
\begin{aligned}
    c_t &\leq 1 - \tfrac12C_1\alpha t^2 + \tfrac16C_1^2\alpha^2 t^4 \leq 1 - \tfrac12 C_1\alpha t^2 + \tfrac{1}{120}C_1^2\alpha t^2(\alpha/\beta)\\
    & \leq 1 - (\tfrac12 C_1 - \tfrac{1}{120}C_1^2)\alpha t^2 \leq 1 - \tfrac15\alpha t^2\,.
\end{aligned}
\end{equation*}
Thus for $t = kh$ with $p_0 = p_0'$,
$|x_t - y_t|_{\ell^2}^2 \leq (1 - \tfrac{\alpha}{5}(kh)^2)|x_0 - y_0|_{\ell^2}^2$,
and by~\eqref{eqn-prop-proof-q-dependence-q0-p0},
\begin{equation}
    |\tilde T_{k}^1|_{\ell^2} \leq 1 - \frac{\alpha}{10}(kh)^2\,.
\end{equation}
 
\noindent\textbf{$\ell^2$ norm of $h\tilde U_{k-1}^1$.}
By~\eqref{eqn-prop-proof-q-dependence-q0-p0}, bounding $|h\tilde U_{k-1}^1|_{\ell^2}$
again reduces to bounding $a_t$, now with initial conditions $z_0 = 0$,
$w_0 = p_0 - p_0'$, so $a_0 = b_0 = 0$ and
\begin{equation}
\label{eqn-solution-a-same-positions}
\begin{aligned}
    a_t &= \int_0^t c_{t-r}\delta_r {\rm d}r + \int_0^t s_{t-r}\eta_r {\rm d}r\\
    &= -\sum_{k=1}^{\lceil t/h\rceil -1}s_{t-kh}\delta_{kh^-}+\int_0^t s_{t-r}(-2\varphi_{\lfloor r\rfloor_h}+C_1\alpha a_{\lfloor r \rfloor_h} + \varepsilon_r) {\rm d} r\,.
\end{aligned}
\end{equation}
The boundary decomposition~\eqref{eqn-bound-boundary-terms-original} still holds,
and since now $\varphi_0=0$, it refines~\eqref{eqn-bound-boundary-terms-same-momenta}
to
\begin{equation}
    -\sum_{k=1}^{t/h-1}s_{t-kh}\delta_{kh^-} \leq \tfrac12 \int_0^t s_{t-r}(|w_{\lfloor r\rfloor_h}|_{\ell^2}^2 + \beta h^2 \varphi_{\lfloor r \rfloor_h}){\rm d}r\,.
\end{equation}
The bound~\eqref{eqn-bound-varepsilont} on $\varepsilon_r$ still applies, so
combining with $\alpha a_{\lfloor r \rfloor_h} \leq \varphi_{\lfloor r \rfloor_h}$,
\begin{equation}
\label{eqn-same-initial-pos-at-intermediate-bound}
\begin{aligned}
    a_t & \leq \int_0^t s_{t-r}(-2+C_1 +5\beta h^2 + C_1^2\alpha h^2 + \tfrac12 C_1\alpha \beta h^4)\varphi_{\lfloor r \rfloor_h}{\rm d}r\\
    & \quad + \int_0^t s_{t-r} (\tfrac{19}{2}+2C_1\alpha h^2)|w_{\lfloor r\rfloor_h}|_{\ell^2}^2 {\rm d}r\,.
\end{aligned}
\end{equation}
For $w_{\lfloor r \rfloor_h}$, repeating the derivation
of~\eqref{eqn-w-integration-no-initial-bound}
 but with general $w_0$, for any $C_2 > 0$,
\begin{equation}
\begin{aligned}
    |w_{\lfloor t \rfloor_h}|_{\ell^2}^2 &= \left| w_0 - \tfrac12 \int_0^{\lfloor t\rfloor_h}(\Phi_{\lfloor s\rfloor_h} + \Phi_{\lceil s\rceil_h}) {\rm d} s\right|_{\ell^2}^2\\
    %& \leq (1+C_2)|w_0|_{\ell^2}^2 + (1+\tfrac{1}{C_2})\left| \tfrac12\int_0^{\lfloor t \rfloor_h}(\Phi_{\lfloor s \rfloor_h}+\Phi_{\lceil s \rceil_h}){\rm d} s \right|_{\ell^2}^2\\
    &\leq (1+C_2)|w_0|_{\ell^2}^2 + (1+\tfrac{1}{C_2}) \tfrac12\beta \lfloor t \rfloor_h \left(\int_0^{\lfloor t\rfloor_h} \varphi_{\lfloor \tau\rfloor_h} + \varphi_{\lceil \tau \rceil_h}{\rm d} \tau \right)\,.
\end{aligned}
\end{equation}
With the derivation of~\eqref{eqn-sw-bound}, this bounds~\eqref{eqn-same-initial-pos-at-intermediate-bound}
by
\begin{equation}
\begin{aligned}
    a_t & \leq \int_0^t s_{t-r}\left[-2 + C_1 + 5\beta h^2 + C_1^2\beta h^2+ \tfrac12 C_1 \beta^2 h^4 \right.\\
    &\qquad\left.+ (1 + \tfrac{1}{C_2})(\tfrac{19}{4}+C_1\beta h^2)\beta t^2 \right]\varphi_{\lfloor r \rfloor_h}{\rm d}r + (\tfrac{19}{2}+2C_1\alpha h^2)(1+C_2)\int_0^t s_{t-r}|w_0|_{\ell^2}^2 {\rm d}r\,.
\end{aligned}
\end{equation}
With $C_1 = 4/9$, $\beta t^2 \leq 1/20$, and $C_2 = 1/4$, the coefficient of
$\varphi_{\lfloor r \rfloor_h}$ is negative. Since $\alpha h^2 \leq \beta t^2$
gives $\tfrac{19}{2} + 2C_1\alpha h^2 \leq 10$, for $t/h \in \mathbb{Z}^+$,
\begin{equation}
    a_t \leq 16 |w_0|_{\ell^2}^2 \int_0^t s_{t-r} {\rm d}r = \frac{16(1 - \cos(\sqrt{C_1\alpha}t))}{C_1\alpha} |w_0|_{\ell^2}^2 \leq 8t^2|w_0|_{\ell^2}^2\leq \frac{2}{5\beta}|w_0|_{\ell^2}^2\,,
\end{equation}
using $\cos t\geq 1-\tfrac12t^2$. By~\eqref{eqn-prop-proof-q-dependence-q0-p0},
\begin{equation}
    |h\tilde U_{k-1}^1|_{\ell^2} \leq \sqrt{\frac{2}{5\beta}}\,.
\end{equation}
\end{proof}
 \begin{proof}[Proof of Proposition~\ref{prop-l2bounds-hmc}]
This follows from Proposition~\ref{prop-l2bounds-hmc-general}, since the
integrated Hessian satisfies $\alpha I \preceq H_i \preceq \beta I$ by
Assumption~\ref{assumption-V-log-concave}.
\end{proof}
\subsection{$\ell^\infty$ bounds under weak interactions}

\begin{proof}[Proof of Proposition~\ref{prop-linftybounds-hmc-weak}]
Following the decomposition of $\nabla^2 V$ in
Assumption~\ref{assumption-weak-interactions}, set
\begin{equation*}
    H_i^{(D)} \coloneqq \int \nabla^2 V^{(D)} {\rm d} \nu_i\,,\qquad
    H_i^{(O)} \coloneqq \int \nabla^2 V^{(O)} {\rm d} \nu_i\,,
\end{equation*}
so that $H_i = H_i^{(D)}+H_i^{(O)}$. We begin with
$|\tilde T_k^1|_{\ell^\infty}$, which we split into a diagonal contribution and a
perturbation. By the sign structure of the coefficients in
Lemma~\ref{lem-matrix-coefficients-sign},
\begin{equation}
\label{eqn-propagator-infty-norm-decomp-hmc-weak}
\begin{aligned}
    |\tilde T_k^1|_{\ell^\infty} &\leq \underbrace{|\tilde T_k^1(\{I - \tfrac{h^2}{2}H_i^{(D)}\}_{i=0}^{k-1})|_{\ell^\infty}}_{(a)}\\
    &\quad +\underbrace{\sum_{j=1}^kh^{2j}\sum_{k-1 \geq i_1 > \dots > i_j \geq 0} \tilde t_{i_1,\dots,i_j}^{(k)} |H_{i_1}\cdots H_{i_j} - H_{i_1}^{(D)}\cdots H_{i_j}^{(D)}|_{\ell^\infty}}_{(b)}\,.
\end{aligned}
\end{equation}
The matrix in $(a)$ is diagonal, so its $\ell^\infty$ norm equals its $\ell^2$
norm. The $\ell^2$ bound of Proposition~\ref{prop-l2bounds-hmc} gives
\begin{equation}
    (a) = |\tilde T_k^1(\{I - \tfrac{h^2}{2}H_i^{(D)}\}_{i=0}^{k-1})|_{\ell^2} \leq 1 - {\frac{\alpha}{10}(kh)^2}\,.
\end{equation}
For $(b)$, expanding
$H_{i_1}\cdots H_{i_j} = (H_{i_1}^{(D)}+H_{i_1}^{(O)}) \cdots (H_{i_j}^{(D)}+H_{i_j}^{(O)})$
and bounding each term yield
\begin{equation}
\begin{aligned}
    |H_{i_1} \cdots H_{i_j} - H_{i_1}^{(D)}\cdots H_{i_j}^{(D)}|_{\ell^\infty}
    & \leq \sum_{l = 1}^j \binom{j}{l}(\sup_k |H_k^{(D)}|_{\ell^\infty})^{j-l}(\sup_k |H_k^{(O)}|_{\ell^\infty})^l\\
    & \leq (\beta + \tfrac{\alpha}{50})^j - \beta^j\,,
\end{aligned}
\end{equation}
so that
\begin{equation}
\label{eqn-bound-b-combine-binomials}
\begin{aligned}
    (b) &\leq \sum_{j=0}^k h^{2j} \sum_{k-1 \geq i_1 > \ldots > i_j \geq 0} \tilde t_{i_1, \ldots, i_j}^{(k)}\Big((\beta + \tfrac{\alpha}{50})^j - \beta^j\Big)\\
    &= \sum_{j=0}^k h^{2j} t_j^{(k)} \Big((\beta + \tfrac{\alpha}{50})^j - \beta^j\Big)
    = T_k(1 + \tfrac{h^2}{2}(\beta + \tfrac{\alpha}{50})) - T_k(1 + \tfrac{h^2}{2} \beta)\,,
\end{aligned}
\end{equation}
using that the $t_j^{(k)}$ are the coefficients of
$T_k(1+\tfrac{1}{2}x) = \sum_{j=0}^k t_j^{(k)}x^j$, by
\eqref{eqn-reduced-matrix-polynomial-coefficients}. Combining the estimates,
\begin{equation}
\label{eqn-propagator-infty-norm-decomp-intermediate}
    |\tilde T_k^1|_{\ell^\infty} \leq 1 - {\frac{\alpha}{10}(kh)^2} + T_k(1 + \tfrac{h^2}{2} (\beta + \tfrac{\alpha}{50})) - T_k(1 + \tfrac{h^2}{2}\beta)\,.
\end{equation}
It remains to show the perturbation in $T_k$ is controlled, so the contraction
factor stays of the same order. For $x\geq 1$,
$T_k(x)= \cosh(k\log(x + \sqrt{x^2 - 1}))$, with
\[
    \frac{{\rm d}}{{\rm d}x} T_k(x) = \sinh(k\log(x + \sqrt{x^2 - 1})) \frac{k}{\sqrt{x^2 - 1}}\,.
\]
Substituting $x = 1 + \tfrac{h^2}{2}z$ with $z > 0$,
\begin{equation}
    \left|\frac{{\rm d}}{{\rm d} x} T_k(x)\right|\Big|_{x = 1+ \frac{h^2}{2}z} \leq \tfrac12\left(1 + \tfrac{h^2}{2}z + h \sqrt{z}\sqrt{1 + \tfrac{h^2z}{4}}\right)^k \frac{k}{h\sqrt{z}}\,,
\end{equation}
so by the mean value theorem, for $z \in [\beta, \beta + \tfrac{\alpha}{50}]$,
\begin{equation}
\label{eqn-difference-Tk-bound}
\begin{aligned}
    &\left|T_k\left(1 + \tfrac{h^2}{2}\left(\beta + \tfrac{\alpha}{50}\right)\right) - T_k\left(1 + \tfrac{h^2}{2}\beta\right)\right|\leq \frac{h^2\alpha}{100} \sup_{\substack{x = 1 + \frac{h^2}{2}z,\\ z \in [\beta,\beta + \frac{\alpha}{50}]}} \left|\frac{\mathrm{d}}{\mathrm{d}x} T_k(x)\right|\\
    \leq & \frac{kh\alpha}{200\sqrt{\beta}} \exp\left( kh\Big(\tfrac{51}{100}h\beta + \sqrt{\tfrac{51}{50}\beta}\sqrt{1 + \tfrac{51}{200}h^2\beta}\Big)\right)\leq {\frac{kh\alpha}{100\sqrt{\beta}}}\,,
\end{aligned}
\end{equation}
where the last step uses $h \leq kh \leq 1/\sqrt{20\beta}$. Substituting
into~\eqref{eqn-propagator-infty-norm-decomp-intermediate},
\begin{equation}
    |\tilde T_k^1|_{\ell^\infty} \leq 1 - {\frac{\alpha}{10}(kh)^2 + \frac{kh\alpha}{100\sqrt{\beta}}}\,.
\end{equation}
This remains a contraction as in the $\ell^2$ norm for large enough $kh$. Plugging in $kh = \frac{1}{\sqrt{20\beta}}$ gives the upper bound $1 - \frac{\alpha}{400\beta}$.

An analogous argument bounds $h \tilde U_{k-1}^1$.
Lemma~\ref{lem-matrix-coefficients-sign} gives
\begin{equation}
\label{eqn-weak-interaction-Bk-infty-norm-decomp}
\begin{aligned}
    |h \tilde U_{k-1}^1|_{\ell^\infty} &\leq \underbrace{|h \tilde U_{k-1}^1(\{I - \tfrac{h^2}{2}H_i^{(D)}\}_{i=1}^{k-1})|_{\ell^\infty}}_{(a)}\\
    &\quad + \underbrace{\sum_{j=1}^{k-1}h^{2j+1}\sum_{k-1\geq i_1 >\ldots >i_j \geq 1} \tilde u_{i_1,\ldots,i_j}^{(k)}|H_{i_1}\cdots H_{i_j} - H_{i_1}^{(D)}\cdots H_{i_j}^{(D)}|_{\ell^\infty}}_{(b)}\,,
\end{aligned}
\end{equation}
where, as before, $(a) = |h \tilde U_{k-1}^1(\{I-\tfrac{h^2}{2}H_i^{(D)}\}_{i=1}^{k-1})|_{\ell^2} \leq \sqrt{\tfrac{2}{5\beta}}$
by Proposition~\ref{prop-l2bounds-hmc}, and
\begin{equation}
\begin{aligned}
    (b) &\leq \sum_{j=0}^{k-1} h^{2j+1}\sum_{k-1\geq i_1 > \ldots > i_j \geq 1}\tilde u_{i_1,\ldots, i_j}^{(k)}\Big((\beta + \tfrac{\alpha}{50})^j - \beta^j \Big)\\
    & = h U_{k-1}(1 + \tfrac{h^2}{2}(\beta + \tfrac{\alpha}{50})) - h U_{k-1}(1 + \tfrac{h^2}{2}\beta)\leq h U_{k-1}(1 + \tfrac{h^2}{2}(\beta + \tfrac\alpha{50}))\,,
\end{aligned}
\end{equation}
since $U_{k-1}(x) \geq 0$ for $x \geq 1$. Using
$U_{k-1}(x) = \sinh(k\log(x+\sqrt{x^2-1}))/\sqrt{x^2-1}$ for $x \geq 1$,
\begin{equation}
\label{eqn-Chebyshev-U-b-bound}
    (b) \leq \frac{1}{2\sqrt\beta}\exp\left( kh\Big(\tfrac{51}{100}h\beta + \sqrt{\tfrac{51}{50}\beta}\sqrt{1 + \tfrac{51}{200}h^2\beta}\Big)\right)\leq \frac{13}{20\sqrt\beta}\,.
\end{equation}
Combining these in~\eqref{eqn-weak-interaction-Bk-infty-norm-decomp},
\begin{equation}
    |h\tilde U_{k-1}^1|_{\ell^\infty} \leq \frac{3}{2\sqrt{\beta}}\,.
\end{equation}
\end{proof}

\subsection{$\ell^\infty$ bounds under sparse interactions} 
\label{sec-sparse-hmc}
Here we use the sparsity of the Hessians to bound the propagators in $\ell^\infty$.
% Again, we introduce assumptions on general matrices that are analogous to Assumption \ref{assumption-sparse-interactions}, so that their results naturally apply to the target measure, whose Hessian is an instantiation of these matrices.
% \begin{assumption}[sparse interactions, general matrices]
% \label{assumption-sparse-interactions-linear}
% We assume $\{H_i\}$, $H$, and $H'$ are symmetric positive definite matrices satisfying $\alpha I \preceq \cdot \preceq \beta I$. Additionally, any product of $r$ of these matrices is assumed to be $s_r$-sparse (i.e., each row contains at most $s_r$ nonzero entries).
% \end{assumption}
The next lemma, adapted from Proposition~C.1 of \cite{chen2024convergence}, shows that the
integrated Hessian products inherit a sparsity property from
Assumption~\ref{assumption-sparse-interactions}. We include the proof for
completeness.
\begin{lemma}
\label{lem-sparse-instantiation}
Suppose $V$ satisfies Assumption~\ref{assumption-sparse-interactions}, and let
$\nu_1, \dots, \nu_r$ be probability measures on $\mathbb{R}^d$. Then the product
\begin{equation}
    \left(\int \nabla^2 V {\rm d} \nu_1\right)\left(\int \nabla^2 V {\rm d} \nu_2\right) \cdots \left(\int \nabla^2 V {\rm d} \nu_r\right)
\end{equation}
is $s_r$-sparse. More precisely, its $(i,j)$ entry is nonzero only if
$j \in \sfN_{2r}(i)$.
\end{lemma}
\begin{proof}[Proof of Lemma~\ref{lem-sparse-instantiation}]
The $(i,j)$ entry of the product is nonzero only if $j \in \mathsf{N}_r(i)$, or if
$i \in \sfN_r(k)$ and $j \in \sfN_r(k)$ for some $k$.
Since $\max_{1\leq i \leq d} |\sfN_{2r}(i)| = s_r$, each row has at most $s_r$
nonzero entries, so the matrix is $s_r$-sparse.
\end{proof}
 
We now prove an extended version of Proposition~\ref{prop-linf-bounds-sparse-hmc}.
\begin{proposition}[Extended version of Proposition~\ref{prop-linf-bounds-sparse-hmc}]
\label{prop-linf-bounds-sparse-hmc-appendix}
Suppose Assumptions~\ref{assumption-V-log-concave}
and~\ref{assumption-sparse-interactions} hold, and
let the integrated Hessians $H_i$ be as in Proposition~\ref{prop-l2bounds-hmc}.
Take $r_i = \lceil ih\sqrt{\beta} e + \frac{\log\sqrt{d}}{\log(5/3)}\rceil$, and
let $mh \leq 1/\sqrt{20\beta}$ with integers $N \geq 0$ and $k < m$.
Define
\begin{equation}
    A_{Nm + k} \coloneqq \Big(\prod_{j \in\{k+im\}_{i=0}^{N-1}}^{\longleftarrow}\tilde T_m^1(\{I - \tfrac{h^2}{2}H_{j+l} \}_{l=0}^{m-1})\Big)\tilde T_k^1\,,
\end{equation}
\begin{equation}
    B_{Nm+k-1} \coloneqq \Big(\prod_{j \in\{k+im\}_{i=0}^{N-1}}^{\longleftarrow}\tilde T_m^1(\{I - \tfrac{h^2}{2}H_{j+l} \}_{l=0}^{m-1})\Big)h\tilde U_{k-1}^1\,,
\end{equation}
with the product read as the identity when $N=0$. Then
\begin{equation}
\label{eqn-prop-linf-bounds-ANmkBNmk}
    |A_{Nm+k}|_{\ell^\infty} \leq 4 \sqrt{s_{r_{Nm+k}}}\,, \qquad
    |B_{Nm+k-1}|_{\ell^\infty} \leq \frac{4}{\sqrt\beta} \sqrt{s_{r_{Nm+k}}}\,.
\end{equation}
Moreover, for arbitrary probability measures $\nu_0, \nu_0'$ and
$H \coloneqq \int \nabla^2 V {\rm d} \nu_0$, $H' \coloneqq \int \nabla^2 V {\rm d} \nu_0'$,
\begin{equation}
\begin{aligned}
    &|A_{Nm+k}H|_{\ell^\infty} \leq 4\beta \sqrt{s_{r_{Nm+k}}}\,,\\
    &|B_{Nm+k-1}H|_{\ell^\infty} \leq 4\sqrt\beta \sqrt{s_{r_{Nm+k}}}\,, \qquad
    |B_{Nm+k-1}HH'|_{\ell^\infty} \leq 5\beta^{\frac32} \sqrt{s_{r_{Nm+k}}}\,.
\end{aligned}
\end{equation}
\end{proposition}
\begin{proof}[Proof of Proposition~\ref{prop-linf-bounds-sparse-hmc-appendix}]
The bounds follow a common framework, so we give the details for $A_{Nm+k}$. The
strategy is to split the matrix expansion by order in $h$ at a threshold $r$ that
balances sparsity against magnitude, as for a tridiagonal matrix, low-order
powers stay sparse (so their $\ell^\infty$ and $\ell^2$ norms are comparable),
while the growing density of higher-order powers is offset by the decay in
$h^{2j}$. The sparsity here is controlled by $s_r$.
 
For any threshold $r$, Lemma~\ref{lem-sparse-instantiation} gives that the entries
of each order-$r$ product are supported on $\sfN_{2r}(i)$, and the absolute row
sum decomposes as
\begin{equation}
\label{eqn-inf-bound-two-parts-a}
\begin{aligned}
    \sum_{j=1}^d|(A_{Nm+k})_{ij}| = & \sum_{j\in \sfN_{2r}(i)}|(A_{Nm+k})_{ij}| + \sum_{j\notin \sfN_{2r}(i)}|(A_{Nm+k})_{ij}|\\
    \leq &\sqrt{|\sfN_{2r}(i)|}\sqrt{\sum_{j\in\sfN_{2r}(i)}|(A_{Nm+k})_{ij}|^2} + \sum_{j \notin \sfN_{2r}(i)}|(A_{Nm+k})_{ij}|\\
    \leq & \sqrt{s_{r}}|A_{Nm+k}|_{\ell^2} + \sum_{j \notin \sfN_{2r}(i)}|(A_{Nm+k})_{ij}|\,.
\end{aligned}
\end{equation}
By Proposition~\ref{prop-matrix-polynomial-agg-coeff-outer},
\begin{equation}
\label{eqn-decomp-neighbor-matrix-A}
\begin{aligned}
    A_{Nm+k} &= a_0^{(Nm+k)} I + \sum_{j=1}^{r} (-1)^jh^{2j}\sum_{i_1>\dots >i_j}a_{i_1,\dots,i_j}^{(Nm+k)} H_{i_1}\cdots H_{i_j}\\
    & \quad + \sum_{j=r+1}^{Nm+k}(-1)^jh^{2j}\sum_{i_1>\dots>i_j}a_{i_1,\dots,i_j}^{(Nm+k)}H_{i_1}\cdots H_{i_j}\,.
\end{aligned}
\end{equation}
For $j \notin \sfN_{2r}(i)$, the entry $(A_{Nm+k})_{ij}$ can only come from the
second line, so
\begin{equation}
\label{eqn-far-neighbor-elements-bound}
    \sum_{j \notin \sfN_{2r}(i)}|(A_{Nm+k})_{ij}| \leq \Big|\sum_{j=r+1}^{Nm+k}(-1)^j h^{2j}\sum_{i_1,\dots,i_j} a_{i_1,\dots,i_j}^{(Nm+k)}H_{i_1}\dots H_{i_j}\Big|_{\ell^\infty}\,.
\end{equation}
We control this $\ell^\infty$ norm through the $\ell^2$ norm,
\begin{equation}
\label{eqn-tail-matrices}
    \Big|\sum_{j=r+1}^{Nm+k}(-1)^j h^{2j}\sum_{i_1,\dots,i_j} a_{i_1,\dots,i_j}^{(Nm+k)}H_{i_1}\dots H_{i_j}\Big|_{\ell^2} \leq \sum_{j=r+1}^{Nm+k}h^{2j}\beta^j a_j^{(Nm+k)}\,.
\end{equation}
By Proposition~\ref{prop-matrix-polynomial-agg-coeff-outer}, $a_j^{(Nm+k)}$ is
bounded by $\max_{j\leq i\leq Nm+k} t_j^{(i)}$, the Chebyshev coefficients whose
explicit formulas are given in
Proposition~\ref{prop-matrix-polynomial-agg-coeff}. It therefore suffices to
bound $h^{2j}\beta^j t_j^{(i)}$. For $i > 0$,
\begin{equation}
\begin{aligned}
    t_j^{(i)} &= i\frac{(i+j-1)!}{(i-j)!(2j)!} = \frac{i(i+j-1) \cdots (i+1)}{(2j)\cdots(j+1)} \frac{i(i-1)\cdots(i-j+1)}{j!}\\
    & \leq \frac{(i+j)^j}{j^j} \frac{i^j}{j!} \leq \frac{(i+j)^j (ie)^j}{j^{2j}} = \left(e\left(\tfrac{i^2}{j^2}+\tfrac{i}{j} \right) \right)^j\,,
\end{aligned}
\end{equation}
using $j! \geq j^j/\exp(j)$ (from $\exp(j) \geq j^j/j!$). Hence for
$j \geq (Nm+k)\sqrt{\beta}eh$ and $i \leq Nm+k$,
\begin{equation}
    h^{2j}\beta^j t_j^{(i)}\leq \left( h^2\beta e\left(\tfrac{i^2}{j^2}+\tfrac{i}{j} \right) \right)^j \leq \left(e^{-1} + \sqrt\beta h\right)^j \leq \left(\tfrac{3}{5}\right)^j\,,
\end{equation}
where we use $h \leq mh \leq 1/\sqrt{20\beta}$. Since this holds uniformly over
$i \leq Nm+k$, 
Proposition~\ref{prop-matrix-polynomial-agg-coeff-outer} gives $h^{2j} \beta^j a_j^{(Nm+k)} \leq  \left(\tfrac35\right)^j$,
so~\eqref{eqn-tail-matrices} yields
\begin{equation}
    \Big|\sum_{j=r+1}^{Nm+k}(-1)^j h^{2j}\sum_{i_1,\dots,i_j} a_{i_1,\dots,i_j}^{(Nm+k)}H_{i_1}\dots H_{i_j}\Big|_{\ell^2} \leq \sum_{j=r+1}^{Nm+k}\left(\tfrac35\right)^j \leq {3} \left(\tfrac35\right)^{r+1}\,.
\end{equation}
Applying the norm inequality $|M|_{\ell^\infty} \leq \sqrt{d}\,|M|_{\ell^2}$, valid
for any matrix $M \in \mathbb{R}^{d\times d}$,
to~\eqref{eqn-far-neighbor-elements-bound},
\begin{equation}
    \sum_{j \notin \sfN_{2r}(i)}|(A_{Nm+k})_{ij}| \leq {3}\sqrt{d}\left( \tfrac35\right)^{r+1}\,.
\end{equation}
Substituting into~\eqref{eqn-inf-bound-two-parts-a} with $|A_{Nm+k}|_{\ell^2} \leq 1$
(Proposition~\ref{prop-l2bounds-hmc}) and taking
$r_{Nm+k} = \lceil (Nm+k)h\sqrt\beta e + \tfrac{\log \sqrt d}{\log (5/3)} \rceil$,
\begin{equation}
    |A_{Nm+k}|_{\ell^\infty} \leq \sqrt{s_{r_{Nm+k}}} + {3} \left( \tfrac35\right)^{(Nm+k)h\sqrt\beta e}\,,
\end{equation}
which gives the first bound in~\eqref{eqn-prop-linf-bounds-ANmkBNmk} since
$s_{r_{Nm+k}} \geq 1$.
 
For $A_{Nm+k}H$ the decomposition is the same
as~\eqref{eqn-inf-bound-two-parts-a} with $H$ appended, giving
$|A_{Nm+k}H|_{\ell^2} \leq \beta |A_{Nm+k}|_{\ell^2}$ and
\begin{equation}
    \Big|\sum_{j=r}^{Nm+k}(-1)^jh^{2j}\sum_{i_1>\dots>i_j}a_{i_1,\dots,i_j}^{(Nm+k)}H_{i_1}\cdots H_{i_j} H\Big|_{\ell^2} \leq \beta \sum_{j=r}^{Nm+k}\left(\tfrac35\right)^j \leq {3}\beta\left( \tfrac{3}{5}\right)^r\,,
\end{equation}
so taking $r_{Nm+k}$ as before,
\begin{equation}
    |A_{Nm+k} H|_{\ell^\infty} \leq \beta \sqrt{s_{r_{Nm+k}}} + {3}\beta\left(\tfrac35\right)^{(Nm+k)h\sqrt\beta e} \leq 4\beta \sqrt{s_{r_{Nm+k}}}\,.
\end{equation}
 
For $B_{Nm+k-1}$, the same decomposition gives, for any $r$,
\begin{equation}
\label{eqn-inf-bound-two-parts-b}
    \sum_{j=1}^d|(B_{Nm+k-1})_{ij}| \leq \sqrt{s_r}|B_{Nm+k-1}|_{\ell^2} + \sqrt{d}\sum_{j=r+1}^{Nm+k-1}h^{2j+1}\beta^jb_j^{(Nm+k-1)}\,.
\end{equation}
The difference from $A_{Nm+k}$ is that $b_j^{(Nm+k-1)}$ depends on $u_j^{(i)}$. By
Proposition~\ref{prop-matrix-polynomial-agg-coeff},
\begin{equation}
    u_j^{(i)} = \frac{(i+j+1)!}{(i-j)!(2j+1)!} \leq \frac{(i+j+1)^j}{j^j}\frac{(i+1)^{j+1}e^{j+1}}{j^{j+1}}\,,
\end{equation}
so by Proposition~\ref{prop-matrix-polynomial-agg-coeff-outer}, for
$j \geq (Nm+k)\sqrt\beta eh$ and $i \leq Nm+k-1$,
\begin{equation}
    h^{2j+1}\beta^jb_j^{(Nm+k-1)} \leq \frac{1}{\sqrt\beta}\left(h\sqrt\beta\left(\tfrac{i+1}{j}+1\right) \right)^j\left(\tfrac{h\sqrt\beta(i+1)e}{j}\right)^{j+1} \leq \frac{1}{\sqrt\beta}\left(\tfrac35\right)^j\,,
\end{equation}
and hence
\begin{equation}
    \sqrt{d}\sum_{j=r+1}^{Nm+k-1}h^{2j+1}\beta^jb_j^{(Nm+k-1)} \leq \frac{{3}\sqrt{d}}{\sqrt{\beta}}\left(\tfrac35\right)^{r+1}\,.
\end{equation}
Taking $r = r_{Nm+k}$ and using $|B_{Nm+k-1}|_{\ell^2} \leq \sqrt{\tfrac{2}{5\beta}}$
(Proposition~\ref{prop-l2bounds-hmc}),~\eqref{eqn-inf-bound-two-parts-b} gives
\begin{equation}
    |B_{Nm+k-1}|_{\ell^\infty} \leq \sqrt{\tfrac{2}{5\beta}}\sqrt{s_{r_{Nm+k}}}+ \frac{{3}}{\sqrt\beta}\left(\tfrac35\right)^{(Nm+k)h\sqrt\beta e} \leq \frac{4}{\sqrt\beta} \sqrt{s_{r_{Nm+k}}}\,.
\end{equation}
The bounds for $B_{Nm+k-1}H$ and $B_{Nm+k-1}HH'$ follow from similar
decompositions:
\begin{equation}
    |B_{Nm+k-1}H|_{\ell^\infty} \leq \sqrt{\tfrac{2\beta}{5}}\sqrt{s_{r_{Nm+k}}} + {3}\sqrt\beta \left(\tfrac35\right)^{(Nm+k)h\sqrt\beta e} \leq 4 \sqrt\beta\sqrt{s_{r_{Nm+k}}}\,,
\end{equation}
\begin{equation}
    |B_{Nm+k-1}H H'|_{\ell^\infty} \leq \sqrt{\tfrac{2\beta^3}{5}}\sqrt{s_{r_{Nm+k}}} + \frac{25}{6} \beta^{\frac32}\left(\tfrac{3}{5} \right)^{(Nm+k)\sqrt\beta eh} \leq 5\beta^{\frac32}\sqrt{s_{r_{Nm+k}}}\,.
\end{equation}
\end{proof}
\section{Bounds for BAOAB Propagators}
\label{appendix-bounds-baoab}
In this section we bound the BAOAB propagators. The key difference from the
leap-frog case is that we work with weighted $\ell^2$ and $\ell^\infty$ norms,
defined through the weight matrix in~\eqref{eqn-def-weighted-linf-norm}, which is
what makes the propagators contractive.
\subsection{$\ell^2$ bounds} The $\ell^2$ bounds rest on the BAOAB contraction result, Theorem~5.1
of~\cite{leimkuhler2024contraction}.
\begin{lemma}[Theorem 5.1 in \cite{leimkuhler2024contraction}]
\label{lem-contraction-baoab-original}
Let Assumption~\ref{assumption-V-log-concave} hold. If
$h \leq \frac{1-\eta}{2\sqrt{\beta}}$, then for any initial states
$(x_0, p_0), (y_0, p'_0) \in \mathbb{R}^{2d}$ and any standard normal sequence
$\{\xi_i\}_{i=0}^{k-1}$, the BAOAB chains driven by the same noise,
\begin{equation}
    (x_k, p_k) = \sfU_{\baoab,h}^{\xi_{k-1}}\cdots \sfU_{\baoab,h}^{\xi_0}(x_0, p_0)\,, \quad
    (y_k, p_k') = \sfU_{\baoab,h}^{\xi_{k-1}}\cdots \sfU_{\baoab,h}^{\xi_0}(y_0, p_0')\,,
\end{equation}
satisfy almost surely
\begin{equation}
    |(x_k - y_k, p_k - p'_k)|_{\ell^2_w} \leq 7 (1-c(h))^{\frac{k-1}{2}}|(x_0 - y_0, p_0 - p'_0)|_{\ell^2_w}\,,
\end{equation}
with $c(h) = \frac{\alpha h^2}{4(1-\eta)}$.
\end{lemma}
As for HMC, we prove a general-matrix version,
Proposition~\ref{prop-l2bounds-baoab-general} below, by reading the polynomials with
general $\{H_i\}$ as BAOAB propagators with varying quadratic potentials and
turning the contraction of Lemma~\ref{lem-contraction-baoab-original} into
operator-norm bounds.
\begin{proposition}[$\ell^2$ bounds for BAOAB, general matrices]
\label{prop-l2bounds-baoab-general}
Let $\{H_i\}$ be symmetric positive definite with
$\alpha I \preceq H_i \preceq \beta I$ for $\alpha, \beta$ as in
Assumption~\ref{assumption-V-log-concave}. For $h \leq \frac{1-\eta}{2\sqrt{\beta}}$
and $c(h) = \frac{\alpha h^2}{4(1-\eta)}$, the polynomials $\tilde T_k^\eta$ and
$\tilde U_{k-1}^\eta$ defined with $\{H_i\}$, and the full-space weighted
propagator $M_{\baoab, k}^w$, satisfy
\begin{equation}
\label{eqn-prop-l2bounds-baoab-fullspace-general}
    |M_{\baoab,k}^w|_{\ell^2} \leq 7 (1-c(h))^{\frac{k-1}{2}}\,,
\end{equation}
\begin{equation}
\label{eqn-prop-l2bounds-baoab-11block-general}
%\label{eqn-prop-l2bounds-baoab-12block-general}
    |\tilde T_k^\eta|_{\ell^2} \leq 7 \sqrt{2} (1 - c(h))^{\frac{k-1}{2}}\,, \quad 
    \frac{1}{\sqrt{a-b^2}}\left|-b \tilde T_k^\eta + h \tilde U_{k-1}^\eta\right|_{\ell^2} \leq 7\sqrt{2}\left( 1 - c(h)\right)^{\frac{k-1}{2}}\,.
\end{equation}
\end{proposition}
\begin{proof}[Proof of Proposition \ref{prop-l2bounds-baoab-general}]
Consider the time-inhomogeneous linear dynamics driven by $\{H_i\}$,
\begin{equation}
\label{eqn-baoab-linear}
\begin{aligned}
    q_{k+1} &= q_k + \tfrac{h}{2}(1+\eta)p_k - \tfrac{h^2}{4}(1+\eta)H_k q_k+ \tfrac{h}{2}\sqrt{1-\eta^2}\xi_{k}\,,\\
    p_{k+1} &= \eta(p_k - \tfrac{h}{2}H_kq_k) + \sqrt{1-\eta^2}\xi_{k}-\tfrac{h}{2}H_{k+1}q_{k+1}\,,
\end{aligned}
\end{equation}
which is the BAOAB scheme~\eqref{eqn-baoab} with the varying potential
$V_k(q) = \tfrac12q^\top H_k q$ at step $k$. For two trajectories coupled through
the same noise $\{\xi_i\}$, the differences $\Delta q_k, \Delta p_k$ obey
\begin{equation}
\label{eqn-leap-frog-varying-potential}
\begin{aligned}
    \Delta q_{k+1} &= (I - \tfrac{h^2}{4}(1+\eta)H_k) \Delta q_k + \tfrac{h}{2}(1+\eta)\Delta p_k\,,\\
    \Delta p_{k+1} &= \eta(\Delta p_k - \tfrac{h}{2}H_k\Delta q_k) -\tfrac{h}{2}H_{k+1}\Delta q_{k+1}\,.
\end{aligned}
\end{equation}
These yield the same three-term recurrence
$\Delta q_{k+1} = (1+\eta)(I - \tfrac{h^2}{2}H_k)\Delta q_k - \eta \Delta q_{k-1}$
as in the proof of Proposition~\ref{prop-matrix-polynomial}, so the same
induction gives
\begin{equation}
\label{eqn-q-on-initial-q0-p0-baoab-linear}
    \Delta q_{k} = \tilde T_{k}^\eta \Delta q_0 + h\tilde U_{k-1}^\eta \Delta p_0\,,
\end{equation}
identifying $\tilde T_k^\eta$ and $h\tilde U_{k-1}^\eta$ as the position-difference
propagators of~\eqref{eqn-baoab-linear}.
 
Although Lemma~\ref{lem-contraction-baoab-original} is stated for a fixed
potential, its proof uses only the spectral bounds
$\alpha I \preceq H_i \preceq \beta I$. \cite{leimkuhler2024contraction} employs the decomposition
\begin{equation}
\label{eqn-baoab-decomposition-abao}
    (\mathsf{B}\mathsf{A}\mathsf{O}\mathsf{A}\mathsf{B})^k = \mathsf{B}\mathsf{A}\mathsf{O}(\mathsf{A}\mathsf{B}\mathsf{A}\mathsf{O})^{k-1}\mathsf{A}\mathsf{B}
\end{equation}
and shows that the core operator $\mathsf{A}\mathsf{B}\mathsf{A}\mathsf{O}$ is a weighted
$\ell^2$ contraction (giving the $(1-c(h))$ factors) while the boundary terms $\mathsf{B}\mathsf{A}\mathsf{O}$ and
$\mathsf{A}\mathsf{B}$ stay bounded, all under those eigenvalue bounds.
Since each $H_i$ satisfies them, the result carries over to the varying-potential
case, giving under~\eqref{eqn-baoab-linear}
\begin{equation}
\label{eqn-baoab-weighted-2-norm-varying-potential}
    |(\Delta q_k, \Delta p_k)|_{\ell^2_w} \leq 7(1-c(h))^{\frac{k-1}{2}}|(\Delta q_0, \Delta p_0)|_{\ell^2_w}\,.
\end{equation}
 
We now derive the operator-norm bounds. Since
$M_{\baoab,k}^w = W M_{\baoab,k} W^{-1}$ maps $W(\Delta q_0, \Delta p_0)$ to
$W(\Delta q_k, \Delta p_k)$, \eqref{eqn-baoab-weighted-2-norm-varying-potential}
is precisely~\eqref{eqn-prop-l2bounds-baoab-fullspace-general}. For the position
blocks, the norm equivalence~\eqref{eqn-l2-norm-equivalence} gives
\begin{equation}
    |\Delta q_k|_{\ell^2} \leq \sqrt{|\Delta q_k|_{\ell^2}^2 + a |\Delta p_k|_{\ell^2}^2} \leq \sqrt{2}|(\Delta q_k, \Delta p_k)|_{\ell^2_w}\,,
\end{equation}
so
\begin{equation}
\label{eqn-useful-cor-baoab-pre-bound}
    |\Delta q_k|_{\ell^2} \leq 7\sqrt{2}(1- c(h))^{\frac{k-1}{2}}|W(\Delta q_0, \Delta p_0)|_{\ell^2}\,.
\end{equation}
Writing~\eqref{eqn-q-on-initial-q0-p0-baoab-linear} as
$\Delta q_k = \begin{bmatrix} \tilde T_{k}^\eta & h\tilde U_{k-1}^\eta \end{bmatrix} W^{-1}\, W(\Delta q_0, \Delta p_0)$,
the bound~\eqref{eqn-useful-cor-baoab-pre-bound} gives
\begin{equation}
    \left|\begin{bmatrix} \tilde T_{k}^\eta & h\tilde U_{k-1}^\eta \end{bmatrix} W^{-1} \right|_{\ell^2} \leq 7\sqrt{2}(1-c(h))^{\frac{k-1}{2}}\,,
\end{equation}
and since
\begin{equation}
    \begin{bmatrix} \tilde T_k^\eta & h \tilde U_{k-1}^\eta \end{bmatrix} W^{-1} =
    \begin{bmatrix} \tilde T_k^\eta & - \frac{b}{\sqrt{a-b^2}} \tilde T_k^\eta + \frac{1}{\sqrt{a-b^2}}h\tilde U_{k-1}^\eta \end{bmatrix}\,,
\end{equation}
the block-wise bounds~\eqref{eqn-prop-l2bounds-baoab-11block-general}
follow.
\end{proof}
We extract a corollary in $\ell^\infty$ norms from this proof.
\begin{corollary}
\label{cor-useful-estimates-bound-baoab}
Under the assumptions and notation of
Proposition~\ref{prop-l2bounds-baoab-general}, any state difference
$(\Delta q_k, \Delta p_k)$ under the varying-potential
scheme~\eqref{eqn-leap-frog-varying-potential} satisfies
\begin{equation}
\label{eqn-useful-cor-baoab-weighted-infty-bound}
    |(\Delta q_k, \Delta p_k)|_{\ell^\infty_w} \leq 7\sqrt{2d}(1-c(h))^{\frac{k-1}{2}}|(\Delta q_0, \Delta p_0)|_{\ell^\infty_w}\,,
\end{equation}
\begin{equation}
\label{eqn-useful-cor-baoab-infty-bound}
    |\Delta q_k|_{\ell^\infty} \leq 14\sqrt{d}(1-c(h))^{\frac{k-1}{2}}|(\Delta q_0, \Delta p_0)|_{\ell^\infty_w}\,.
\end{equation}
\end{corollary}
\begin{proof}[Proof of Corollary~\ref{cor-useful-estimates-bound-baoab}]
Inequality~\eqref{eqn-useful-cor-baoab-weighted-infty-bound} follows
from~\eqref{eqn-baoab-weighted-2-norm-varying-potential} and the matrix norm
equivalence $|M|_{\ell^\infty} \leq \sqrt{2d}|M|_{\ell^2}$ for
$M \in \mathbb{R}^{2d \times 2d}$.
For~\eqref{eqn-useful-cor-baoab-infty-bound},~\eqref{eqn-useful-cor-baoab-pre-bound}
gives
\begin{equation}
    |\Delta q_k|_{\ell^\infty} \leq |\Delta q_k|_{\ell^2} \leq 7\sqrt{2}(1-c(h))^{\frac{k-1}{2}}|(\Delta q_0, \Delta p_0)|_{\ell^2_w} \leq 14\sqrt{d}(1-c(h))^{\frac{k-1}{2}}|(\Delta q_0, \Delta p_0)|_{\ell^\infty_w}\,.
\end{equation}
\end{proof}
\begin{proof}[Proof of Proposition~\ref{prop-l2bounds-baoab}]
This follows from Proposition~\ref{prop-l2bounds-baoab-general}, since the
integrated Hessians satisfy the assumption.
\end{proof}
\subsection{$\ell^\infty$ bounds under weak interactions}
We bound the BAOAB propagators in $\ell^\infty$ following
\cite{leimkuhler2024contraction}, decomposing the scheme into the operators
$\mathsf{A}\mathsf{B}\mathsf{A}\mathsf{O}$, $\mathsf{B}\mathsf{A}\mathsf{O}$, and
$\mathsf{A}\mathsf{B}$ as in~\eqref{eqn-baoab-decomposition-abao}, where
$\mathsf{A}\mathsf{B}\mathsf{A}\mathsf{O}$ drives the contraction and the boundary
terms $\mathsf{B}\mathsf{A}\mathsf{O}$, $\mathsf{A}\mathsf{B}$ stay bounded. The contraction remains under the weak interactions.
 
Unlike HMC, a diagonal Hessian does not give a diagonal propagator here. On the
joint $(q,p)$ space the propagator is a $2\times2$ block matrix with diagonal
blocks. The crude inequality $|M|_{\ell^\infty} \leq \sqrt{2}|M|_{\ell^2}$ (from
the two nonzero entries per row) would not preserve the $\ell^2$ contraction, so
we bound each block in $\ell^\infty$ directly, using the $\ell^2$ bounds as a
guide to the contraction rate.
\begin{proof}[Proof of Proposition~\ref{prop-propagators-weak-interactions-baoab}]
\noindent\textbf{Estimate for $M_{\abao}^w$.}
The matrix is
\begin{equation}
    M_{\abao}^w(H_0) = \begin{bmatrix}
        I - \tfrac12 h^2 H_0 - \tfrac{\eta}{1-\eta}h^2 H_0 & h^3\tfrac{\frac14 + \frac{\eta}{(1-\eta)^2}}{\sqrt{a-b^2}}H_0\\
        - \eta h \sqrt{a-b^2}H_0 & \eta (I + (-\tfrac12 + \tfrac{1}{1-\eta})h^2H_0)
    \end{bmatrix}\,.
\end{equation}
For the second block of rows, the $\ell^1$ norm of each row is at most
\begin{equation}
\label{eqn-abao-row-2-bound}
    \eta h \sqrt{a-b^2} \cdot \tfrac{21}{20}\beta + \eta (1 + (-\tfrac12 + \tfrac{1}{1-\eta}) h^2 \cdot \tfrac{21}{20}\beta)\,.
\end{equation}
This increases in $\eta \in [0,1]$, while the target bound
$1 - \frac{h^2\alpha}{8(1-\eta)}$ decreases in $\eta$. Comparing them at the
largest admissible value $\eta = 1-2\sqrt\beta h$ confirms the target bound for
the second block.
 
For the first block of rows, split off the diagonal and off-diagonal parts
$H_0^{(D)}, H_0^{(O)}$ (defined as in the proof of
Proposition~\ref{prop-linftybounds-hmc-weak}):
\begin{equation}
\begin{aligned}
    &\begin{bmatrix} I - \tfrac12 h^2 H_0^{(D)} - \tfrac{\eta}{1-\eta}h^2 H_0^{(D)} & h^3\tfrac{\frac14 + \frac{\eta}{(1-\eta)^2}}{\sqrt{a-b^2}} H_0^{(D)} \end{bmatrix}\\
    & \quad + \begin{bmatrix} - \tfrac12 h^2 H_0^{(O)} - \tfrac{\eta}{1-\eta}h^2 H_0^{(O)} & h^3\tfrac{\frac14 + \frac{\eta}{(1-\eta)^2}}{\sqrt{a-b^2}} H_0^{(O)} \end{bmatrix}\,.
\end{aligned}
\end{equation}
In the first matrix, $H_0^{(D)}$ is diagonal, so each row reduces to a scalar
$\lambda \in [\alpha, \beta]$. Since
$\tfrac12 h^2\beta + \tfrac{\eta}{1-\eta}h^2\beta$ increases in $\eta$ and is
below $1$ at $\eta = 1 - 2\sqrt{\beta} h$, it is below $1$ throughout, and the
$\ell^1$ norm of the row is
\begin{equation}
\label{eqn-abao-row-1-bound}
    1 - \tfrac12 h^2\lambda - \tfrac{\eta}{1-\eta}h^2\lambda + h^3 \tfrac{\frac14 + \frac{\eta}{(1-\eta)^2}}{\sqrt{a-b^2}}\lambda\,.
\end{equation}
The condition $h \leq \frac{1-\eta}{2\sqrt\beta}$ gives
$\sqrt{a-b^2} \geq \frac{\sqrt{3}}{2 \sqrt{\beta}}$, so the last term is at most
$\frac{1}{4\sqrt{3}}h^2 \lambda + \frac{1}{\sqrt{3}}\frac{\eta}{1-\eta}h^2 \lambda$,
and~\eqref{eqn-abao-row-1-bound} is bounded by $1 - c(h)$. Adding the
off-diagonal contribution, the $\ell^\infty$ norm of the first block of rows is
\begin{equation}
    1 - c(h) + \left((\tfrac12 + \tfrac{\eta}{1-\eta})h^2 + h^3\tfrac{\frac14 + \frac{\eta}{(1-\eta)^2}}{\sqrt{a-b^2}}\right)| H_0^{(O)}|_{\ell^\infty} \leq 1 - \tfrac12c(h)\,.
\end{equation}
 
\noindent\textbf{Estimate for $M_{\bao}^w$.} 
By the scheme definition in~\eqref{eqn-def-baoab-aba-bao-ab},
\begin{equation}
    M_{\bao}^w(H_0) = \begin{bmatrix}
        I - \tfrac{h^2}{4}H_0 - \tfrac{h^2\eta}{2(1-\eta)}H_0 & \tfrac{1}{\sqrt{a-b^2}}(-\tfrac{h}{2}I + \tfrac{h^3(1+\eta)}{4(1-\eta)^2}H_0)\\
        - \tfrac{h}{2}\eta \sqrt{a-b^2}H_0 & \eta I + \tfrac{h^2\eta}{2(1-\eta)}H_0
    \end{bmatrix}\,.
\end{equation}
We bound each $2\times2$ block in $\ell^\infty$. The $(1,1)$ block is bounded by
$1 + (\tfrac{h^2}{4} + \tfrac{h}{2} \cdot \tfrac{1}{2\sqrt\beta})\tfrac{21}{20}\beta \leq \tfrac{383}{320}$,
the $(1,2)$ block by
$\frac{2\sqrt{\beta}}{\sqrt{3}}(\tfrac{h}{2} + \tfrac{h}{4}\cdot 2 \cdot \tfrac{1}{4\beta} \cdot \tfrac{21}{20}\beta) \leq \frac{101}{160\sqrt{3}}$,
the $(2,1)$ block by
$\tfrac12 \tfrac{1}{\sqrt{\beta}} h \cdot \tfrac{21}{20}\beta \leq \tfrac{21}{80}$,
and the $(2,2)$ block by
$\tfrac12 \cdot \tfrac{1}{2\sqrt\beta}h \cdot\tfrac{21}{20}\beta + 1 \leq \tfrac{181}{160}$.
Summing the block bounds over each row gives the result.

\noindent\textbf{Estimate for $M_{\ab}^w$.}
By \eqref{eqn-def-baoab-aba-bao-ab},
\begin{equation}
    M_{\ab}^w(H_0) = \begin{bmatrix}
        I - \tfrac{h^2}{2(1-\eta)}H_0 & \tfrac{1}{\sqrt{a-b^2}}(\tfrac{h}{2}I + \tfrac{h^3(1+\eta)}{4(1-\eta)^2}H_0)\\
        - \tfrac{h}{2}\sqrt{a-b^2}H_0 & I + (\tfrac{h^2}{2(1-\eta)}-\tfrac{h^2}{4})H_0
    \end{bmatrix}\,.
\end{equation}
By the same block-wise bounding, the $(1,1)$ block is at most
$1 + \tfrac{h}{4\sqrt\beta}\tfrac{21}{20}\beta \leq \tfrac{181}{160}$, the $(1,2)$
block at most
$\frac{2\sqrt{\beta}}{\sqrt{3}}(\tfrac{h}{2}+\tfrac{h}{2}\tfrac{1}{4\beta}\tfrac{21}{20}\beta) \leq \frac{101}{160\sqrt{3}}$,
the $(2,1)$ block at most
$\tfrac{h}{2\sqrt{\beta}}\tfrac{21}{20}\beta \leq \tfrac{21}{80}$, and the $(2,2)$
block at most
$1 + \tfrac{h}{2}\tfrac{1}{2\sqrt{\beta}}\tfrac{21}{20}\beta \leq \tfrac{181}{160}$.
A row sum concludes the proof.
\end{proof}
\subsection{$\ell^\infty$ bounds under sparse interactions} This proof follows that of
 Proposition~\ref{prop-linf-bounds-sparse-hmc-appendix}, with two changes. The
$\ell^2$ input is now the BAOAB contraction of
Proposition~\ref{prop-l2bounds-baoab}, and the coefficient bounds come from the
$\eta < 1$ estimates of Proposition~\ref{prop-matrix-polynomial-agg-coeff-eta}
rather than the exact $\eta=1$ formulas. We therefore give only the parts that
differ.
\begin{proof}[Proof of Proposition~\ref{prop-linf-bounds-sparse-baoab}]
Write $A_k \coloneqq \tilde T_{k}^\eta$ and $B_{k-1}\coloneqq h\tilde U_{k-1}^\eta$.
As in~\eqref{eqn-inf-bound-two-parts-a}, for each threshold $r$,
\begin{equation}
\label{eqn-linf-sparse-baoab-decomp}
    \sum_{j=1}^d|(A_k)_{ij}| \leq \sqrt{s_r}|A_k|_{\ell^2} + \sum_{j \notin \sfN_{2r}(i)}|(A_k)_{ij}|\,,
\end{equation}
where $|A_k|_{\ell^2} \leq 7\sqrt{2}(1-c(h))^{\frac{k-1}{2}}$ by
Proposition~\ref{prop-l2bounds-baoab}, and the far-neighbor sum comes from the terms of order greater than $r$,
\begin{equation}
    \sum_{j \notin \sfN_{2r}(i)}|(A_k)_{ij}| \leq \Big|\sum_{j=r+1}^k(-1)^jh^{2j}\sum_{k-1 \geq i_1 > \dots > i_j \geq 0} \tilde t_{i_1, \dots, i_j}^{(k)}H_{i_1}\dots H_{i_j}\Big|_{\ell^\infty}\,,
\end{equation}
whose $\ell^2$ norm is at most $\sum_{j=r+1}^k h^{2j} \beta^j t_j^{(k)}$. By
Proposition~\ref{prop-matrix-polynomial-agg-coeff-eta},
\begin{equation}
    t_j^{(k)} \leq \left(\tfrac{1}{1-\eta}\right)^j \binom{k}{j} \leq \left( \tfrac{k e}{(1-\eta)j} \right)^j\,,
\end{equation}
using $j! \geq j^j/\exp(j)$. Hence for $r > \frac{ke^2\beta}{1-\eta}h^2$,
\begin{equation}
\begin{aligned}
    \Big|\sum_{j=r+1}^k(-1)^jh^{2j}\sum_{k-1 \geq i_1 > \dots > i_j \geq 0} \tilde t_{i_1, \dots, i_j}^{(k)}H_{i_1}\dots H_{i_j}\Big|_{\ell^2}
    &\leq \sum_{j=r+1}^k \left( h^2 \tfrac{ke}{(1-\eta)j}\beta\right)^j\\
    &\leq \sum_{j=r+1}^k \exp(-j) \leq \exp(-r)\,.
\end{aligned}
\end{equation}
Substituting into~\eqref{eqn-linf-sparse-baoab-decomp} with
$|M|_{\ell^\infty} \leq \sqrt d |M|_{\ell^2}$ and taking
$r_k \coloneqq \lceil \frac{ke^2\beta}{1-\eta}h^2+\log\sqrt{d} \rceil$,
\begin{equation}
\begin{aligned}
    |A_k|_{\ell^\infty} &\leq 7\sqrt2\,\sqrt{s_{r_k}}(1- c(h))^{\frac{k-1}{2}}+ \exp\left(-\tfrac{ke^2\beta}{1-\eta}h^2\right) \\&\leq (7\sqrt 2 + 1) \sqrt{s_{r_k}} \exp\left(-\tfrac{\alpha h^2}{8(1-\eta)}(k-1)\right)\,.
\end{aligned}
\end{equation}
For $B_{k-1}$, Proposition~\ref{prop-matrix-polynomial-agg-coeff-eta} gives
\begin{equation}
    u_j^{(k-1)} \leq \left(\tfrac{1}{1-\eta}\right)^{j+1}\binom{k-1}{j} \leq \left( \tfrac{1}{1-\eta}\right)^{j+1} \left(\tfrac{(k-1)e}{j}\right)^j\,,
\end{equation}
so the terms of order greater than $r$ have $\ell^2$ norm at most
$\sum_{j=r+1}^{k-1}h^{2j+1}\beta^ju_j^{(k-1)} \leq \frac{h}{1-\eta} \exp(-r) = b \exp(-r)$.
Combining the tail estimates for $A_k$ and $B_{k-1}$,
\begin{equation}
    \sum_{j\notin \sfN_{2r}(i)}\left| \left(-\tfrac{b}{\sqrt{a-b^2}} A_k + \tfrac{1}{\sqrt{a-b^2}}B_{k-1} \right)_{ij}\right| \leq 2\sqrt{d}\, \tfrac{b}{\sqrt{a-b^2}}\exp(-r)\,,
\end{equation}
while Proposition~\ref{prop-l2bounds-baoab} gives
$|-\tfrac{b}{\sqrt{a-b^2}} A_k + \tfrac{1}{\sqrt{a-b^2}}B_{k-1}|_{\ell^2} \leq 7\sqrt{2}(1-c(h))^{\frac{k-1}{2}}$.
Hence, simplifying $\frac{b}{\sqrt{a-b^2}}$ via $h\leq \frac{1-\eta}{2\sqrt{\beta}}$,
\begin{equation}
\begin{aligned}
    &\left|-\tfrac{b}{\sqrt{a-b^2}} A_k + \tfrac{1}{\sqrt{a-b^2}}B_{k-1}\right|_{\ell^\infty} \leq \sqrt{s_{r_k}}7\sqrt{2}(1-c(h))^{\frac{k-1}{2}} + \tfrac{2b}{\sqrt{a-b^2}}\exp\left(-\tfrac{ke^2\beta}{1-\eta}h^2\right)\\
    &\qquad \leq \left(7\sqrt{2}+\tfrac{2}{\sqrt{3}}\right)\sqrt{s_{r_k}}\exp\left(-\tfrac{\alpha h^2}{8(1-\eta)}(k-1)\right)\,.
\end{aligned}
\end{equation}
\end{proof}

\section{Discretization Error Analysis}
\label{appendix-discretization-error}
We now bound the discretization error by comparing the discrete and continuous
dynamics. The following two bounds from \cite{chen2024convergence}, used
throughout, produce the $\log(2d)$ dependence of the bias.

\begin{lemma}[Lemma B.1 in \cite{chen2024convergence}]
\label{lem-Gaussian-linf}
Suppose $Y = (Y^{(1)}, \dots, Y^{(d)}) \in \bR^d$ with each $Y^{(i)}$ centered and
sub-Gaussian with variance proxy $\sigma^2$, i.e.\
$\bE[\exp(\lambda Y^{(i)})] \leq \exp(\tfrac{1}{2}\lambda^2\sigma^2)$. Then
\begin{equation}
    \bE[|Y|_{\ell^\infty}^2] \leq 4\sigma^2 \log(2d)\,.
\end{equation}
\end{lemma}
 
\begin{lemma}[Proposition 2.3 in \cite{chen2024convergence}]
\label{lem-gradV-expectation}
Let Assumption~\ref{assumption-V-log-concave} hold. Then
\begin{equation}
\label{eqn-expected-inf-V-bound}
    \sqrt{\bE_\pi[|\nabla V|_{\ell^\infty}^2]} \leq 2\sqrt{\beta\log (2d)}\,.
\end{equation}
\end{lemma}
\subsection{Discretization error in HMC}
\label{sec-discretization-error-hmc}
The discretization error
decompositions~\eqref{eqn-one-step-hmc-discretization-error-decomposition-propagation}
and~\eqref{eqn-multi-step-hmc-discretization-error-decomposition-propagation} in
Section~\ref{sec-sketch-hmc} share the common structure
\begin{equation}
\label{eqn-hmc-discretization-error-common-structure}
    |M \sfPi_1\sfU_{\hmc,h}^{m-i-1}(\sfU_{\hmc,h}^{1} - \sfU_{\hmc}^h)(q^*, p^*)|_{2,\ell^\infty}\,,
\end{equation}
with a matrix $M \in \mathbb{R}^{d\times d}$ and $(q^*, p^*) \sim \pi \otimes \mathcal{N}(0,I)$,
where, as recalled in Section~\ref{sec-discretization-error-main-text},
\begin{equation}
    (\sfU_{\hmc,h}^1 - \sfU_{\hmc}^h)(q^*, p^*) = ((\sfQ_h^1-\sfQ^h), (\sfP_h^1-\sfP^h))(q^*, p^*)\,.
\end{equation}

The following proposition bounds the error of each component.
\begin{proposition}[Generalization of Proposition~\ref{prop-discretization-error-hmc-main-text}]
\label{prop-discretization-error}
For any matrix $M \in \mathbb{R}^{d \times d}$, suppose there exist constants
$C_3, C_4 > 0$ (possibly depending on $\beta$) such that
$|M \int \nabla^2 V {\rm d} \nu_1|_{\ell^\infty} \allowbreak< C_3$ and
$|M (\int \nabla^2 V {\rm d} \nu_1)(\int \nabla^2 V {\rm d} \nu_2)|_{\ell^\infty} < C_4$
for all probability measures $\nu_1, \nu_2$ on $\mathbb{R}^d$. Then for
$h \leq 1/\sqrt{20\beta}$ and $(q^*, p^*) \sim \pi \otimes \mathcal{N}(0, I)$,
\begin{equation}
\label{eqn-one-step-position-err}
    |M(\sfQ_h^1(q^*,p^*)-\sfQ^h(q^*,p^*))|_{2,\ell^\infty} \leq \frac1{\sqrt{2}}h^3C_3\sqrt{\log(2d)}\,,
\end{equation}
\begin{equation}
\label{eqn-one-step-momentum-err}
    |M(\sfP_h^1(q^*,p^*)-\sfP^h(q^*,p^*))|_{2,\ell^\infty}\leq 2 h^2C_3 \sqrt{\log(2d)} + \frac{1}{2\sqrt{5}}h^4 C_4 \sqrt{\log(2d)}\,.
\end{equation}
\end{proposition}
\begin{proof}[Proof of Proposition~\ref{prop-discretization-error}]
Let $(q_t, p_t) = \sfU_{\hmc}^t(q^*, p^*)$. The continuous and discrete position
updates have the integral representations
\begin{equation}
\label{eqn-integral-hmc-position}
    \sfQ^h(q^*,p^*) = q^* + p^*h - \int_0^h\int_0^t \nabla V(q_s){\rm d} s {\rm d} t\,,\quad \sfQ_h^1(q^*,p^*) = q^* + p^*h - \tfrac12 h^2 \nabla V(q^*)\,,
\end{equation}
whose difference is
$\sfQ_h^1(q^*,p^*)-\sfQ^h(q^*,p^*) = \int_0^h \int_0^t (\nabla V(q_s)-\nabla V(q^*)) {\rm d}s {\rm d}t$.
Writing $\nabla V(q_s) - \nabla V(q^*) = (\int_0^1 \nabla^2 V(\tau q_s + (1-\tau)q^*){\rm d}\tau)(q_s - q^*)$
and using the hypothesis on $M$,
\begin{equation}
    |M(\sfQ_h^1(q^*,p^*)-\sfQ^h(q^*,p^*))|_{\ell^{\infty}}^2 \leq \tfrac12 h^2 C_3^2 \int_0^h\int_0^t|q_s-q^*|_{\ell^\infty}^2{\rm d}s {\rm d}t\,.
\end{equation}
Expanding $q_s$ bounds the remaining integral,
\begin{equation}
\label{eqn-qs-q-double-integral}
\begin{aligned}
    \int_0^h\int_0^t|q_s-q^*|_{\ell^\infty}^2 {\rm d}s{\rm d}t
    &\leq 2\int_0^h\int_0^t|p^*s|_{\ell^\infty}^2 {\rm d}s{\rm d}t + 2 \int_0^h\int_0^t\Big|\int_0^s\int_0^\tau \nabla V(q_r){\rm d}r {\rm d}\tau\Big|_{\ell^\infty}^2{\rm d}s{\rm d}t\\
    &\leq \tfrac16h^4|p^*|_{\ell^\infty}^2 + \int_0^h\int_0^ts^2\int_0^s\int_0^\tau|\nabla V(q_r)|_{\ell^\infty}^2{\rm d}r{\rm d}\tau {\rm d}s{\rm d}t\,.
\end{aligned}
\end{equation}
Since $(q^*,p^*) \sim \pi \otimes \mathcal{N}(0,I)$, stationarity gives $q_t \sim \pi$
for all $t$, so by Lemma~\ref{lem-Gaussian-linf},
\begin{equation}
\label{eqn-qs-q-double-integral-expectation}
    \bE \int_0^h\int_0^t|q_s-q^*|_{\ell^\infty}^2{\rm d}s{\rm d}t\leq \tfrac23h^4\log(2d) + \tfrac{1}{60}h^6 \bE_\pi[|\nabla V|_{\ell^\infty}^2]\,.
\end{equation}
Combining these with Lemma~\ref{lem-gradV-expectation} and $h^2\beta \leq 1/20$,
\begin{equation}
\begin{aligned}
    |M(\sfQ_h^1(q^*,p^*)-\sfQ^h(q^*,p^*))|_{2,\ell^\infty}
    &\leq \tfrac1{\sqrt{3}} h^3C_3\sqrt{\log(2d)} + \tfrac{1}{\sqrt{120}}h^4C_3\sqrt{\bE_{\pi}[|\nabla V|_{\ell^\infty}^2]}\\
    &\leq \tfrac1{\sqrt{2}}h^3C_3\sqrt{\log(2d)}\,,
\end{aligned}
\end{equation}
which is~\eqref{eqn-one-step-position-err}. For the momentum component,
\begin{equation}
    \sfP_h^1(q^*,p^*) = p^*-\tfrac12h\nabla V(q^*)-\tfrac12h\nabla V(\sfQ_h^1(q^*,p^*))\,,\quad \sfP^h(q^*,p^*) = p^*-\int_0^h\nabla V(q_t) {\rm d}t\,,
\end{equation}
so
\begin{equation}
\label{eqn-decomp-MPhP-a-b}
\begin{aligned}
    &|M(\sfP_h^1(q^*,p^*)-\sfP^h(q^*,p^*))|_{2,\ell^\infty}^2
   \\
   &\leq \tfrac12 h C_3^2 \underbrace{\bE\int_0^h|q_t-q^*|_{\ell^\infty}^2 {\rm d}t}_{(a)} + \tfrac12\underbrace{\bE\Big|\int_0^h \tilde M(q_t - \sfQ_h^1(q^*,p^*)){\rm d}t\Big|_{\ell^\infty}^2}_{(b)}\,,
\end{aligned}
\end{equation}
where $\tilde M \coloneqq \int_0^1 M \nabla^2 V(\tau q_t + (1-\tau)\sfQ_h^1(q^*,p^*)){\rm d}\tau$.
The bounds used in~\eqref{eqn-qs-q-double-integral}
and~\eqref{eqn-qs-q-double-integral-expectation} give
\begin{equation}
\label{eqn-bound-a}
    (a) \leq \tfrac83h^3\log(2d)+ \tfrac{1}{10} h^5 \bE_\pi[|\nabla V|_{\ell^\infty}^2]\,.
\end{equation}
Using the position differences in~\eqref{eqn-integral-hmc-position},
\begin{equation}
\begin{aligned}
    &q_t - \sfQ_h^1(q^*,p^*)
    = p^*(t-h) - \int_0^t \int_0^s (\nabla V(q_\tau)-\nabla V(q^*)){\rm d}\tau {\rm d}s+\tfrac12(h^2-t^2)\nabla V(q^*)\\
    = & p^*(t-h) - \int_0^t \int_0^s \Big(\int_0^1\nabla^2V(rq_\tau + (1-r)q^*) {\rm d}r \Big) (q_\tau - q^*){\rm d}\tau {\rm d}s+\tfrac12(h^2-t^2)\nabla V(q^*)\,,
\end{aligned}
\end{equation}
so $(b)$ splits as
\begin{equation}
\label{eqn-b-decomp-b1-b2-b3}
\begin{aligned}
    (b) &\leq 3h C_3^2\underbrace{\bE\int_0^h |p^*(t-h)|_{\ell^\infty}^2 {\rm d}t}_{(b_1)} + 3h C_3^2 \underbrace{\bE \int_0^h \Big| \tfrac12 (h^2-t^2)\nabla V(q^*)\Big|_{\ell^\infty}^2 {\rm d}t}_{(b_2)}\\
    & \quad + \tfrac12 h^3 C_4^2 \underbrace{\bE\int_0^h\int_0^t\int_0^s |q_\tau - q^*|_{\ell^\infty}^2{\rm d}\tau {\rm d}s {\rm d}t}_{(b_3)}\,.
\end{aligned}
\end{equation}
Direct computation of $(b_1), (b_2)$ and the approach
of~\eqref{eqn-qs-q-double-integral-expectation} for $(b_3)$ give
\begin{equation}
\label{eqn-bounds-b1-b2}
    (b_1) \leq \tfrac{4}{3}h^3\log(2d)\,,\quad (b_2) \leq \tfrac{2}{15}h^5 \bE_\pi[|\nabla V|_{\ell^\infty}^2]\,,\quad (b_3) \leq \tfrac{2}{15}h^5\log(2d) + \tfrac{1}{420}h^7 \bE_\pi[|\nabla V|_{\ell^\infty}^2]\,.
\end{equation}
Combining these via~\eqref{eqn-decomp-MPhP-a-b} and applying
Lemma~\ref{lem-gradV-expectation} gives~\eqref{eqn-one-step-momentum-err}.
\end{proof}

\subsection{Discretization error in UL} We adapt the framework of \cite{leimkuhler2024contraction}, in particular the
analysis leading to their Proposition~8.3, to the $W_{2,\ell^\infty}$ norm. We use
an HOH chain preserving the invariant distribution exactly, and measure the
discrepancy between the BAOAB and HOH chains.
\begin{proposition}
\label{prop-discretization-error-decomp}
Consider initial states $(x, p)$ and $(y, p')$ with
$(x, p) \sim \pi \otimes \mathcal{N}(0,I)$. Apply one HOH step to $(x,p)$ and one
BAOAB step to $(y,p')$, coupled through the same Gaussian noise $\xi$, and set
$(\Delta_q, \Delta_p) \coloneqq \sfU_{\hoh, h}^\xi(x,p) - \sfU_{\baoab,h}^\xi(y,p')$.
Suppose $\pi$ satisfies Assumption~\ref{assumption-V-log-concave} and
$h \leq \frac{1-\eta}{2\sqrt{\beta}}$. If
$|\int \nabla^2 V {\rm d}\nu_1|_{\ell^\infty} \leq C_5$ and
$|(\int \nabla^2 V {\rm d}\nu_1)(\int \nabla^2 V {\rm d}\nu_2)|_{\ell^\infty} \leq C_6$
for all probability measures $\nu_1, \nu_2$ on $\mathbb{R}^d$, then
\begin{equation}
\label{eqn-position-diff-discretization-error-decomp-baoab}
    |\Delta_q|_{2,\ell^\infty}\leq (1 + \tfrac14h^2(\eta+1)C_5)|x - y|_{2,\ell^\infty} + \tfrac{h}{2}(\eta + 1)|p - p'|_{2,\ell^\infty} + \tfrac{27}{64}C_5h^3\sqrt{\log(2d)}\,,
\end{equation}
\begin{equation}
\label{eqn-momentum-diff-discretization-error-decomp-baoab}
    \Delta_p = (\eta I - \tfrac{h^2(\eta+1)}{4}Q_2)(p-p')+(-\tfrac h 2 \eta Q_1 - \tfrac h 2 Q_2 + \tfrac{h^3(\eta+1)}{8}Q_2 Q_1)(x - y) + \varepsilon_p\,,
\end{equation}
where $|\varepsilon_p|_{2,\ell^\infty} \leq {\tfrac{27}{8}}h^2 C_5\sqrt{\log(2d)}+{\tfrac38}h^4 C_6\sqrt{\log(2d)}$,
and $Q_1, Q_2 \in \mathbb{R}^{d\times d}$ satisfy
$\alpha I \preceq Q_1, Q_2 \preceq \beta I$, $|Q_1|_{\ell^\infty} \leq C_5$, $|Q_2|_{\ell^\infty} \leq C_5$, and $|Q_2Q_1|_{\ell^\infty}\leq C_6$. 
\end{proposition}
\begin{proof}[Proof of Proposition~\ref{prop-discretization-error-decomp}]
We use integral representations of the updates. Let
$\sfU_{\ho,h}^\xi$ be  the composition of a half-step of Hamiltonian dynamics $\sfH = \sfU_{\hmc}^{h/2}$, 
followed by a full Ornstein--Uhlenbeck map $\sfO$ with noise $\xi$. We write
$(\bar x, \bar p) \coloneqq \sfU_{\ho,h}^\xi(x,p)$. With
$(x(t), p(t)) = \sfU_{\hmc}^t(x,p)$ the continuous Hamiltonian
trajectory~\eqref{eqn-continuous-time-Hamiltonian}, we have
$(x(t), \allowbreak p(t))\sim \pi \otimes \mathcal{N}(0,I)$ and
\begin{equation}
    \bar x = x + \tfrac h 2 p - \int_0^{h/2} \nabla V(x(t))(\tfrac h 2 - t){\rm d}t\,, \quad \bar p = \eta\Big(p - \int_0^{h/2}\nabla V(x(t)){\rm d}t\Big) + \sqrt{1-\eta^2}\xi\,.
\end{equation}
 
\noindent\textbf{Position component.} Let $(\bar x(t), \bar p(t)) = \sfU_{\hmc}^t(\bar x, \bar p)$,
so $(\bar x(t), \bar p(t)) \sim \pi \otimes \mathcal{N}(0, I)$. The position
updates are
\begin{equation}
\begin{aligned}
    q_{\mathsf{HOH}} \coloneqq & x + \tfrac{h}{2} p(1+\eta) - \int_0^{h/2}\nabla V(x(t))(\tfrac{h}{2}-t){\rm d}t - \int_0^{h/2}\nabla V(\bar x(t))(\tfrac{h}{2} - t) {\rm d}t\\
    &+\tfrac{h}{2}\Big(\eta\Big(-\tfrac{h}{2}\nabla V(x) - \int_0^{h/2}\nabla^2 V(x(t))p(t)(\tfrac{h}{2}-t){\rm d}t \Big) + \sqrt{1-\eta^2}\xi \Big)\,,
\end{aligned}
\end{equation}
\begin{equation}
    q_{\mathsf{BAOAB}} \coloneqq y + \tfrac{h}{2} p' (1+\eta) - \tfrac{h^2}{4}(1+\eta) \nabla V(y) + \tfrac{h}{2}\sqrt{1-\eta^2}\xi\,,
\end{equation}
and their difference satisfies
\begin{equation}
\begin{aligned}
    &|\Delta_q|_{2,\ell^\infty}\\
    \leq &(1 + \tfrac{h^2}{4}\eta C_5) |x - y|_{2,\ell^\infty} + \tfrac{h}{2}(1+\eta)|p - p'|_{2,\ell^\infty}+ \Big| \int_0^{h/2}(\nabla V(x(t)) - \nabla V(y))(\tfrac{h}{2}-t){\rm d}t\\
    & + \eta \tfrac{h}{2}\int_0^{h/2}\nabla^2 V(x(t)) p(t) (\tfrac{h}{2}-t){\rm d}t + \int_0^{h/2}(\nabla V(\bar x(t)) - \nabla V(y))(\tfrac{h}{2}-t){\rm d}t\Big|_{2,\ell^\infty}\,.
\end{aligned}
\end{equation}
The last term is at most
\begin{equation}
    C_5\int_0^{h/2}(|x(t) - y|_{2,\ell^\infty}+|\bar x(t) - y|_{2,\ell^\infty})(\tfrac{h}{2}-t){\rm d}t + \tfrac{h\eta C_5}{2} \int_0^{h/2}|p(t)|_{2,\ell^\infty}(\tfrac{h}{2}-t){\rm d}t\,.
\end{equation}
Expanding the integral representation of $x(t)$ gives
$|x(t) - y|_{2, \ell^\infty} \leq |x-y|_{2, \ell^\infty} + |tp - \int_0^t \nabla V(x(s))(t-s){\rm d}s|_{2, \ell^\infty}$,
and by Lemmas~\ref{lem-Gaussian-linf} and~\ref{lem-gradV-expectation},
\begin{equation}
    |tp - \int_0^t \nabla V(x(s))(t-s){\rm d}s|_{2, \ell^\infty} \leq h\sqrt{\log (2d)} + \tfrac{h^2}{4}\sqrt{\beta\log(2d)}\,.
\end{equation}
Similarly, with $\bar x = x(h/2)$,
$|\bar x(t) - y|_{2, \ell^\infty} \leq |\bar x(t) - \bar x|_{2, \ell^\infty} + |\bar x - y|_{2, \ell^\infty}$,
where
\begin{equation}
\begin{aligned}
    &|\bar x(t) - \bar x|_{2, \ell^\infty} \leq h \sqrt{\log(2d)} + \tfrac{h^2}{4}\sqrt{\beta \log(2d)}\,,\\
    & |\bar x - y|_{2, \ell^\infty} \leq |x - y|_{2, \ell^\infty} + h\sqrt{\log(2d)}+\tfrac{h^2}{4}\sqrt{\beta \log(2d)}\,.
\end{aligned}
\end{equation}
Combining these and applying $h \leq 1/(2\sqrt{\beta})$
gives~\eqref{eqn-position-diff-discretization-error-decomp-baoab}.
 
\noindent\textbf{Momentum component.} The momentum updates are
\begin{equation}
\begin{aligned}
    p_{\mathsf{HOH}} &\coloneqq \eta\Big(p -\tfrac{h}{2}\nabla V(x) - \int_0^{h/2}\nabla^2 V(x(t))p(t)( \tfrac h 2 - t){\rm d}t \Big) + \sqrt{1-\eta^2}\xi\\
    &\quad - \tfrac{h}{2}\nabla V(\bar x) - \int_0^{h/2}\nabla^2V(\bar x(t)) \bar p(t) ( \tfrac h 2 - t){\rm d}t\,,
\end{aligned}
\end{equation}
\begin{equation}
    p_{\mathsf{BAOAB}} \coloneqq \eta\Big(p' - \tfrac h 2 \nabla V(y) \Big) + \sqrt{1-\eta^2}\xi - \tfrac h 2 \nabla V(\hat y)\,,
\end{equation}
where $\hat y = y + \tfrac h2 (\eta+1)p' - \tfrac{h^2}4(\eta+1)\nabla V(y) + \tfrac h2 \sqrt{1-\eta^2}\xi$,
so
\begin{equation}
\begin{aligned}
    \Delta_p &= \eta(p - p') - \tfrac{h}{2}\eta(\nabla V(x) - \nabla V(y)) - \eta \int_0^{h/2}\nabla^2V(x(t))p(t)( \tfrac h 2 - t){\rm d}t\\
    &\quad - \tfrac h 2 (\nabla V(\bar x) - \nabla V(\hat y)) - \int_0^{h/2}\nabla^2V(\bar x(t)) \bar p(t) ( \tfrac h 2 - t) {\rm d}t\,.
\end{aligned}
\end{equation}
Grouping the principal terms into
\begin{equation}
    \alpha_p \coloneqq \eta (p - p') - \tfrac h 2 \eta (\nabla V(x) - \nabla V(y)) - \tfrac h 2 (\nabla V(\hat x_c) - \nabla V(\hat y))\,,
\end{equation}
with $\hat x_c = x + \tfrac h2(\eta+1)p - \tfrac{h^2}4(\eta+1)\nabla V(x) + \tfrac h2\sqrt{1-\eta^2}\xi$,
\begin{equation}
    \alpha_p = (\eta I - \tfrac{h^2(\eta+1)}{4}Q_2)(p-p')+(-\tfrac h 2 \eta Q_1 - \tfrac h 2 Q_2 + \tfrac{h^3(\eta+1)}{8}Q_2 Q_1)(x - y)\,,
\end{equation}
where $Q_1 = \nabla^2 V([y, x])$, $Q_2 = \nabla^2 V([\hat x_c, \hat y])$, and
$\nabla^2 V([v_1, v_2]) \coloneqq \int_0^1 \nabla^2 V(v_1 + s(v_2 - v_1)) {\rm d}s$. The
remainder $\varepsilon_p \coloneqq \Delta_p - \alpha_p$ is $I_1 + I_2 + I_3 + I_4$,
with
\begin{equation*}
\begin{aligned}
    I_1 &\coloneqq -\eta \int_0^{h/2}\nabla^2 V(x(t))p ( \tfrac h 2 -t){\rm d}t + \eta \tfrac{h^2}{4}\nabla^2 V([\hat x_c, \bar x])p- \eta \int_0^{h/2}\nabla^2 V(\bar x(t)) \tilde p( \tfrac h 2 - t){\rm d}t\\
    & \quad + \tfrac{h^2}{4}\nabla^2 V([\hat x_c, x])\sqrt{1-\eta^2}\xi - \sqrt{1-\eta^2}\int_0^{h/2}\nabla^2 V(\bar x(t)) \xi ( \tfrac h 2 - t){\rm d}t\,,
\end{aligned}
\end{equation*}
where $\tilde p \coloneqq p - \tfrac h2\nabla V(x) - \int_0^{h/2}\nabla^2 V(x(t)) p(t) (\tfrac h2 - t){\rm d}t$;
\begin{equation*}
    I_2 \coloneqq \eta \int_0^{h/2} \nabla^2 V(x(t)) (p(t) - p) (\tfrac h 2 - t){\rm d} t\,,
\end{equation*}
\begin{equation*}
    I_3 \coloneqq \tfrac h 2 \nabla^2 V([\hat x_c, \bar x])(-\tfrac{h^2}{4}(\eta+1)\nabla V(x) + \int_0^{h/2}\nabla V(x(t))( \tfrac h 2 - t){\rm d}t)\,,
\end{equation*}
\begin{equation*}
    I_4 \coloneqq - \int_0^{h/2}\nabla^2 V(\bar x(t))(-t \nabla V(\bar x) - \int_0^t\nabla^2 V(\bar x(s))\bar p(s)(t-s){\rm d}s)( \tfrac h 2 - t){\rm d}t\,.
\end{equation*}
These terms are bounded by 
\begin{equation}
\begin{aligned}
    &|I_1|_{2,\ell^\infty} \leq (\tfrac{21}{8}h^2 C_5 + \tfrac18h^4C_6) \sqrt{\log(2d)}\,,\quad |I_2|_{2,\ell^\infty} \leq ({\tfrac14}h^2 C_5 + \tfrac18h^4 C_6)\sqrt{\log(2d)}\,,\\
    &|I_3|_{2, \ell^\infty} \leq {\tfrac{3}{8}}h^2 C_5\sqrt{\log(2d)}\,,\quad |I_4|_{2,\ell^\infty} \leq (\tfrac18h^2 C_5 + \tfrac18h^4C_6)\sqrt{\log(2d)}\,,
\end{aligned}
\end{equation}
using $h \leq 1/(2\sqrt{\beta})$, and combining them
gives~\eqref{eqn-momentum-diff-discretization-error-decomp-baoab}.
\end{proof}
The decomposition~\eqref{eqn-multi-step-baoab-discretization-error-decomposition-propagation}
requires bounding the discretization error over $l$ steps. Through the
decomposition~\eqref{eqn-weighted-propagators-primitive} of
Section~\ref{sec-l2-bounds}, the relevant quantity is the expected weighted error
$|(\Delta_q^l, \Delta_p^l)|_{2,\ell^\infty_w}$. We bound it through the surrogate
\begin{equation}
\label{hmc-eqn-weighted-linf-norm-surrogate}
    |\Delta_q^l|_{2,\ell^\infty} + \sqrt{a}|\Delta_p^l|_{2,\ell^\infty}\,,
\end{equation}
an upper bound for $|(\Delta_q^l, \Delta_p^l)|_{2,\ell^\infty_w}$ since
\begin{equation}
    |(\Delta_q^l, \Delta_p^l)|_{2,\ell^\infty_w} = \left| (\Delta_q^l + b \Delta_p^l, \sqrt{a-b^2}\Delta_p^l) \right|_{2, \ell^\infty} \leq |\Delta_q^l|_{2,\ell^\infty} + \sqrt{a}|\Delta_p^l|_{2,\ell^\infty}\,.
\end{equation}
 
To bound this surrogate, one could combine the single-step error of
Proposition~\ref{prop-discretization-error-decomp} with the joint-space
contraction of Corollary~\ref{cor-useful-estimates-bound-baoab}. This is too
crude, as it treats the position and momentum errors together. Instead, we use a
Gr\"onwall argument that tracks the two separately, mirroring the $\ell^2$
analysis of~\cite{leimkuhler2024contraction}.

\begin{proposition}[Refined statement of Proposition~\ref{prop-discretization-error-underdamped-main-text}]
\label{prop-l-step-error-gronwall}
For the $(\Delta_q^l, \Delta_p^l)$ process with shared initial state
$(q^*,p^*)\sim \pi \otimes \mathcal{N}(0,I)$, if $\pi$ satisfies
Assumption~\ref{assumption-V-log-concave} and $h \leq \frac{1-\eta}{2\sqrt\beta}$,
then the accumulated discretization error over $l$ steps satisfies
\begin{equation}
\label{eqn-prop-l-step-error-gronwall}
    |(\Delta_q^l, \Delta_p^l)|_{2,\ell^\infty_w} \leq |\Delta_q^l|_{2,\ell^\infty} + \sqrt{a}|\Delta_p^l|_{2,\ell^\infty} \leq \exp\left[(l-1)h\Big(\tfrac32\tfrac{C_5}{\sqrt\beta} + \tfrac{1}{16}\tfrac{C_6}{\beta^{3/2}} +\sqrt\beta\Big)\right] E_l\,,
\end{equation}
where $C_5, C_6$ are the constants in
Proposition~\ref{prop-discretization-error-decomp} and 
\begin{equation}
    E_l = \left( \tfrac{27}{64}C_5 h^3l + {\tfrac{27}{8}}\tfrac{h^2C_5}{\sqrt{\beta}(1-\eta)} + {\tfrac38}\tfrac{h^4C_6}{\sqrt{\beta}(1-\eta)}\right) \sqrt{\log(2d)}\,.
\end{equation}
\end{proposition}
\begin{proof}[Proof of Proposition \ref{prop-l-step-error-gronwall}]
By Proposition~\ref{prop-discretization-error-decomp}, unrolling the position
error recurrence to $\Delta_q^0 = 0$ gives
\begin{equation}
\begin{aligned}
    |\Delta_q^l|_{2,\ell^\infty}
    &\leq |\Delta_q^{l-1}|_{2,\ell^\infty} + h (hC_5|\Delta_q^{l-1}|_{2,\ell^\infty} + |\Delta_p^{l-1}|_{2,\ell^\infty}) + \tfrac{27}{64}C_5h^3\sqrt{\log(2d)}\\
    &\leq \sum_{i=1}^{l-1} h(hC_5|\Delta_q^i|_{2, \ell^\infty} + |\Delta_p^i|_{2,\ell^\infty}) + \tfrac{27}{64}C_5lh^3\sqrt{\log(2d)}\,.
\end{aligned}
\end{equation}
Likewise, unrolling the momentum error recurrence and bounding the geometric sum
of the constant term by $(1-\eta)^{-1}$ yields
\begin{equation}
\begin{aligned}
    |\Delta_p^l|_{2,\ell^\infty}
    &\leq (\eta + \tfrac{h^2}{2}C_5)|\Delta_p^{l-1}|_{2,\ell^\infty} + (h C_5 + \tfrac{h^3}{4}C_6)|\Delta_q^{l-1}|_{2,\ell^\infty} + |\varepsilon_p|_{2,\ell^\infty}\\
    &\leq \sum_{i=1}^{l-1}\eta^{i-1}\tfrac{h^2}{2}C_5|\Delta_p^{l-i}|_{2,\ell^\infty} + \sum_{i=1}^{l-1}\eta^{i-1}(hC_5 + \tfrac{h^3C_6}{4})|\Delta_q^{l-i}|_{2,\ell^\infty} + \tfrac{1}{1-\eta}|\varepsilon_p|_{2,\ell^\infty}\,.
\end{aligned}
\end{equation}
Combining the two, with $h \leq 1/(2\sqrt{\beta})$ and $\eta \leq 1$, we obtain 
\begin{equation}
    |\Delta_q^l|_{2,\ell^\infty} + \sqrt a |\Delta_p^l|_{2,\ell^\infty} \leq \sum_{i=1}^{l-1}h\Big(\tfrac32\tfrac{C_5}{\sqrt\beta}+ \tfrac{1}{16}\tfrac{C_6}{\beta^{3/2}} + \sqrt\beta\Big)(|\Delta_q^{l-i}|_{2,\ell^\infty}+\sqrt a |\Delta_p^{l-i}|_{2,\ell^\infty}) + E_l\,.
\end{equation}
Applying Gr\"onwall's inequality leads to~\eqref{eqn-prop-l-step-error-gronwall}.
\end{proof}
\section{Sampling Bias Bounds}
\label{appendix-sampling-bias-bounds}
Under the framework of Section~\ref{sec-sketch-techniques}, we now combine the
propagator norm estimates with the discretization error to track the sampling
error across steps, and obtain the bias bounds in the limit of infinitely many
steps.
\subsection{HMC}
We prove the iterative $W_{2,\ell^\infty}$ bounds of
Propositions~\ref{prop-sketch-weak-bound-hmc}
and~\ref{prop-sketch-sparse-bound-hmc}, from which the HMC cases of
Theorems~\ref{thm-weak-formal} and~\ref{thm-sparse-formal} follow by iterating to
the stationary limit.
We begin with the weak interaction case.
\begin{proof}[Proof of Proposition \ref{prop-sketch-weak-bound-hmc}]
Following the coupling and one-step split~\eqref{eqn-sketch-one-step} into the
contraction term $(a)$ and discretization term $(b)$, we bound each.
 
\noindent\textbf{Contraction term.} By Proposition~\ref{prop-matrix-polynomial},
\begin{equation}
    (a) = |\tilde T_m^1(\{I - \tfrac{h^2}{2}H_i\}_{i=0}^{m-1})(X_k - Y_k)|_{2,\ell^\infty}\,,
\end{equation}
with $H_i = \int_0^1 \nabla^2 V(\tau x_i + (1-\tau) y_i){\rm d}\tau$ and $x_i, y_i$
the internal positions after $i$ leap-frog steps,
$x_i \coloneqq \sfPi_1\sfU_{\hmc,h}^i(X_k, \xi_k)$,
$y_i \coloneqq \sfPi_1\sfU_{\hmc, h}^i(Y_k, \xi_k)$.
The contractive $\ell^\infty$ bound of Proposition~\ref{prop-linftybounds-hmc-weak}
gives
\begin{equation}
    (a) \leq \left(1 - \frac{\alpha}{400\beta}\right) |X_k - Y_k|_{2,\ell^\infty}\,.
\end{equation}
 
\noindent\textbf{Discretization error.} The telescoping
decomposition~\eqref{eqn-one-step-hmc-discretization-error-decomposition-propagation}
and the component split~\eqref{eqn-fine-discretization-error-hmc-decomp} express
$(b)$ through single integration-step errors. Representing each propagated factor
as a matrix polynomial via Proposition~\ref{prop-matrix-polynomial}, with
$(q^*, p^*) \sim \pi \otimes \mathcal{N}(0,I)$,
\begin{equation}
\begin{aligned}
    (b) & \leq\sum_{i=0}^{m-1} \underbrace{|\tilde T_{m-i-1}^1(\{I - \tfrac{h^2}{2}H_j\}_{j=0}^{m-i-2})(\sfQ_h^1 - \sfQ^h)(q^*, p^*)|_{2,\ell^\infty}}_{(b_1)}\\
    & \quad +\sum_{i=0}^{m-1}\underbrace{|h \tilde U_{m-i-2}^1(\{I -\tfrac{h^2}{2}H_j\}_{j=1}^{m-i-2})(\sfP_h^1-\sfP^h)(q^*, p^*)|_{2,\ell^\infty}}_{(b_2)}\,,
\end{aligned}
\end{equation}
where $H_j = \int_0^1\nabla^2 V(\tau x_j + (1-\tau) y_j) {\rm d} \tau$ and $x_j, y_j$
are the $j$-th internal leap-frog positions started at $\sfU_{\hmc,h}^1(q^*, p^*)$
and $\sfU_{\hmc}^h(q^*,p^*)$, i.e.\ we redefine
$x_j \coloneqq \sfPi_1\sfU_{\hmc,h}^j \sfU_{\hmc,h}^1(q^*, p^*)$,
$y_j \coloneqq \sfPi_1\sfU_{\hmc,h}^j\sfU_{\hmc}^h(q^*,p^*)$.
By Proposition~\ref{prop-linftybounds-hmc-weak}, almost surely
{$|\tilde T_{m-i-1}^1|_{\ell^\infty} \leq 1 + \frac{\alpha}{400\beta} \leq \frac{401}{400}$} and
$|h\tilde U_{m-i-2}^1|_{\ell^\infty} \leq \tfrac{3}{2\sqrt\beta}$.

To bound $(b_1), (b_2)$ we apply Proposition~\ref{prop-discretization-error},
instantiating its general matrix $M$ to the specific polynomials appearing here
and determining the corresponding constants $C_3, C_4$. For $(b_1)$, we take
$M = \tilde T_{m-i-1}^1$, which satisfies $|M|_{\ell^\infty} \leq {\frac{401}{400}}$. Under the
weak interaction assumption the integrated Hessian obeys
$|\int\nabla^2 V{\rm d}\nu|_{\ell^\infty} \leq \beta + \tfrac{\alpha}{50} \leq \tfrac{51}{50}\beta$,
so $|M\int\nabla^2 V{\rm d}\nu|_{\ell^\infty} \leq {\sqrt{2}\beta}$ and we may
take $C_3 = \sqrt{2}\beta$. For $(b_2)$, we take $M = h\tilde U_{m-i-2}^1$
with $|M|_{\ell^\infty} \leq \tfrac{3}{2\sqrt\beta}$, giving
$C_3 = \tfrac{3}{2\sqrt\beta}\cdot\tfrac{51}{50}\beta = \tfrac{153}{100}\sqrt\beta$
and, for a product of two Hessians,
$C_4 = \tfrac{3}{2\sqrt\beta}\cdot(\tfrac{51}{50}\beta)^2 = \tfrac{3}{2}(\tfrac{51}{50})^2\beta^{\frac32}$.
With $h \leq 1/\sqrt{20\beta}$, Proposition~\ref{prop-discretization-error} yields
\begin{equation}
    (b_1) \leq h^3\beta\sqrt{\log(2d)} \leq h^2\sqrt{\beta\log(2d)}\,, \quad (b_2) \leq 4h^2\sqrt{\beta\log(2d)}\,,
\end{equation}
and summing over the $m$ steps,
\begin{equation}
    (b) \leq 5mh^2\sqrt{\beta \log(2d)} \leq 2h\sqrt{\log(2d)}\,,
\end{equation}
using $mh = 1/\sqrt{20\beta}$ in the last step. Thus,
\begin{equation}
    |X_{k+1}-Y_{k+1}|_{2,\ell^\infty} \leq \left( 1 - \frac{\alpha}{400\beta} \right) |X_k - Y_k|_{2,\ell^\infty} + 2h\sqrt{\log(2d)}\,.
\end{equation}
Taking $X_k, Y_k$ optimally coupled so that
$W_{2,\ell^\infty}(\rho_k, \pi) = |X_k - Y_k|_{2,\ell^\infty}$, we obtain
\begin{equation}
    W_{2,\ell^\infty}(\rho_{k+1},\pi) \leq \left( 1 - \frac{\alpha}{400\beta} \right) W_{2,\ell^\infty}(\rho_k, \pi) + 2h\sqrt{\log(2d)}\,.
\end{equation}
\end{proof}
\begin{proof}[Proof of the HMC case of Theorem~\ref{thm-weak-formal}]
Under $C^{(O)} = 1/50$ and $mh = 1/\sqrt{20\beta}$, taking $k \to \infty$ in the
inequality of Proposition~\ref{prop-sketch-weak-bound-hmc} gives
\begin{equation}
    W_{2,\ell^\infty}(\pi_h, \pi) = O\left(\frac{\beta}{\alpha}h\sqrt{\log(2d)}\right)\,.
\end{equation}
\end{proof}
We turn to the sparse interaction case.
\begin{proof}[Proof of Proposition \ref{prop-sketch-sparse-bound-hmc}]
With $X_k, Y_k$ as before, the multi-step split~\eqref{eqn-sketch-multi-step}
bounds $|X_{k+N} - Y_{k+N}|_{2,\ell^\infty}$ by its contraction term $(a)$ and
discretization term $(b)$.
 
\noindent\textbf{Contraction term.} Proposition~\ref{prop-matrix-polynomial}
represents the propagators as matrix polynomials, bounded in $\ell^2$ by
Proposition~\ref{prop-l2bounds-hmc} and in $\ell^\infty$ by
Proposition~\ref{prop-linf-bounds-sparse-hmc-appendix}, which extends
Proposition~\ref{prop-linf-bounds-sparse-hmc}. The $\ell^2$ bound, converted to
$\ell^\infty$ at a cost of $\sqrt d$, gives
\begin{equation}
\label{eqn-part-a-proof-prop-sketch-sparse-bound-hmc}
    (a) \leq (1 - \tfrac{\alpha}{200\beta} )^N\sqrt{d}\, |X_k - Y_k|_{2,\ell^\infty}\,.
\end{equation}
 
\noindent\textbf{Discretization error.} The telescoping
decomposition~\eqref{eqn-multi-step-hmc-discretization-error-decomposition-propagation}
and the component split~\eqref{eqn-fine-discretization-error-hmc-decomp}, with the
matrix-polynomial representation of Proposition~\ref{prop-matrix-polynomial}, give
\begin{equation}
\begin{aligned}
    (b)
    & \leq \sum_{l=0}^{N-1}\sum_{i=0}^{m-1} \underbrace{\Big|A_{(N-l-1)m+(m-i-1)}(\sfQ_h^1-\sfQ^h)(q^*, p^*)\Big|_{\ell^\infty}}_{(b_1)}\\
    & \quad + \sum_{l=0}^{N-1}\sum_{i=0}^{m-1} \underbrace{\Big|B_{(N-l-1)m+(m-i-2)}(\sfP_h^1-\sfP^h)(q^*, p^*)\Big|_{\ell^\infty}}_{(b_2)}\,,
\end{aligned}
\end{equation}
where, writing
$P_{l,i} \coloneqq \prod_{k \in\{m-i-1+jm\}_{j=0}^{N-l-2}}^{\longleftarrow}\tilde T_m^1(\{I - \tfrac{h^2}{2}H_{k+s} \}_{s=0}^{m-1})$
for the product of full-loop propagators over the $N-l-1$ subsequent outer loops,
\begin{equation}
    A_{(N-l-1)m+(m-i-1)} = P_{l,i}\,\tilde T_{m-i-1}^1\,,\qquad B_{(N-l-1)m+(m-i-2)} = P_{l,i}\, h\tilde U_{m-i-2}^1\,.
\end{equation}
The Hessians $H_j = \int_0^1 \nabla^2 V(\tau x_j + (1-\tau) y_j){\rm d} \tau$ are
indexed by the total internal leap-frog step count:
\begin{equation*}
\begin{aligned}
    &0 \leq j \leq m-i-1: \quad x_{j} = \sfPi_1 \sfU_{\hmc,h}^j \sfU_{\hmc,h}^1(q^*, p^*)\,,\; y_j = \sfPi_1 \sfU_{\hmc,h}^j \sfU_{\hmc}^h(q^*, p^*)\,;\\
    & 0 \leq \hat l \leq N-l-1,\ 0 \leq j < m-1:\\
    &\quad x_{\hat l m+j+(m-i-1)} = \sfPhi_{h,\xi_{k+l+\hat l + 1}}^j \sfPhi_{h,\xi_{k+l+\hat l}}^m\cdots \sfPhi_{h,\xi_{k+l+1}}^m\sfPi_1\sfU_{\hmc,h}^{m-i-1}\sfU_{\hmc,h}^1(q^*, p^*)\,, \\
    &\quad y_{\hat l m+j+(m-i-1)} = \sfPhi_{h,\xi_{k+l+\hat l + 1}}^j \sfPhi_{h,\xi_{k+l+\hat l}}^m\cdots \sfPhi_{h,\xi_{k+l+1}}^m\sfPi_1\sfU_{\hmc,h}^{m-i-1}\sfU_{\hmc}^h(q^*, p^*)\,,
\end{aligned}
\end{equation*}
with the convention that for $\hat l = 0$ only $\sfPhi_{h,\xi_{k+l+\hat l + 1}}^j$
is applied. By Propositions~\ref{prop-linf-bounds-sparse-hmc-appendix}
and~\ref{prop-discretization-error}, for $r = r_{(N-l-1)m+(m-i-1)}$,
\begin{equation}
    (b_1) \leq 3h^2 \sqrt{s_{r}\beta\log(2d)}\,,\quad (b_2) \leq 9h^2\sqrt{s_{r}\beta\log(2d)}\,,
\end{equation}
so summing over the $Nm$ steps,
\begin{equation}
\label{eqn-part-b-proof-prop-sketch-sparse-bound-hmc}
    (b) \leq 12h^2\sqrt{\beta\log(2d)}\sum_{i=0}^{Nm-1}\sqrt{s_{r_{i}}}\,.
\end{equation}
Combining~\eqref{eqn-part-a-proof-prop-sketch-sparse-bound-hmc}
and~\eqref{eqn-part-b-proof-prop-sketch-sparse-bound-hmc} with $X_k, Y_k$
optimally coupled completes the proof.

\end{proof}
The following corollary follows by specifying $N$ and estimating $\sum \sqrt{s_{r_i}}$ under the polynomial growth condition.

\begin{proof}[Proof of the HMC case of Theorem~\ref{thm-sparse-formal}]
Take $N = \lceil \frac{200\beta}{\alpha}\log(2\sqrt{d})\rceil$, so that the
contraction factor $(1-\tfrac{\alpha}{200\beta})^N\sqrt{d} \leq \tfrac12$. Under
$s_k \leq C(k+1)^n$, an integral comparison gives
\begin{equation}
\label{eqn-integral-sum-sparsity-hmc}
\begin{aligned}
    &\sum_{i=0}^{Nm-1}\sqrt{s_{r_{i}}}
    \leq \sum_{i=0}^{Nm-1} \sqrt{C}\left(ih\sqrt\beta e + \tfrac{\log\sqrt d}{\log (5/3)} + 2\right)^{\frac n 2}\\
    &\quad \leq \int_0^{Nm}\sqrt{C}\left(yh\sqrt{\beta}e + \tfrac{\log\sqrt{d}}{\log(5/3)} + 2 \right)^{\frac n 2}{\rm d}y \leq \sqrt{C}\frac{\left(Nmh\sqrt{\beta}e + \tfrac{\log\sqrt{d}}{\log(5/3)} + 2 \right)^{\frac n 2+1}}{(n/2+1)h\sqrt{\beta}e}\,.
\end{aligned}
\end{equation}
With this $N$ and $mh \leq 1/\sqrt{20\beta}$,
Proposition~\ref{prop-sketch-sparse-bound-hmc} gives the recursion
\begin{equation}
    W_{2,\ell^\infty}(\rho_{k+N},\pi) \leq \tfrac12 W_{2,\ell^\infty}(\rho_k,\pi) + h \sqrt{\log(2d)}\left(O(\tfrac{\beta}{\alpha}\log(2d))\right)^{\frac n 2 + 1}\,,
\end{equation}
and taking $k\to \infty$ gives
\begin{equation}
    W_{2,\ell^\infty}(\pi_h, \pi) = h\sqrt{\log(2d)}\left( O\left( \frac\beta\alpha\log(2d)\right)\right)^{\frac n 2 + 1}\,.
\end{equation}
\end{proof}
\subsection{UL}
\begin{proof}[Proof of the UL case of Theorem~\ref{thm-weak-formal}]
We follow the coupling framework~\eqref{eqn-sketch-underdamped} with the weighted
matrix representation of Section~\ref{sec-l2-bounds}. Denote the initial distribution of $X_0$ by $\rho_0$. The BAOAB chain
$(X_k, P_k)$ from $(X_0, P_0) \sim \rho_0 \otimes \mathcal{N}(0,I)$  and the HOH
chain $(Y_k, P_k')$ from $(Y_0, P_0') \sim \pi \otimes \mathcal{N}(0,I)$ are
coupled through the same noise $\{\xi_i\}$. Introducing the auxiliary process
$(\bar Y_k, \bar P'_k) = \sfU_{\baoab,h}^{\xi_{k-1}}\cdots \sfU_{\baoab,h}^{\xi_0}(Y_0, P'_0)$,
which starts at $(Y_0, P_0')$ but runs BAOAB with the same noise, the
split~\eqref{eqn-sketch-underdamped} becomes
\begin{equation}
\label{eqn-prop-baoab-weak-bias-iteration-decomposition}
    |X_k - Y_k|_{2,\ell^\infty} \leq \underbrace{|X_k - \bar Y_k|_{2,\ell^\infty}}_{(a)} + \underbrace{|\bar Y_k - Y_k|_{2,\ell^\infty}}_{(b)}\,,
\end{equation}
the contraction between two BAOAB chains $(a)$ and the BAOAB-vs-HOH discretization
error $(b)$.
 
\noindent\textbf{Contraction term.} Using the
decomposition $\sfU_{\baoab,h}^{\xi_{k-1}}\cdots \sfU_{\baoab,h}^{\xi_0} = \sfU_{\ab,h} \sfU_{\abao,h}^{\xi_{k-1}} \cdots \allowbreak\sfU_{\abao,h}^{\xi_{1}} \sfU_{\bao,h}^{\xi_0}$
of Section~\ref{sec-sketch-underdamped} and the weighted position
propagators~\eqref{eqn-weighted-position-propagators},
\begin{equation}
\label{eqn-weighted-matrix-representation-weak-interaction}
    X_k - \bar Y_k = \begin{bmatrix} I & 0 \end{bmatrix} W^{-1} M_{\ab}^w (H_{k}) M_{\abao}^w(H_{k-1}) \cdots M_{\abao}^w(H_1)M_{\bao}^w(H_0) W\Delta Z_0\,,
\end{equation}
with $\Delta Z_0 = (X_0-Y_0, P_0 - P'_0)$ and the matrices defined through the
intermediate iterates
\begin{equation*}
\begin{aligned}
    &z_1\coloneqq \sfU_{\bao,h}^{\xi_0}(X_0, P_0)\,, \quad z_1'\coloneqq \sfU_{\bao,h}^{\xi_0}(Y_0, P_0')\,, \quad z_{k+1} \coloneqq \sfU_{\ab,h}z_{k}\,, \quad z_{k+1}' \coloneqq \sfU_{\ab,h} z_{k}'\,,\\
    &z_i \coloneqq \sfU_{\abao,h}^{\xi_{i-1}} \cdots \sfU_{\abao,h}^{\xi_1}z_1\,, \quad z_i' \coloneqq \sfU_{\abao,h}^{\xi_{i-1}} \cdots \sfU_{\abao,h}^{\xi_1}z_1'\,, \quad 2\leq i \leq k\,,
\end{aligned}
\end{equation*}
satisfying
\begin{equation*}
\begin{aligned}
    & W(z_1 - z_1') = M_{\bao}^w(H_0)W\Delta Z_0\,, \quad W(z_{k+1} - z_{k+1}') = M_{\ab}^w(H_k)W(z_{k}-z_{k}')\,,\\
    & W(z_i -z_i') = M_{\abao}^{w}(H_{i-1}) W(z_{i-1}- z_{i-1}')\,, \quad 2\leq i \leq k\,.
\end{aligned}
\end{equation*}
By Proposition~\ref{prop-propagators-weak-interactions-baoab} and
$|\begin{bmatrix} I & 0 \end{bmatrix}W^{-1}|_{\ell^\infty} = 1 + \frac{b}{\sqrt{a-b^2}}$,
\begin{equation}
    (a) \leq 6\Big(1 + \tfrac{b}{\sqrt{a-b^2}}\Big) ( 1 - \tfrac12c(h))^{k-1}|\Delta Z_0|_{2,\ell^\infty_w}\,.
\end{equation}
 
\noindent\textbf{Discretization error.} Write
$(\Delta_q^{l}, \Delta_p^{l}) = (\sfU_{\baoab,h}^{l} - \sfU_{\hoh,h}^{l})(q^*, p^*)$.
The matrix representation of
$\sfPi_1\sfU_{\baoab,h}^{k-(i+1)\tilde l}(\sfU_{\baoab,h}^{\tilde l} - \sfU_{\hoh,h}^{\tilde l})(q^*, p^*)$
parallels~\eqref{eqn-weighted-matrix-representation-weak-interaction}, with
$\Delta Z_0 = (\Delta_q^{\tilde l}, \Delta_p^{\tilde l})$ and Hessians along the
intermediate iterates. Each term of the telescoping
decomposition~\eqref{eqn-multi-step-baoab-discretization-error-decomposition-propagation} 
is thus bounded by
\begin{equation}
\label{eqn-baoab-discretization-error-individual-term-with-propagation}
\begin{aligned}
    &|\sfPi_1\sfU_{\baoab,h}^{k-(i+1)\tilde l}(\sfU_{\baoab,h}^{\tilde l} - \sfU_{\hoh,h}^{\tilde l}) (q^*, p^*)|_{2,\ell^\infty}\\
    &\qquad \leq 6\Big(1 + \tfrac{b}{\sqrt{a-b^2}}\Big) ( 1 - \tfrac12c(h))^{k-(i+1)\tilde l - 1}|(\Delta_{q}^{\tilde l}, \Delta_p^{\tilde l})|_{2,\ell^\infty_w}\,,
\end{aligned}
\end{equation}
\begin{equation}
\label{eqn-baoab-last-step-discretization-error}
    \left|\sfPi_1\left(\sfU_{\baoab,h}^{k-\lfloor k/\tilde l \rfloor \tilde l} - \sfU_{\hoh,h}^{k - \lfloor k/\tilde l \rfloor \tilde l}\right)(q^*, p^*)\right|_{2,\ell^\infty} \leq \Big(1 + \tfrac{b}{\sqrt{a-b^2}}\Big)|(\Delta_q^{k - \lfloor k/\tilde l\rfloor \tilde l}, \Delta_p^{k - \lfloor k/\tilde l\rfloor \tilde l})|_{2,\ell^\infty_w}\,.
\end{equation}
In~\eqref{eqn-baoab-discretization-error-individual-term-with-propagation}, when
$i = \lfloor k/\tilde l \rfloor -1$ there may be no internal ABAO steps and the
exponent can be negative, so we instead bound the operator by
$6(1+\tfrac{b}{\sqrt{a-b^2}})$, accounting for the $\sfB\sfA\sfO$ and $\sfA\sfB$
prefactors. Proposition~\ref{prop-l-step-error-gronwall}  gives for any $l \leq \tilde l$
\begin{equation}
    |(\Delta_q^l, \Delta_p^l)|_{2,\ell^\infty_w} \leq \exp\left[(\tilde l-1)h\Big(\tfrac32\tfrac{C_5}{\sqrt\beta} + \tfrac{1}{16}\tfrac{C_6}{\beta^{3/2}} + \sqrt\beta\Big)\right] E_{\tilde l}\,.
\end{equation}
Substituting these into~\eqref{eqn-prop-baoab-weak-bias-iteration-decomposition},
with $X_0, Y_0$ optimally coupled so
$W_{2,\ell^\infty}(\rho_0, \pi) = |X_0-Y_0|_{2,\ell^\infty}$ and $P_0' = P_0$,
\begin{equation}
\label{eqn-term-by-term-k-step-error-UL-weak}
\begin{aligned}
    &|X_k - Y_k|_{2,\ell^\infty}
    \leq 6 \Big(1 + \tfrac{b}{\sqrt{a-b^2}}\Big) \left( 1 - \tfrac12c(h)\right)^{k-1}W_{2,\ell^\infty}(\rho_0,\pi) + \Big(1 + \tfrac{b}{\sqrt{a-b^2}}\Big)\\
    &\quad \cdot\left( 7 + \sum_{i=0}^{\lfloor k / \tilde l\rfloor-2} 6(1-\tfrac12c(h))^{k-(i+1)\tilde l-1}\right) \exp\left[(\tilde l-1)h\Big(\tfrac32\tfrac{C_5}{\sqrt\beta} + \tfrac{1}{16}\tfrac{C_6}{\beta^{3/2}}+\sqrt\beta\Big)\right] E_{\tilde l}\,.
\end{aligned}
\end{equation}
Under the weak interaction assumption $|H_i^{(O)}|_{\ell^\infty} \leq \alpha/20$,
the constants of Proposition~\ref{prop-discretization-error-decomp} satisfy
$C_5 \leq \tfrac{21}{20}\beta$ and $C_6 \leq (\tfrac{21}{20})^2\beta^2$, so that
$C_5/\beta$ and $C_6/\beta^2$ are $O(1)$. Since
$(\tilde l-1)h\sqrt\beta = O(1)$ by the choice
$\tilde l = \lceil 1/(2\sqrt\beta h)\rceil$, the exponential is $O(1)$, and
substituting $E_{\tilde l}$ gives
\begin{equation*}
    \exp\left[(\tilde l-1)h\Big(\tfrac32\tfrac{C_5}{\sqrt\beta} + \tfrac{1}{16}\tfrac{C_6}{\beta^{3/2}} + \sqrt\beta\Big)\right] E_{\tilde l} = O\Big(\frac{h^2\sqrt\beta}{1-\eta}\sqrt{\log(2d)}\Big)\,.
\end{equation*}
 
As $k \to \infty$, the first term vanishes. For the geometric sum, with
$c(h) = \frac{\alpha h^2}{4(1-\eta)}$ and $\tilde l \geq 1/(2\sqrt\beta h)$, the
inequality $1-(1-x)^{\tilde l} \geq \tilde l x - \tfrac12(\tilde l x)^2$ for $x>0$ gives
\begin{equation*}
    1 - (1- \tfrac12c(h))^{\tilde l} \geq \frac{63}{1024}\frac{\alpha h}{\sqrt\beta (1-\eta)}\,,
\end{equation*}
so that $\sum_{i=0}^{\lfloor k/\tilde l\rfloor - 2} 6(1-\tfrac12 c(h))^{k-(i+1)\tilde l - 1} \leq 6/(1-(1-\tfrac12 c(h))^{\tilde l}) = O(\tfrac{\sqrt\beta(1-\eta)}{\alpha h})$.
Since $1 + \tfrac{b}{\sqrt{a-b^2}} = O(1)$, \eqref{eqn-term-by-term-k-step-error-UL-weak} becomes
\begin{equation}
    W_{2,\ell^\infty}(\pi_h, \pi) \leq O(1)\cdot O\Big(\frac{\sqrt\beta(1-\eta)}{\alpha h}\Big)\cdot O\Big(\frac{h^2\sqrt\beta}{1-\eta}\sqrt{\log(2d)}\Big) = O\Big(\frac{\beta}{\alpha}h\sqrt{\log(2d)}\Big)\,.
\end{equation}
\end{proof}
\begin{proof}[Proof of the UL case of Theorem~\ref{thm-sparse-formal}]
As in the HMC sparse case, we induce an $\ell^\infty$ contraction from the
$\ell^2$ contraction over multiple steps. We partition the $k$ steps into blocks
of $B$ steps, with $B$ specified below. We use the three
processes $(X_k, P_k), (Y_k, P'_k), (\bar Y_k, \bar P'_k)$ of the proof of the UL
case of Theorem~\ref{thm-weak-formal} with $P_0 = P_0' = \bar P'_0$, and the
decomposition
\begin{equation}
\label{eqn-prop-baoab-sparse-bias-iteration-decomposition}
    |X_k - Y_k|_{2,\ell^\infty} \leq \underbrace{|X_k - \bar Y_k|_{2,\ell^\infty}}_{(a)} + \underbrace{|\bar Y_k - Y_k|_{2,\ell^\infty}}_{(b)}\,.
\end{equation}
 
\noindent\textbf{Contraction term.} Corollary~\ref{cor-useful-estimates-bound-baoab} gives
\begin{equation}
\label{eqn-part-a-baoab-sparse}
    (a) \leq 14\sqrt{d}(1-c(h))^{\frac{k-1}{2}}|X_0 - Y_0|_{2,\ell^\infty}\,,
\end{equation}
which vanishes as $k \to \infty$.
 
\noindent\textbf{Discretization error.} The telescoping
decomposition~\eqref{eqn-multi-step-baoab-discretization-error-decomposition-propagation}
accumulates a local error every $\tilde l = \lceil \frac{1}{2\sqrt\beta h} \rceil$
steps, each propagated through the later steps. Write
$(\Delta_q^l, \Delta_p^l) = (\sfU_{\baoab,h}^{l} - \sfU_{\hoh,h}^{l})(q^*, p^*)$
for an $l$-step error, with $(q^*, p^*) \sim \pi \otimes \mathcal{N}(0,I)$, and recall that
$\sfU_{\baoab,h}^l, \sfU_{\hoh,h}^l$ are the $l$-fold compositions omitting the
coupled noise. For the sparse case the constants of
Proposition~\ref{prop-discretization-error-decomp} are $C_5 = \beta\sqrt{s_1}$ and
$C_6 = \beta^2\sqrt{s_2}$, so Proposition~\ref{prop-l-step-error-gronwall} bounds
the $\tilde l$-step error,
\begin{equation}
    |(\Delta_q^{\tilde l}, \Delta_p^{\tilde l})|_{2,\ell^\infty_w} \leq \exp\left[(\tilde l-1)h\Big(\tfrac32\tfrac{C_5}{\sqrt\beta} + \tfrac{1}{16}\tfrac{C_6}{\beta^{3/2}} + \sqrt\beta\Big)\right] E_{\tilde l} = O\Big(\frac{h^2\sqrt\beta}{1-\eta}\sqrt{\log(2d)}\Big)\,.
\end{equation}
Setting the block length $B \coloneqq K_1 \tilde l$ with $K_1$ specified later and grouping the trailing
steps in~\eqref{eqn-multi-step-baoab-discretization-error-decomposition-propagation}
into full blocks of $B$ steps followed by residual steps,
\begin{equation}
\label{eqn-multi-step-baoab-discretization-error-block-decomposition-propagation}
\begin{aligned}
    (b) & \leq \sum_{i=0}^{\lceil\frac{\lfloor k/ \tilde l \rfloor}{K_1}\rceil}\sum_{k_1=0}^{K_1-1}\left| \sfPi_1 \sfU_{\baoab,h}^{j_{k_1}}\sfU_{\baoab,h}^{iB}(\sfU_{\baoab,h}^{\tilde l} - \sfU_{\hoh,h}^{\tilde l})(q^*, p^*)\right|_{2,\ell^\infty}\\
    & \quad + \left| \sfPi_1\left(\sfU_{\baoab,h}^{k - \lfloor k/\tilde l \rfloor \tilde l} - \sfU_{\hoh,h}^{k - \lfloor k / \tilde l \rfloor \tilde l} \right)(q^*, p^*)\right|_{2,\ell^\infty}\,,
\end{aligned}
\end{equation}
where $i$ counts the full $B$-step blocks, $k_1 \in \{0, \dots, K_1-1\}$ indexes
the residual $\tilde l$-step groups, and
$j_{k_1} \coloneqq k - \lfloor k /\tilde l\rfloor \tilde l + k_1\tilde l$ is the
corresponding number of trailing steps.

Each summand in the first sum is one $\tilde l$-step error propagated through $iB$
full-block steps and $j_{k_1}$ trailing steps. For two initial states
$(x,p), (y, p')$, Proposition~\ref{prop-matrix-polynomial}
and~\eqref{eqn-weight-matrix-polynomial-baoab} give the weighted matrix-polynomial
form of the propagated position difference over the residual steps,
\begin{equation}
\label{eqn-baoab-individual-discretization-error-matrix-polynomial}
    \sfPi_1\sfU_{\baoab,h}^{j_{k_1}} (x,p) - \sfPi_1\sfU_{\baoab,h}^{j_{k_1}} (y,p')
    = \begin{bmatrix} \tilde T_{j_{k_1}}^\eta & - \tfrac{b}{\sqrt{a-b^2}} \tilde T_{j_{k_1}}^\eta + \tfrac{1}{\sqrt{a-b^2}}h\tilde U_{j_{k_1}-1}^\eta \end{bmatrix} W \begin{bmatrix} x- y \\ p - p' \end{bmatrix}\,,
\end{equation}
the polynomials evaluated at matrices of the form $I - \tfrac{h^2}{2}H$ with $H$ the integrated
Hessians along the intermediate iterates of the two trajectories. Applying this
with $(x,p) = \sfU_{\baoab,h}^{iB}\sfU_{\baoab,h}^{\tilde l} (q^*, p^*)$ and
$(y, p') = \sfU_{\baoab,h}^{iB}\sfU_{\hoh,h}^{\tilde l}(q^*, p^*)$, the $iB$
full-block steps contract the $\tilde l$-step error by
Corollary~\ref{cor-useful-estimates-bound-baoab},
\begin{equation}
\label{eqn-baoab-contraction-in-discretization-error}
    |(x-y, p-p')|_{2,\ell^\infty_w} \leq 7\sqrt{2d}(1-c(h))^{\frac{iB-1}{2}} |(\Delta_q^{\tilde l}, \Delta_p^{\tilde l})|_{2,\ell^\infty_w}\,,
\end{equation}
while Proposition~\ref{prop-linf-bounds-sparse-baoab} bounds the
$\ell^\infty$ norm of the matrix-polynomial row
in~\eqref{eqn-baoab-individual-discretization-error-matrix-polynomial}. Together
they give
\begin{equation}
\begin{aligned}
    &\left| \sfPi_1 \sfU_{\baoab,h}^{j_{k_1}}\sfU_{\baoab,h}^{iB}(\sfU_{\baoab,h}^{\tilde l} - \sfU_{\hoh,h}^{\tilde l})(q^*, p^*)\right|_{2,\ell^\infty}\\
    &\qquad \leq 7\sqrt{2d}\Big(14\sqrt{2}+1+\tfrac{2}{\sqrt{3}}\Big)\sqrt{s_{r_{j_{k_1}}}}\, (1-c(h))^{\frac{iB-1}{2}}|(\Delta_q^{\tilde l}, \Delta_p^{\tilde l})|_{2,\ell^\infty_w}\,.
\end{aligned}
\end{equation}
{
For $i=0$ where there are no intermediate full blocks we instead employ
\begin{equation}
    \left| \sfPi_1 \sfU_{\baoab,h}^{j_{k_1}}(\sfU_{\baoab,h}^{\tilde l} - \sfU_{\hoh,h}^{\tilde l})(q^*, p^*)\right|_{2,\ell^\infty} \leq \left(14\sqrt{2}+1+\tfrac{2}{\sqrt{3}}\right)\sqrt{s_{r_{j_{k_1}}}}|(\Delta_q^{\tilde l}, \Delta_p^{\tilde l})|_{2,\ell^\infty_w}\,.
\end{equation}}
We now fix $K_1 = \lceil \frac{16\sqrt\beta(1-\eta)\log(14\sqrt{2d})}{h\alpha}+1\rceil$,
so that the per-block factor contracts,
$7\sqrt{2d}(1-c(h))^{\frac{B-1}{2}} \leq \tfrac12$, making the contraction over $i$
blocks at most $2^{-i}$. Substituting the $\tilde l$-step error bound and these
per-term estimates
into~\eqref{eqn-multi-step-baoab-discretization-error-block-decomposition-propagation},
together with the last-term
bound~\eqref{eqn-baoab-last-step-discretization-error}, and summing the geometric
factor $2^{-i}$ over $i$, we obtain
\begin{equation}
\label{eqn-part-b-baoab-sparse}
    (b) \leq \left(1 + 2\left( 14\sqrt{2}+1+\tfrac{2}{\sqrt{3}}\right)\sum_{k_1=0}^{K_1-1} \sqrt{s_{r_{j_{k_1}}}} \right)O\Big(\frac{h^2\sqrt\beta}{1-\eta}\sqrt{\log(2d)}\Big)\,.
\end{equation}
By an integral comparison as in~\eqref{eqn-integral-sum-sparsity-hmc}, with
$s_{r_i}$ increasing in $i$,
\begin{equation}
\begin{aligned}
    \sum_{k_1=0}^{K_1-1} \sqrt{s_{r_{j_{k_1}}}}
    \leq \sum_{k_1=1}^{K_1}\sqrt{s_{r_{k_1\tilde l}}}
    &\leq \sqrt{C}\frac{\left(\tfrac{\tilde le^2\beta (K_1+1)}{1-\eta}h^2+\log\sqrt{d}+2\right)^{\frac n 2 + 1}}{(\frac n 2 + 1) \frac{\tilde l e^2\beta}{1-\eta}h^2}\\
    &= \frac{1-\eta}{\sqrt\beta h}\left(O(\tfrac{\beta}{\alpha}\log(2d))\right)^{\frac n 2 +1}\,.
\end{aligned}
\end{equation}
Combining with~\eqref{eqn-part-a-baoab-sparse}
and~\eqref{eqn-part-b-baoab-sparse}
in~\eqref{eqn-prop-baoab-sparse-bias-iteration-decomposition} and taking
$k \to \infty$,
\begin{equation}
    W_{2,\ell^\infty}(\pi_h, \pi) = h \sqrt{\log(2d)} \left(O(\tfrac{\beta}{\alpha}\log(2d))\right)^{\frac n 2 +1}\,.
\end{equation}
\end{proof}
\end{document}